\documentclass[%
10pt,tightenlines,
twocolumn,
superscriptaddress,
 amsmath,amssymb,
 aps,
pra,
]{revtex4-2}

\usepackage[margin=0.8in]{geometry}

\usepackage{graphicx}
\usepackage{dcolumn}
\usepackage{bm}
\usepackage{color,xcolor}
\usepackage[colorlinks]{hyperref}
\usepackage{CJKutf8}

\usepackage{amsthm}
\usepackage{amsmath}

\newtheorem{theorem}{Theorem}[section]
\newtheorem{lemma}[theorem]{Lemma}

\usepackage{bbm}
\newcommand{\id}{\mathbbm{1}}
\newcommand{\Htar}{H_{\mathrm{tar}}}
\newcommand{\Henc}{H_{\mathrm{enc}}}
\newcommand{\Hencfirst}{H_{\mathrm{enc}}^{(1)}}
\newcommand{\Hencsecond}{H_{\mathrm{enc}}^{(2)}}
\newcommand{\Hpen}{H_{\mathrm{pen}}}
\newcommand{\Hsim}{H_{\mathrm{sim}}}
\newcommand{\Htot}{H_{\mathrm{tot}}}

\newcommand{\Heff}{H_{\mathrm{eff}}}
\newcommand{\Penc}{P_{\mathrm{enc}}}
\newcommand{\Senc}{\mathcal{S}_{\mathrm{enc}}}
\newcommand{\Szero}{\mathcal{S}_{0}}
\newcommand{\Pzero}{P_{0}}

\newcommand{\ket}[1]{\left| #1 \right\rangle}
\newcommand{\bra}[1]{\left\langle #1 \right|}
\newcommand{\sgn}{\operatorname{sgn}}
\usepackage{makecell}
\usepackage{enumitem}
\usepackage{cancel}
\usepackage{amsthm}
\newtheorem{corollary}[theorem]{Corollary}

\begin{document}
\begin{CJK}{UTF8}{gbsn}

\title{Robust analog quantum simulators by quantum error-detecting codes}

\author{Yingkang Cao} 
\thanks{These authors contributed equally.}
\affiliation{Joint Center for Quantum Information and Computer Science, University of Maryland, College Park, Maryland 20742, USA}
\affiliation{Department of Computer Science, University of Maryland, College Park, USA}

\author{Suying Liu} 
\thanks{These authors contributed equally.}
\affiliation{Joint Center for Quantum Information and Computer Science, University of Maryland, College Park, Maryland 20742, USA}
\affiliation{Department of Computer Science, University of Maryland, College Park, USA}

\author{Haowei Deng} 
\affiliation{Joint Center for Quantum Information and Computer Science, University of Maryland, College Park, Maryland 20742, USA}
\affiliation{Department of Computer Science, University of Maryland, College Park, USA}

\author{Zihan Xia} 
\affiliation{Center for Quantum Information Science \& Technology, University of Southern California, Los Angeles, CA 90089, USA}
\affiliation{Department of Electrical \& Computer Engineering, University of Southern California, Los Angeles, CA 90089, USA}

\author{Xiaodi Wu} 
\affiliation{Joint Center for Quantum Information and Computer Science, University of Maryland, College Park, Maryland 20742, USA}
\affiliation{Department of Computer Science, University of Maryland, College Park, USA}

\author{Yu-Xin Wang (王语馨)}
\email{yxwang.physics@outlook.com}
\affiliation{Joint Center for Quantum Information and Computer Science, University of Maryland, College Park, Maryland 20742, USA}

\date{\today}

\begin{abstract}
Achieving noise resilience is an outstanding challenge in Hamiltonian-based quantum computation. To this end, energy-gap protection provides a promising approach, where the desired quantum dynamics are encoded into the ground space of a penalty Hamiltonian that suppresses unwanted noise processes. However, existing approaches either explicitly require high-weight penalty terms that are not directly accessible in current hardware, or utilize non-commuting $2$-local Hamiltonians, which typically leads to an exponentially small energy gap. In this work, we provide a general recipe for designing error-resilient Hamiltonian simulations, making use of an excited encoding subspace stabilized by solely $2$-local commuting Hamiltonians. Our results thus overcome a no-go theorem previously derived for ground-space encoding that prevents noise suppression schemes with such Hamiltonians. Importantly, our method is scalable as it only requires penalty terms that scale polynomially with system size. To illustrate the utility of our approach, we further apply this method to a variety of $1$- and $2$-dimensional many-body spin models, potentially extending the duration of high-fidelity simulation by orders of magnitude in current hardware.

\end{abstract}

\maketitle
\end{CJK}

\section{\label{sec:intro}Introduction}

Hamiltonian simulation is an essential application of quantum computers due to its broad relevance to addressing problems in quantum many-body physics~\cite{bernien2017probing-many-body-51-atom,browaeys2020many-body-phys-individually-Rydberg,bluvstein2021controlling-many-body-Rydberg}, quantum chemistry~\cite{arguello2019analogue-quantum-chemistry-sim,clinton2024towards-near-term-qsim-material}, critical phenomena~\cite{keesling2019quantum-Kibble-Zurek,andersen2024thermalization}, and optimization~\cite{Farhi_2001,Crosson_2021,ebadi2022quantum-optimization-mis,leng2023QHD}.
Analog quantum simulation, which directly emulates the target Hamiltonian by a controllable simulator Hamiltonian, stands out to be a promising approach in the noisy intermediate-scale quantum (NISQ) era~\cite{cirac2012goals-and-opportunities-qsim,gross2017qsim-ultracold-atoms-optical-lattice,cubitt2018universal,zhou2021strongly,altman2021quantum,daley2022practical,trivedi2024quantum-adv-stab-to-error-analog-qsim}. However, one of the main challenges in realizing these applications is the extreme sensitivity of quantum information to perturbations, making it fragile against hardware noise. This fragility arises because isolating a quantum computer from external influences and controlling it to perform desired computations are inherently conflicting tasks, making noise inevitable.
In the context of digital quantum computing, quantum error correction (QEC) schemes have been established to achieve fault tolerance~\cite{gottesman2010introduction,kitaev2003fault,dennis2002topological,fowler2009high,takita2016demonstration,marques2022logical,krinner2022realizing,zhao2022realization,harvard-qec-Bluvstein_2023,quantinuum-qec1-da2024demonstration,quantinuum-qec2-ryan2024high,google-qec-2023,eth-qec-Self_2024}. However, those techniques suffer from two restrictions that make them incompatible with current analog quantum hardware: (1) high-weight quantum operations are required in a typical QEC scheme, while analog simulators naturally support 2-local operations; (2) standard QEC requires an active measurement-and-correction procedure, which is noisy and time-consuming on the existing hardware. 

Alternatively, energy-gap protection provides a powerful and versatile approach to suppressing errors in analog simulators~\cite{bacon-shor-code,JFS06-PRA,kevin-young2013error,Pudenz_2014,Marvian-Lidar-2014-PRL-2local-is-too-local,bookatz2015error,Matsuura_2016,Vinci_2016,marvian2017error,marvian2017-PRL-error-supp-subsystem-codes,marvian2019-arxiv-robust,Pearson_2019,Xia_2024}.
In those schemes, error resilience is achieved by a large penalty Hamiltonian that pins the system into its ground subspace, which forms a quantum error-detecting code (QEDC) by design.
The type of noise operators that the scheme protects from is determined by the code properties, while the level of noise suppression depends on the excitation gap of the penalty Hamiltonian.
The desired computation is then realized through an encoding of the target Hamiltonian into this code space. However, existing methods for implementing energy-gap-protected universal Hamiltonians either explicitly require high-weight stabilizers as the penalty terms~\cite{JFS06-PRA,kevin-young2013error,bookatz2015error,marvian2017error}, or involve an exponentially large penalty Hamiltonian~\cite{marvian2019-arxiv-robust}.
This raises the question of whether one can devise energy-gap-protected universal Hamiltonian simulations using only $2$-local physical interactions and a penalty Hamiltonian whose interaction strength scales polynomially with system size.
Towards this end, Ref.~\cite{Marvian-Lidar-2014-PRL-2local-is-too-local} proved an important no-go theorem that precludes any error-resilient constructions using the ground space of commuting $2$-local Hamiltonians.

In this work, we propose a general recipe for designing universal Hamiltonian simulations that are robust against $1$-local errors.
Our approach only utilizes $2$-local commuting penalty Hamiltonians that scale at most polynomially in system size.
Intriguingly, our result circumvents the no-go theorem in~\cite{Marvian-Lidar-2014-PRL-2local-is-too-local}, making use of an excited subspace of the penalty Hamiltonian.
More specifically, we identify a $2$-local commuting Hamiltonian that stabilizes a $[[4,2,2]]$ QEDC as an energy eigenspace.
We further generalize it to a $[[4n,2n,2]]$ QEDC via a tensor product of individual code blocks, which is naturally protected by the sum over the penalty terms on each block.
This code directly enables 1-local, as well as inner-block 2-local logical operations.
Moreover, using perturbative gadgets, we can also implement cross-block 2-local logical operations. Importantly, all of these encodings only involve 2-local physical terms.

We present rigorous bounds on error suppression properties of our approach, as well as extensive numerical evidence supporting our analytical constructions. 
To demonstrate the power of our scheme, we provide fully 2-local, error-resilient implementations of $1$D XY, $1$D and $2$D transverse-field Ising (TFI) model, and $2$D compass model.
With minimal requirements and universal applicability, our scheme opens a new pathway towards hardware-efficient robust analog Hamiltonian simulation at scale.

The rest of this paper is structured as follows. Sec.~\ref{sec:general framework} introduces the general framework of our scheme including the error suppression conditions.
Sec.~\ref{sec:code} introduces a $[[4,2,2]]$ code with a $2$-local penalty Hamiltonian; Sec.~\ref{sec:code-block composition} then discusses the construction of $2$-local commuting penalty Hamiltonians and inner-block logical operations for a $[[4n,2n,2]]$ code.
In Sec.~\ref{sec:cross block}, we describe how perturbative gadgets enable $2$-local cross-block interactions. 
Then, in Sec.~\ref{sec:application}, we utilize all aforementioned building blocks to achieve robust Hamiltonian simulations of various spin models.
We conclude in Sec.~\ref{sec:discussion} with a discussion on physical implementation, as well as open questions for further investigation.

\section{\label{sec:general framework} General encoding framework }
\label{sec:esc.gen}
We first discuss the general framework for designing the encodings that enable robust analog quantum simulation (see Fig.~\ref{fig:general_framework}).
The overarching goal is to simulate a target Hamiltonian $\Htar$ via a machine-native simulator Hamiltonian $\Hsim$ with built-in noise resilience. 
We specifically refer to $2$-local interactions as machine-native terms, as these are the interactions that are naturally supported by most current quantum hardware \cite{oliver-2024probing-2D-hubbard,immanuel-2024emergence,bernien2017probing-many-body-51-atom,monroe-2023continuous}.
The desired computation is achieved by encoding the target quantum system into a subspace $\Senc$ of a larger Hilbert space, such that the restriction of an encoded Hamiltonian $\Henc$ (acting on the larger space) to $\Senc$ recovers the target Hamiltonian.
More precisely, denoting the projector onto $\Senc$ as $\Penc$, we require the following \textit{encoding condition}~\footnote{Here we use $\Htar$ 
to denote the target Hamiltonian after performing certain logical encoding. Such Hamiltonian operator can thus be related to the original Hamiltonian (i.e., without encoding) via an isometry defined by the logical codewords.}
\begin{equation}\label{eq:enc-cond}
    \Penc \Henc \Penc = \Htar.
\end{equation}
We choose the encoding subspace $\Senc$ to be an eigenspace of a penalty Hamiltonian $\Hpen$, which we assume to have zero eigenvalue without loss of generality.
The penalty Hamiltonian $\Hpen$ is thus used to provide an energy gap between $\Senc$ and its orthogonal complement $\Senc^{\perp}$, which prevents the encoded states from leaking out of the zero-energy subspace.
The energy gap can always be amplified by applying a large penalty coefficient $\lambda$ to the penalty Hamiltonian. The overall simulator Hamiltonian is then given by the sum of the encoding and the penalty Hamiltonians, i.e.
\begin{equation} \label{eq:Hsim}
    \Hsim := \Henc +\lambda \Hpen.
\end{equation}

In a realistic quantum simulation, the simulator Hamiltonian $\Hsim$ is inevitably affected by unwanted noises in the physical system.
Throughout this work, we focus on noises that consist of 
$1$-local coherent perturbations to the simulator Hamiltonian, which can be written in the form $V = \sum_i \epsilon_i V_i$, where $V_i$'s are single-qubit Pauli operators and $\epsilon_i$'s denote the strength of each noise component.
Such an error model describes the dominant noise source in various state-of-the-art analog quantum simulator systems \cite{manuel-endres-2024benchmarking}. 
Specifically, miscalibration and slow drifts in classical controls (laser pulses, microwave drives, etc.) can lead to such $1$-local coherent errors, which is a major error source in current AMO-based quantum hardware. The total, noisy simulator Hamiltonian now becomes 
\begin{equation}
\label{eq:Htot.def}
\Htot := \Hsim + \sum_i \epsilon_i 
V_i
. 
\end{equation}

Error resilience boils down to the requirement that, for initial states inside the encoding subspace, the evolution generated by the total Hamiltonian agrees with the target Hamiltonian up to a deviation that can be arbitrarily reduced by increasing the penalty coefficient $\lambda$.
As discussed e.g.~in Refs.~\cite{Marvian-Lidar-2014-PRL-2local-is-too-local,marvian2017error,marvian2017-PRL-error-supp-subsystem-codes,marvian2019-arxiv-robust}, a necessary and sufficient condition for such error suppression property can be written as follows:
\begin{theorem}
Given an encoding subspace $\Senc$ with projector $\Penc$ and a set of error terms $\{V_i\}$, $\Senc$ is protected against $\{V_i\}$
if
\begin{align}
\label{eqn:error_suppresion_cond}
\Penc V_i\Penc= c_i\Penc,  \quad \forall i,
\end{align}
where $c_i$'s are real coefficients.
\end{theorem}

We refer to Eq.~\eqref{eqn:error_suppresion_cond} as the \textit{error suppression condition}.
In a different context, Eq.~\eqref{eqn:error_suppresion_cond} also defines a quantum error-detecting code.
As shown e.g.~in Refs.~\cite{zanardi2014coherent,Marvian-Lidar-2014-PRL-2local-is-too-local}, for encodings satisfying Eq.~\eqref{eqn:error_suppresion_cond}, one can rigorously bound the deviation between the time evolution under the total Hamiltonian and the target dynamics. More concretely, given a target time evolution $\ket{\psi_{\mathrm{tar}}(t)} := e^{-i \Htar t} \ket{\psi_0}$ where $\ket{\psi_0} \in \Senc$, the infidelity of the corresponding encoded time evolution $\ket{\psi(t)} := e^{-i \Htot t} \ket{\psi_0}$ can be upper bounded as
\begin{equation}\label{eq:infidelity-1st-order-bound}
    1 - \left| \langle \psi_{\mathrm{tar}}(t) | \psi(t) \rangle \right|^2 \le O(\lambda^{-2})
    .
\end{equation}

\begin{figure*}[t]
\includegraphics[width=0.8\linewidth]{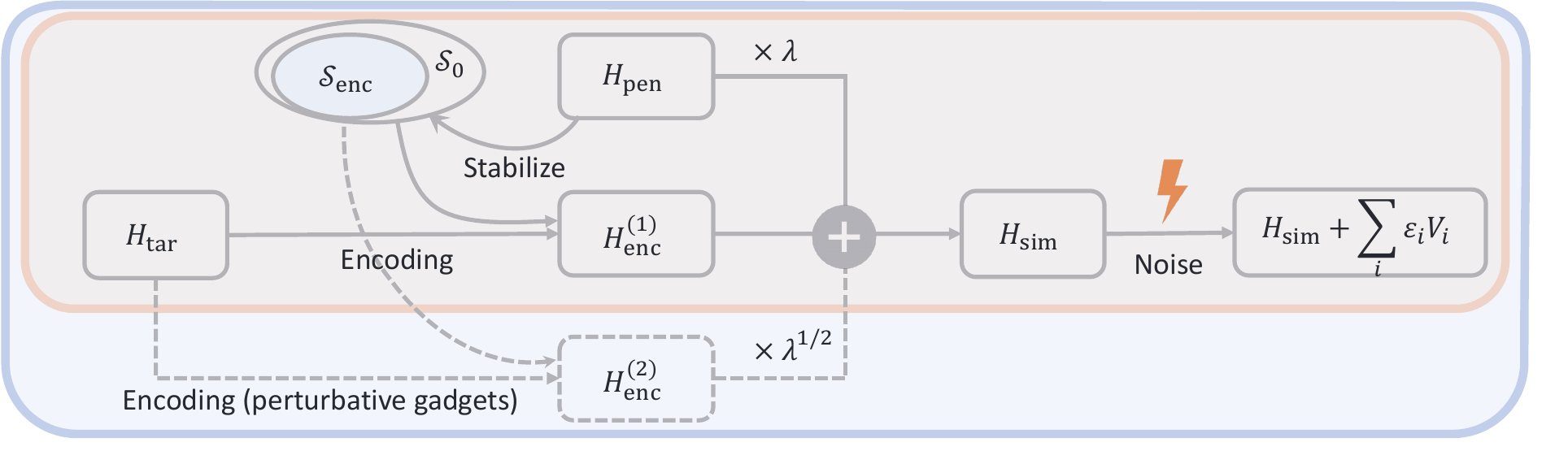}
\centering
\caption{Illustration of the general framework for robust Hamiltonian simulation scheme. The orange block depicts the standard strategy of designing an error-resilient encoding (see Sec.~\ref{sec:general framework}). 
The blue block describes our extension of the framework, featuring perturbative gadgets that preserve the error resilience and enable the implementation of universal logical interactions (see Sec.~\ref{sec:cross block}).}
\label{fig:general_framework}
\end{figure*}

\begin{table}[h!]
    \centering
\begin{tabular}{|c|c||c|c|}
    \hline
        \text{Logical term} & \text{Physical term} & \text{Logical term} & \text{Physical term} \\ 
        \hline
        $\overline{Z_1}$ & $Z _{1} Z _{2} $  &  $\overline{X_1} $ &   $-X _{1} X _{3}  $\\ 
         $\overline{Z_2} $ &  $Z _{1} Z _{3} $ &  $\overline{X_2} $ &  $X _{1} X _{2} $ \\ 
        $ \overline{Z_1 Z_2} $ &  $Z _{2} Z _{3} $ &  $ \overline{X_1 X_2} $&  $- X _{2} X _{3} $\\ 
        \hline
    \end{tabular}
     \caption{Logical operators for the $[[4,2,2]]$ code}
     \label{table:4_2_2_logical_mapping}
\end{table}

\section{\label{sec:results} Results }

While the construction in Sec.~\ref{sec:esc.gen} is generally applicable to suppressing noise in Hamiltonian simulations, in this section, we focus on encodings that can be realized with purely $2$-local Hamiltonians, i.e., that are naturally compatible with current experimental platforms. 
We start by introducing the basic building block in Sec.~\ref{sec:code}, which is a $[[4,2,2]]$ QEDC that can be stabilized by only $2$-local penalty Hamiltonians.
We then generalize this to a scalable code family in Sec.~\ref{sec:code-block composition}.
Finally, in Sec.~\ref{sec:cross block}, we discuss the implementation of logical interactions in this code space via fully $2$-local physical terms, which allows the encoding of computationally universal target Hamiltonians.

\subsection{\label{sec:code}$[[4,2,2]]$ Hamiltonian code}

We consider the following 4-qubit $2$-local commuting penalty Hamiltonian
\begin{align}
\label{neq:H.pen2.422}
H _{\text{pen}} (\mathbf{g} ) 
=  g _{x} X_{1} X_{2} + g _{z} Z _{1} Z_{2} + g _{x}X_{3} X_{4}+ g _{z}Z_{3} Z_{4}
,  
\end{align}
where $\mathbf{g} = (g_x,g_z) \in \mathbb{R}^2$ denotes a pair of nonzero parameters that satisfy the constraint $|g_x/g_z| \ne 1$. It is more informative to group the four qubits into two decoupled pairs; that is, we write $\Hpen(\mathbf{g}) = H_{1,2}(\mathbf{g}) + H_{3,4}(\mathbf{g})$, where $H_{2k-1,2k} = g_x X_{2k-1} X_{2k} + g_z Z_{2k-1} Z_{2k}$ describes the interactions within the $k$th pair, for $k=1,2$.
The four eigenstates of $H_{2k-1,2k}$ coincide with the four Bell states:
\begin{align}
\label{neq:2qb.xz.eigsys}
H _{2k-1,2k} |\Phi _{\mu\nu}\rangle = [(-1)^{\mu}g_x+(-1)^{\nu}g_z] |\Phi _{\mu\nu}\rangle , 
\end{align}
where $|\Phi_{\mu\nu}\rangle \equiv (|0\nu\rangle + (-1)^{\mu} |1\bar{\nu} \rangle) /{\sqrt{2}}$ denotes one of the four Bell pairs that are simultaneous eigenstates of the $2$-local operators: $X_{2k-1} X_{2k} |\Phi_{\mu\nu}\rangle  = (-1)^{\mu} |\Phi_{\mu\nu}\rangle $, $ Z_{2k-1} Z_{2k}|\Phi_{\mu\nu}\rangle =(-1)^{\nu} |\Phi_{\mu\nu}\rangle $, $\mu, \nu\in \{0,1\}$, and $\bar{\mu} \equiv 1- \mu$. Under the constraint $|g_x/g_z| \ne 1$, the four eigenvalues above are distinct from each other.
It follows that the zero-energy eigenspace $ \Senc $ of $\Hpen(\mathbf{g})$ is spanned by the following four states, which we identify as the logical basis states of two logical qubits:
\begin{subequations}
\label{neq:hcode.422.basis}
\begin{align}
& | \overline{0+} \rangle := |\Phi_{00}\rangle \otimes |\Phi_{11}\rangle 
, \quad  | \overline{0-} \rangle := |\Phi_{10}\rangle \otimes |\Phi_{01}\rangle 
, \\
& | \overline{1+} \rangle := |\Phi_{01}\rangle \otimes |\Phi_{10}\rangle 
, \quad | \overline{1-} \rangle := |\Phi_{11}\rangle \otimes |\Phi_{00}\rangle 
.
\end{align}
\end{subequations}
It is straightforward to check that $\Senc  $ satisfies the error suppression condition in Eq.~\eqref{eqn:error_suppresion_cond} against 1-local errors.
In other words, Eq.~\eqref{neq:hcode.422.basis} forms a QEDC with $n=4$ physical qubits, $k=2$ logical qubits, and distance $d=2$.
We refer to this code as the $[[4,2,2]]$ \textit{Hamiltonian code}.
It is interesting to note that this code is equivalent to the standard CSS $[[4,2,2]]$ stabilizer code~\cite{eczoo_stab_4_2_2,quantinuum-icebergcode-Self_2024} up to a single-qubit unitary transformation (see Appendix~\ref{subsec: code equivalence}).

Note that at this point, the assignment of different logical basis states is arbitrary; we choose the convention in Eq.~\eqref{neq:hcode.422.basis} out of convenience, so that the logical $\overline{Z_1 Z_2} $ and $\overline{X_1 X_2} $ operators can be realized with two-local physical operators.
See Table~\ref{table:4_2_2_logical_mapping} for a list of logical operators on the Hamiltonian code that are two-local on physical qubits.

\subsection{\label{sec:code-block composition}Generalization to $[[4n,2n,2]]$ code family by block composition}

To maintain the $2$-local feature of the penalty Hamiltonian of the above $[[4,2,2]]$ Hamiltonian code, we straightforwardly extend the penalty Hamiltonian in Eq.~\eqref{neq:H.pen2.422} to more pairs, i.e.,
\begin{align}\label{eqn:general_penalty_hamiltonian}
\Hpen(\mathbf{g}) = \sum_{i = 1}^n & \left( g_{x} X_{4i-3} X_{4i-2} + g_{z} Z_{4i-3} Z_{4i-2} \right. \nonumber \\
&\left. + g_{x} X_{4i-1} X_{4i} + g_{z} Z_{4i-1} Z_{4i} \right).
\end{align}
All eigenstates of $\Hpen(\mathbf{g})$ can be written as tensor products of $2n$ Bell pairs shown in Eq.~\eqref{neq:2qb.xz.eigsys}. 
The spectrum of the full penalty Hamiltonian is thus given by 
\begin{align}
\label{neq:Hnqb.XZ.pair.spec}
\mathcal{E}_{\boldsymbol{\alpha}} = \sum_{\ell = 1}^{2n} \nu_{\alpha_{\ell}}
, 
\end{align}
where $\boldsymbol{\alpha} = (\alpha_1, \dots, \alpha_{2n})$ is a $2n$-component array with each component taking values from the set $\{ 00,01,10,11\}$, and $\nu _{\alpha _{\ell}}$ corresponds to the eigenvalue of the $\ell$th Bell pair ($\ell=1,2,\dots,2n$). We can thus identify the zero-energy eigenspace $\Szero$ of $\Hpen(\mathbf{g})$ by enumerating all states in each pair such that their corresponding eigenvalues add up to zero, i.e., finding $\boldsymbol{\alpha}$ such that $\mathcal{E} _{\boldsymbol{\alpha}}  =0$ in Eq.~\eqref{neq:Hnqb.XZ.pair.spec}. 
We stress that the commuting nature of Eq.~\eqref{eqn:general_penalty_hamiltonian} guarantees that $\Hpen(\mathbf{g})$ possesses a constant energy gap between $\Szero$ and $\Szero^{\perp}$, which is crucial to the scalable implementation of our error suppression scheme.

Intriguingly, we note that the zero-energy eigenspace $\Szero$ above always forms a QEDC with distance 2.
This occurs because any $1$-local error operator would change the eigenvalue of one of the qubit pairs while leaving the other pairs unchanged.
Consequently, the eigenvalues of the penalty terms would no longer sum to zero, leading to a leakage from $\Szero$.

The degeneracy of the zero-energy eigenspace $\Szero$ can be computed by counting all zero-energy configurations of Eq.~\eqref{neq:Hnqb.XZ.pair.spec}. 
For system sizes greater than $4$, the degeneracy is in general not a power of $2$. Therefore, to encode qubit systems, we choose the encoding subspace $\Senc$ to be a subspace of $\Szero$ that has the proper dimension.
Specifically, we generalize the $[[4,2,2]]$ Hamiltonian code via \textit{code-block composition}, so that $\Senc$ is formed by the tensor product of $n$ copies of the $[[4,2,2]]$ Hamiltonian code. In this case, $\Senc$ corresponds to a $[[4n,2n,2]]$ code, which is a QEDC with a constant code rate.
As $\Senc$ is a subspace of $\Szero$, the error suppression condition in Eq.~\eqref{eqn:error_suppresion_cond} needs to be modified to as follows:
\begin{align}
\label{eqn:esc_subspace}
\Pzero V_i\Pzero= c_i\Pzero,  \quad \forall i,
\end{align}
where $V_i$ denote $1$-local error operators, and $\Pzero $ is the projector onto the zero-energy subspace $\Szero$. This modification is because 
the penalty Hamiltonian only suppresses state transitions between the zero-energy eigenspace $\Szero$ and other energy levels, but not between the encoding subspace $\Senc$ and its orthogonal complement within $\Szero$.
For the same reason, the encoding condition in Eq.~\eqref{eq:enc-cond} should also be augmented with an additional requirement that $\Henc$ induces no transitions between $\Senc$ and $\Senc^{\perp} \cap \Szero$, i.e.,
\begin{equation}\label{eq:enc-cond-aug}
(\Pzero - \Penc) \Henc \Penc = 0.
\end{equation}

\begin{table}[t]
    \centering
    \begin{tabular}{|c|c||c|c|}
    \hline
        \text{Logical term} & \text{Physical term} & \text{Logical term} & \text{Physical term} \\ 
        \hline
        $\overline{Z}_{2i-1}$ & $Z_{4i-3} Z_{4i-2} $  &  $\overline{X}_{2i-1} $ &   $- X_{4i-3} X_{4i-1}$\\ 
         $\overline{Z}_{2i}$ &  $Z_{4i-3} Z_{4i-1} $ &  $\overline{X}_{2i}$ &  $X_{4i-3} X_{4i-2}$ \\ 
        $ \overline{Z}_{2i-1} \overline{Z}_{2i}$ &  $Z_{4i-2} Z_{4i-1} $ &  $\overline{X}_{2i-1} \overline{X}_{2i}$&  $- X_{4i-2} X_{4i-1}$\\ 
        \hline
    \end{tabular} 
    \caption{Inner-block logical operators of the $[[4n,2n,2]]$ Hamiltonian code.}
    \label{table:4n_2n_2_logical_mapping}
\end{table}

The block-based structure of the $[[4n,2n,2]]$ encoding ensures that any $1$- and $2$-local logical operator within each block can be directly implemented via $2$-local physical terms as in Table~\ref{table:4_2_2_logical_mapping}. More specifically, for the $i$th block, the \textit{inner-block logical operators} are given in Table~\ref{table:4n_2n_2_logical_mapping}. In contrast, the cross-block logical operators cannot be realized in this way. 
This is because any $2$-local physical operator that acts on two different blocks necessarily causes leakage out of the code space as a result of the error suppression condition, Eq.~\eqref{eqn:esc_subspace}. Other techniques are thus required to achieve 2-local physical implementation of cross-block logical operators. 
In the next section (Sec.~\ref{sec:cross block}), we utilize the technique of \textit{perturbative gadgets} to address this problem.

\subsection{\label{sec:cross block}Perturbative gadget}
In this section, we introduce perturbative gadgets that allow $2$-local implementation of \textit{cross-block logical operators}.
The basic idea is to make use of $2$-local cross-block physical operators that induce transitions between the code space and states outside of the zero-energy eigenspace. 
When scaling those perturbative terms relative to the penalty Hamiltonian appropriately, we can generate an effective cross-block logical interaction as a leading-order perturbative correction. See Fig.~\ref{fig:general_framework} for a schematic illustration of our general framework.

More concretely, we write the simulator Hamiltonian as 
\begin{equation}\label{eq:Hsim-2nd}
    \Hsim := \lambda\Hpen+\Henc^{(1)}+\sqrt{\lambda}\Henc^{(2)}.
\end{equation}
Here, for convenience, we have divided the encoding Hamiltonian into two parts: $\Henc^{(1)}$ denoting terms that only encode $1$- and $2$-local inner-block logical operators, and $\Henc^{(2)}$ representing terms used in the perturbative gadgets.
We again stress that throughout our discussion, the encoding Hamiltonian only contains $2$-local terms.

In the large penalty coefficient limit $\lambda \gg 1$, we can use perturbation theory to derive the leading order effects generated by $\Henc^{(1)}$ and $\Henc^{(2)}$. 
In this regime, the dominant contributions are given by first-order perturbative term in $\Henc^{(1)}$ and up to second-order perturbative effects of $\Henc^{(2)}$. Instead of Eqs.~\eqref{eq:enc-cond} and \eqref{eq:enc-cond-aug}, we now obtain the following condition for the encoded dynamics to match the desired target Hamiltonian 
\begin{subequations}\label{eq:enc-cond-new}
\begin{align}
& P_0 \Hencsecond P_0 = 0, \\
& \Penc (\Hencfirst + H_{\text{eff}}^{(2)}) \Penc = \Htar, \\
& (\Pzero - \Penc) \Hencfirst \Penc = 0, \\
& (\Pzero - \Penc) \Heff^{(2)} \Penc = 0, \label{eq:enc-cond-new(d)}
\end{align}
\end{subequations}
where $H_{\text{eff}}^{(2)}$ describes the $2$nd-order perturbative effect of $\Henc^{(2)}$ in the zero-energy eigenspace of $\Hpen$, and $\Hpen^{-1}$ denotes the pseudoinverse. Specifically, we have 
\begin{equation}\label{eq:2nd_perturbation}
    H_{\text{eff}}^{(2)} := - \Pzero \Hencsecond
    \Hpen^{-1}
    \Hencsecond \Pzero
    .
\end{equation}
Given the desired cross-block logical interaction as $H_{\text{eff}}^{(2)}$, we can reverse design the corresponding $2$-local physical terms according to Eq.~\eqref{eq:2nd_perturbation}. As concrete examples, Table~\ref{table:gadget-example} provide perturbative gadgets for the cross-block Ising-type and $XY$-type interactions; the gadget for the logical Ising term is also illustrated in Fig.~\ref{fig:gadget-example}.
See Appendix \ref{appxsec:gadgets} for details on a general recipe for constructing 2-local physical terms that systematically create various cross-block logical interactions.

\begin{table}[t]
    \centering
    \begin{tabular}{|c|c|}
    \hline
        $\Hencsecond$ & $H_{\text{eff}}^{(2)}$ \\ 
        \hline
        $\sqrt{\frac{8}{3}} \left( Z_2^{(i)} X_3^{(j)} \pm Z_4^{(i)} X_3^{(j)} \right)$ & $\pm \overline{Z}_2^{(i)} \overline{Z}_1^{(j)}$ \\ 
        $\sqrt{8} \left( Z_2^{(i)} X_1^{(j)} \pm Z_2^{(i)} X_3^{(j)} \right)$ & $\pm \overline{X}_2^{(i)} \overline{X}_1^{(j)}$ \\
        $\sqrt{\frac{8}{3}} \left( Z_2^{(i)} X_3^{(j)} + Z_4^{(i)} X_3^{(j)} \pm 3 Z_2^{(i)} X_1^{(j)} \right)$ & $\overline{Z}_2^{(i)} \overline{Z}_1^{(j)} \pm \overline{X}_2^{(i)} \overline{X}_1^{(j)}$ \\
        \hline
    \end{tabular} 
    \caption{Examples of perturbative gadgets.
    The physical interactions in the left column generate the logical interactions in the right column (cf.~Eq.~\eqref{eq:2nd_perturbation}), up to an inner-block residual logical term which can be canceled out by adjusting $\Hencfirst$ (see Appendix \ref{appxsec:gadgets} for details).
    Here $Z_k^{(i)}$ denotes the Pauli $Z$ operator acting on the $k$-th physical qubit ($k \in \{1,2,3,4\}$) of the $i$-th block, and $\overline{Z}_k^{(i)}$ denotes the Pauli $Z$ operator acting on the $k$-th logical qubit ($k \in \{1,2\}$) of the $i$-th block.
    The parameters of the penalty Hamiltonian are chosen to be $(g_x, g_z) = (1,3)$.}
    \label{table:gadget-example}
\end{table}

\begin{figure}[t]
\includegraphics[width=\linewidth]{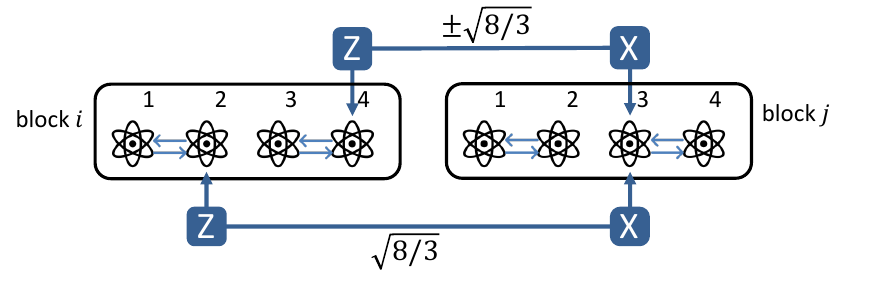}
\centering
\caption{Schematic illustration of the perturbative gadget for generating the 
Ising interaction $\pm \overline{Z}_2^{(i)} \overline{Z}_1^{(j)}$, as shown in the first row in Table~\ref{table:gadget-example}, between the second logical qubit of the $i$-th block and the first logical qubit of the $j$-th block. The sign of this coupling is set by the relative sign between the $Z_2^{(i)} X_3^{(j)}$ and $Z_4^{(i)} X_3^{(j)}$ terms.}
\label{fig:gadget-example}
\end{figure}

The error suppression properties of our perturbative gadget are summarized in the following theorem.
\begin{theorem}\label{thm:main-theorem}
    Consider a target Hamiltonian $\Htar$ acting on the code space $\Penc$ that can be encoded as a simulator Hamiltonian $\Hsim$ of the form in Eq.~\eqref{eq:Hsim-2nd}, such that the encoding conditions Eq.~\eqref{eq:enc-cond-new} and the error suppression condition Eq.~\eqref{eqn:esc_subspace} hold for the errors $V$ introduced in Eq.~\eqref{eq:Htot.def}. Suppose the penalty Hamiltonian $\Hpen$ has a zero-energy eigenspace $\Szero$ with a spectral gap of at least 1~\footnote{Here we refer to the constraint such that all the nonzero eigenvalues of $\Hpen$ are outside the interval $(-1,1)$.}.
    Let $M := \max \{\|\Hencsecond\|^2, \|\Hencfirst + V\|\}$ and suppose $\lambda \ge 25 M$.
    For an arbitrary initial state $\ket{\psi_0} \in \Senc$, let $\ket{\psi_{\mathrm{tar}}(t)} := e^{-i \Htar t} \ket{\psi_0}$ be the state evolved under the target Hamiltonian for time $t>0$, and let $\ket{\psi(t)} := e^{-i (\Hsim+V) t} \ket{\psi_0}$ be the state evolved under the noisy simulator Hamiltonian for the same amount of time. Then the infidelity of the analog quantum simulation is bounded by
    \begin{equation}
        1 - \left| \langle \psi_{\mathrm{tar}}(t) | \psi(t) \rangle \right|^2 \le \frac{M}{\lambda} \left( 6 + 13Mt \right)^2.
    \end{equation}
\end{theorem}
We provide the proof in Appendix \ref{appxsec:sim-err-bounds}. We remark that in contrast to the infidelity bound Eq.~\eqref{eq:infidelity-1st-order-bound} that applies to the case without any perturbative gadgets, here the infidelity scales as $O(1/\lambda)$.
This suggests that perturbative gadgets enable the error-resilient encoding of cross-block interactions at the cost of requiring a larger penalty coefficient to reach the same level of noise suppression. Regardless, we stress that the penalty coefficient required here to ensure a small simulation infidelity still scales polynomially in the system size. 

\section{\label{sec:application} Application}

To demonstrate the power of our robust Hamiltonian simulation scheme, we now discuss multiple examples applying those general techniques to the encoding of a variety of quantum many-body spin models. 

We start with a family of target Hamiltonians on a $2n$-site 1D chain with nearest-neighbor XY couplings:
\begin{equation}\label{eq:Htar-1d}
    \begin{aligned}
    \Htar^{\text{1d}} &= \sum_{k=1}^{2n-1} \left( J_k^{X} \overline{X}_k \overline{X}_{k+1} + J_k^{Z} \overline{Z}_k \overline{Z}_{k+1} \right) \\ &+ \sum_{k=1}^{2n} \left( h_k^{Z} \overline{Z}_k + h_k^{X} \overline{X}_k \right).
    \end{aligned}
\end{equation}
This family includes, for example, the 1D transverse-field Ising (TFI) and spin XY models.
To encode $\Htar^{\text{1d}}$, it suffices to use $n$ copies of the [[4,2,2]] code block, where one identifies sites $2i-1$ and $2i$ as the two logical qubits in the $i$th block.
With this arrangement, the single-site fields $\overline{Z}_k$ and $\overline{X}_k$ for all $k$, as well as the two-site interactions $\overline{Z}_k \overline{Z}_{k+1}$ and $\overline{X}_k \overline{X}_{k+1}$ for odd $k$, are all inner-block logical operators and hence can be directly encoded into $\Hencfirst$ using the prescriptions in Table~\ref{table:4n_2n_2_logical_mapping}.
The interactions $\overline{Z}_k \overline{Z}_{k+1}$ and $\overline{X}_k \overline{X}_{k+1}$ with even $k$ span two different blocks and need to be realized via perturbative gadgets (cf.~Eq.~\eqref{eq:2nd_perturbation}), contributing to terms in $\Hencsecond$; see Appendix~\ref{appxsec:scheme-1D-models} for detailed expressions of the two parts of the encoding Hamiltonian $\Hencfirst$ and $\Hencsecond$, as well as rigorous proof that they satisfy the conditions in Eq.~\eqref{eq:enc-cond-new}.

Further, our scheme allows error-resilient simulation of interacting many-body spin models on higher-dimensional lattices or more general interaction graphs.
However, in contrast to 1D models, if the $\mathbf{g}$ parameters of the penalty Hamiltonian were chosen to be uniform among all blocks, two perturbative gadgets that overlap at a common block can generate unwanted interference terms, causing leakage of quantum states from the encoding subspace $\Senc$ to its orthogonal complement within the zero-energy eigenspace $\Szero$.
Such leakage cannot be protected by the energy-gap protection mechanism and should be avoided (cf.~Eq.~\eqref{eq:enc-cond-new(d)}). As discussed in Appendix~\ref{appxsec:scheme-2D-models}, we have developed a systematic method for constructing perturbative gadgets without any leakage, by varying the $\mathbf{g}$ parameters of the neighboring blocks. Crucially, our method allows the leakage-free encoding of any finite-degree interaction graph using only finitely many different block parameters (i.e., independent of the graph size), which provides a systematic and scalable method for designing robust encoding Hamiltonians. As illustrative examples, we apply our general method to encode the 2D TFI~\cite{2d-TFIM-bernaschi2024quantum} model on a square lattice:
\begin{equation}\label{eq:Htar-2d-TFI-maintext}
    \begin{aligned}
    \Htar^{\text{2d-TFI}} &= \sum_{i,j} \left( J_{i,j}^{(1)} \overline{Z}_{i,j} \overline{Z}_{i+1,j} + J_{i,j}^{(2)} \overline{Z}_{i,j} \overline{Z}_{i,j+1} \right) \\
    &+ \sum_{i,j} \left( h_{i,j}^{Z} \overline{Z}_{i,j} + h_{i,j}^{X} \overline{X}_{i,j} \right).
    \end{aligned}
\end{equation}
See Appendix \ref{appxsec:scheme-2D-models} for detailed constructions, as well as encodings of other 2D models.

Using the Python package \textsf{dynamite} \cite{gregory_d_kahanamoku_meyer_2024_10906046}, we performed extensive numerical simulations of our error suppression scheme for the spin models mentioned above.
In FIG.~\ref{fig:2D-TFI}(a), we compare the encoded Hamiltonian simulation versus the unencoded case for the 2D TFI model.
More concretely, each run of the encoded dynamics consists of time-evolving the Hamiltonian $\Hsim + V$ for a duration of time $t$ starting from a Haar-random initial state inside $\Senc$, with the noise $V$ being a weighted sum over all 1-local Pauli $X$, $Y$ and $Z$ operators with coefficients drawn independently from the uniform distribution over the interval $[-0.1,0.1]$.
The average infidelities with respect to the target final state are plotted in FIG.~\ref{fig:2D-TFI}(a).
Compared with the error in the unencoded dynamics corresponding to time-evolving the Hamiltonian $\Htar + V$, simulation results using our encodings clearly show energy gap protection.
Our scheme notably extends the high-fidelity simulation time by orders of magnitude.
We also plot in FIG.~\ref{fig:2D-TFI}(b) the scaling of the infidelity of the encoded simulation with respect to the penalty coefficient at a fixed time ($t=1$, although the generic behavior does not rely on the choice of the time).
As shown in FIG.~\ref{fig:2D-TFI}(b), the infidelity of the noisy, encoded evolution scales inversely with the penalty coefficient, which is in excellent agreement with the analytical error bound in Theorem~\ref{thm:main-theorem}.
For numerical simulation results for other spin models, see Appendix \ref{appxsec:numerics}.

\begin{figure}[t]
\includegraphics[width=0.9\linewidth]{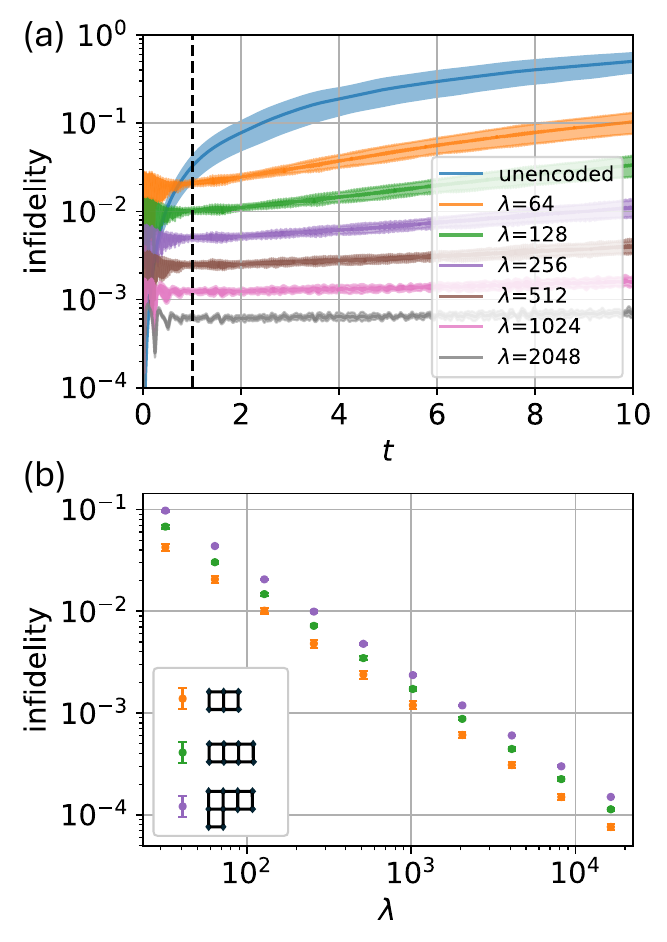}
\centering
\caption{Numerically simulated performance of the encoded 2D TFI model, as quantified by the average infidelity over randomly selected initial state within the encoding subspace $\Senc$.
The target Hamiltonian is given in Eq.~\eqref{eq:Htar-2d-TFI-maintext} with all the Ising interaction strength and local fields chosen to be $1$.
(a) Comparison of the noisy time evolution of the unencoded target Hamiltonian versus the energy-gap-protected encoded Hamiltonian at various penalty coefficients, for a $2\times3$ square lattice.
Each curve is obtained by averaging over 10 random samples of initial states, with the shaded area showing the standard deviation.
(b) Scaling of the infidelity in terms of the penalty coefficient at a fixed time $t=1$ (indicated as the dashed line in (a)), for three different lattice shapes and dimensions shown in the legend.
Each data point is obtained by averaging over 20 random samples of initial states for the first two lattices, and 5 samples for the third lattice, and the error bar shows the standard deviation.
The numerically computed infidelity exhibits a $\propto \lambda ^{-1}$ scaling, in agreement with the prediction from Theorem \ref{thm:main-theorem}.
}
\label{fig:2D-TFI}
\end{figure}

\section{\label{sec:discussion} Discussion}

In this work, we present a robust Hamiltonian simulation scheme, utilizing a block-based encoding protected solely by $2$-local commuting penalty Hamiltonians. 
Crucially, our encoding scheme can be used to realize universal logical interactions with fully $2$-local physical terms. Furthermore, we prove rigorous upper bounds on the simulation infidelity in the presence of generic $1$-local errors.
In particular, to achieve arbitrarily small simulation infidelity, our approach requires penalty terms that scale at most polynomially in system size, which provides a significant improvement over previous work~\cite{Marvian-Lidar-2014-PRL-2local-is-too-local,marvian2017error,marvian2017-PRL-error-supp-subsystem-codes,marvian2019-arxiv-robust}. Our approach allows an error-resilient implementation of a wide range of many-body quantum spin models; notably, it can be used to encode $2$-local XY Hamiltonians, which have been shown to be QMA-complete~\cite{biamonte-realizable,2d-lattice-complexity-oliveira2005complexity}. As such, our methods potentially enable the realization of hardware-efficient universal robust analog quantum simulation in current experimental platforms. While we focused on suppressing $1$-local coherent errors, it will also be interesting to explore the application of those ideas to more general stochastic error models such as Markovian and non-Markovian dissipation~\cite{Lidar2019arbitrarytime}.

Previous work has shown that $2$-local commuting Hamiltonians are not enough to protect quantum dynamics in the ground space against $1$-local errors~\cite{Marvian-Lidar-2014-PRL-2local-is-too-local}. Our work bypasses this constraint, by lifting the encoding subspace from the ground to an excited space, enabling encodings that suppress $1$-local errors via purely $2$-local commuting penalty terms. A natural question to ask is whether our framework can be further extended to enable the suppression of higher-weight physical errors. For this, we also prove a generalized version of the no-go theorem, which prevents any $2$-local commuting Hamiltonians from protecting the encoding subspace against $2$-local or higher-weight errors (see Appendix \ref{appendix:new no-go}). Thus, the question of how to extend the robust Hamiltonian simulation scheme to suppress errors beyond $1$-local terms remains an open problem worth future exploration.

We reiterate that our scheme only utilizes $2$-local physical terms in both the penalty and the encoded simulator Hamiltonians, making it directly compatible with NISQ devices. Furthermore, our approach involves commuting penalty Hamiltonians, so that the encoded dynamics are protected by an energy gap that stays constant as system size grows, assuming individual penalty terms are fixed. This makes our encoding a promising candidate for exploring scalable quantum many-body systems with high accuracy.

Finally, we note that while our approach can significantly improve the hardware efficiency of error-suppression schemes for Hamiltonian simulation, we still assume in our considerations that one has controlled access to generic \textit{$2$-local} Hamiltonian interactions. This may not be the case in practical systems: for example, in Rydberg atoms, the natural $2$-local interactions between qubits are given by the Rydberg blockade effect, which corresponds to $ZZ$-type couplings in the physical Hamiltonian. In those cases, our simulator Hamiltonian can be readily implemented using Floquet engineering techniques (see e.g.~\cite{koyluouglu2024floquet}). While we expect such Floquet-based realizations to effectively suppress error in the ideal limit, it is also interesting to explore the performance of our protocol in the presence of physical constraints such as finite Floquet periods and other imperfections. Alternatively, our robust Hamiltonian simulator can also be modified to use protection enabled by dissipative dynamics, which can be potentially realized using engineered dissipation techniques~\cite{harley2024going-beyond-gadget}. We leave those to future work.

\vspace{0.1in}

\begin{acknowledgments}
We thank Daniel Lidar, Andrew Daley, and Iman Marvian for helpful discussions and insightful feedback. 
We acknowledge the support from a QuICS seed grant.
S.L., Y.C., H.D., and X.W.~were supported by the U.S.~Department of Energy, Office of Science, Office of Advanced Scientific Computing Research, Accelerated Research in Quantum Computing under Award Number DE-SC002027, Air Force Office of Scientific Research under award number FA9550-21-1-0209, the U.S. National Science Foundation grant
CCF-1955206 and CCF-1942837 (CAREER), and a Sloan research fellowship.
Z.X.~was supported by the ARO MURI grant W911NF-22-S-0007, by the National Science Foundation the Quantum Leap Big Idea under Grant No.~OMA-1936388, and by the Defense Advanced Research Projects Agency (DARPA) under Contract No.~HR00112190071.
This research was supported in part by grant NSF PHY-2309135 to the Kavli Institute for Theoretical Physics (KITP). 
Y.-X.W.~acknowledges support from a QuICS Hartree Postdoctoral Fellowship.
\end{acknowledgments}

\appendix

\bibliography{main}

\begin{thebibliography}{63}%
\makeatletter
\providecommand \@ifxundefined [1]{%
 \@ifx{#1\undefined}
}%
\providecommand \@ifnum [1]{%
 \ifnum #1\expandafter \@firstoftwo
 \else \expandafter \@secondoftwo
 \fi
}%
\providecommand \@ifx [1]{%
 \ifx #1\expandafter \@firstoftwo
 \else \expandafter \@secondoftwo
 \fi
}%
\providecommand \natexlab [1]{#1}%
\providecommand \enquote  [1]{``#1''}%
\providecommand \bibnamefont  [1]{#1}%
\providecommand \bibfnamefont [1]{#1}%
\providecommand \citenamefont [1]{#1}%
\providecommand \href@noop [0]{\@secondoftwo}%
\providecommand \href [0]{\begingroup \@sanitize@url \@href}%
\providecommand \@href[1]{\@@startlink{#1}\@@href}%
\providecommand \@@href[1]{\endgroup#1\@@endlink}%
\providecommand \@sanitize@url [0]{\catcode `\\12\catcode `\$12\catcode `\&12\catcode `\#12\catcode `\^12\catcode `\_12\catcode `\%12\relax}%
\providecommand \@@startlink[1]{}%
\providecommand \@@endlink[0]{}%
\providecommand \url  [0]{\begingroup\@sanitize@url \@url }%
\providecommand \@url [1]{\endgroup\@href {#1}{\urlprefix }}%
\providecommand \urlprefix  [0]{URL }%
\providecommand \Eprint [0]{\href }%
\providecommand \doibase [0]{https://doi.org/}%
\providecommand \selectlanguage [0]{\@gobble}%
\providecommand \bibinfo  [0]{\@secondoftwo}%
\providecommand \bibfield  [0]{\@secondoftwo}%
\providecommand \translation [1]{[#1]}%
\providecommand \BibitemOpen [0]{}%
\providecommand \bibitemStop [0]{}%
\providecommand \bibitemNoStop [0]{.\EOS\space}%
\providecommand \EOS [0]{\spacefactor3000\relax}%
\providecommand \BibitemShut  [1]{\csname bibitem#1\endcsname}%
\let\auto@bib@innerbib\@empty
\bibitem [{\citenamefont {Bernien}\ \emph {et~al.}(2017)\citenamefont {Bernien}, \citenamefont {Schwartz}, \citenamefont {Keesling}, \citenamefont {Levine}, \citenamefont {Omran}, \citenamefont {Pichler}, \citenamefont {Choi}, \citenamefont {Zibrov}, \citenamefont {Endres}, \citenamefont {Greiner} \emph {et~al.}}]{bernien2017probing-many-body-51-atom}%
  \BibitemOpen
  \bibfield  {author} {\bibinfo {author} {\bibfnamefont {H.}~\bibnamefont {Bernien}}, \bibinfo {author} {\bibfnamefont {S.}~\bibnamefont {Schwartz}}, \bibinfo {author} {\bibfnamefont {A.}~\bibnamefont {Keesling}}, \bibinfo {author} {\bibfnamefont {H.}~\bibnamefont {Levine}}, \bibinfo {author} {\bibfnamefont {A.}~\bibnamefont {Omran}}, \bibinfo {author} {\bibfnamefont {H.}~\bibnamefont {Pichler}}, \bibinfo {author} {\bibfnamefont {S.}~\bibnamefont {Choi}}, \bibinfo {author} {\bibfnamefont {A.~S.}\ \bibnamefont {Zibrov}}, \bibinfo {author} {\bibfnamefont {M.}~\bibnamefont {Endres}}, \bibinfo {author} {\bibfnamefont {M.}~\bibnamefont {Greiner}}, \emph {et~al.},\ }\bibfield  {title} {\bibinfo {title} {Probing many-body dynamics on a 51-atom quantum simulator},\ }\href {https://doi.org/10.1038/nature24622} {\bibfield  {journal} {\bibinfo  {journal} {Nature}\ }\textbf {\bibinfo {volume} {551}},\ \bibinfo {pages} {579} (\bibinfo {year} {2017})}\BibitemShut {NoStop}%
\bibitem [{\citenamefont {Browaeys}\ and\ \citenamefont {Lahaye}(2020)}]{browaeys2020many-body-phys-individually-Rydberg}%
  \BibitemOpen
  \bibfield  {author} {\bibinfo {author} {\bibfnamefont {A.}~\bibnamefont {Browaeys}}\ and\ \bibinfo {author} {\bibfnamefont {T.}~\bibnamefont {Lahaye}},\ }\bibfield  {title} {\bibinfo {title} {Many-body physics with individually controlled rydberg atoms},\ }\href {https://doi.org/10.1038/s41567-019-0733-z} {\bibfield  {journal} {\bibinfo  {journal} {Nature Physics}\ }\textbf {\bibinfo {volume} {16}},\ \bibinfo {pages} {132} (\bibinfo {year} {2020})}\BibitemShut {NoStop}%
\bibitem [{\citenamefont {Bluvstein}\ \emph {et~al.}(2021)\citenamefont {Bluvstein}, \citenamefont {Omran}, \citenamefont {Levine}, \citenamefont {Keesling}, \citenamefont {Semeghini}, \citenamefont {Ebadi}, \citenamefont {Wang}, \citenamefont {Michailidis}, \citenamefont {Maskara}, \citenamefont {Ho}, \citenamefont {Choi}, \citenamefont {Serbyn}, \citenamefont {Greiner}, \citenamefont {Vuletić},\ and\ \citenamefont {Lukin}}]{bluvstein2021controlling-many-body-Rydberg}%
  \BibitemOpen
  \bibfield  {author} {\bibinfo {author} {\bibfnamefont {D.}~\bibnamefont {Bluvstein}}, \bibinfo {author} {\bibfnamefont {A.}~\bibnamefont {Omran}}, \bibinfo {author} {\bibfnamefont {H.}~\bibnamefont {Levine}}, \bibinfo {author} {\bibfnamefont {A.}~\bibnamefont {Keesling}}, \bibinfo {author} {\bibfnamefont {G.}~\bibnamefont {Semeghini}}, \bibinfo {author} {\bibfnamefont {S.}~\bibnamefont {Ebadi}}, \bibinfo {author} {\bibfnamefont {T.~T.}\ \bibnamefont {Wang}}, \bibinfo {author} {\bibfnamefont {A.~A.}\ \bibnamefont {Michailidis}}, \bibinfo {author} {\bibfnamefont {N.}~\bibnamefont {Maskara}}, \bibinfo {author} {\bibfnamefont {W.~W.}\ \bibnamefont {Ho}}, \bibinfo {author} {\bibfnamefont {S.}~\bibnamefont {Choi}}, \bibinfo {author} {\bibfnamefont {M.}~\bibnamefont {Serbyn}}, \bibinfo {author} {\bibfnamefont {M.}~\bibnamefont {Greiner}}, \bibinfo {author} {\bibfnamefont {V.}~\bibnamefont {Vuletić}},\ and\ \bibinfo {author} {\bibfnamefont {M.~D.}\ \bibnamefont {Lukin}},\ }\bibfield  {title} {\bibinfo {title}
  {Controlling quantum many-body dynamics in driven rydberg atom arrays},\ }\href {https://doi.org/10.1126/science.abg2530} {\bibfield  {journal} {\bibinfo  {journal} {Science}\ }\textbf {\bibinfo {volume} {371}},\ \bibinfo {pages} {1355} (\bibinfo {year} {2021})}\BibitemShut {NoStop}%
\bibitem [{\citenamefont {Arg{\"u}ello-Luengo}\ \emph {et~al.}(2019)\citenamefont {Arg{\"u}ello-Luengo}, \citenamefont {Gonz{\'a}lez-Tudela}, \citenamefont {Shi}, \citenamefont {Zoller},\ and\ \citenamefont {Cirac}}]{arguello2019analogue-quantum-chemistry-sim}%
  \BibitemOpen
  \bibfield  {author} {\bibinfo {author} {\bibfnamefont {J.}~\bibnamefont {Arg{\"u}ello-Luengo}}, \bibinfo {author} {\bibfnamefont {A.}~\bibnamefont {Gonz{\'a}lez-Tudela}}, \bibinfo {author} {\bibfnamefont {T.}~\bibnamefont {Shi}}, \bibinfo {author} {\bibfnamefont {P.}~\bibnamefont {Zoller}},\ and\ \bibinfo {author} {\bibfnamefont {J.~I.}\ \bibnamefont {Cirac}},\ }\bibfield  {title} {\bibinfo {title} {Analogue quantum chemistry simulation},\ }\href {https://doi.org/10.1038/s41586-019-1614-4} {\bibfield  {journal} {\bibinfo  {journal} {Nature}\ }\textbf {\bibinfo {volume} {574}},\ \bibinfo {pages} {215} (\bibinfo {year} {2019})}\BibitemShut {NoStop}%
\bibitem [{\citenamefont {Clinton}\ \emph {et~al.}(2024)\citenamefont {Clinton}, \citenamefont {Cubitt}, \citenamefont {Flynn}, \citenamefont {Gambetta}, \citenamefont {Klassen}, \citenamefont {Montanaro}, \citenamefont {Piddock}, \citenamefont {Santos},\ and\ \citenamefont {Sheridan}}]{clinton2024towards-near-term-qsim-material}%
  \BibitemOpen
  \bibfield  {author} {\bibinfo {author} {\bibfnamefont {L.}~\bibnamefont {Clinton}}, \bibinfo {author} {\bibfnamefont {T.}~\bibnamefont {Cubitt}}, \bibinfo {author} {\bibfnamefont {B.}~\bibnamefont {Flynn}}, \bibinfo {author} {\bibfnamefont {F.~M.}\ \bibnamefont {Gambetta}}, \bibinfo {author} {\bibfnamefont {J.}~\bibnamefont {Klassen}}, \bibinfo {author} {\bibfnamefont {A.}~\bibnamefont {Montanaro}}, \bibinfo {author} {\bibfnamefont {S.}~\bibnamefont {Piddock}}, \bibinfo {author} {\bibfnamefont {R.~A.}\ \bibnamefont {Santos}},\ and\ \bibinfo {author} {\bibfnamefont {E.}~\bibnamefont {Sheridan}},\ }\bibfield  {title} {\bibinfo {title} {Towards near-term quantum simulation of materials},\ }\href {https://doi.org/10.1038/s41467-023-43479-6} {\bibfield  {journal} {\bibinfo  {journal} {Nature Communications}\ }\textbf {\bibinfo {volume} {15}},\ \bibinfo {pages} {211} (\bibinfo {year} {2024})}\BibitemShut {NoStop}%
\bibitem [{\citenamefont {Keesling}\ \emph {et~al.}(2019)\citenamefont {Keesling}, \citenamefont {Omran}, \citenamefont {Levine}, \citenamefont {Bernien}, \citenamefont {Pichler}, \citenamefont {Choi}, \citenamefont {Samajdar}, \citenamefont {Schwartz}, \citenamefont {Silvi}, \citenamefont {Sachdev} \emph {et~al.}}]{keesling2019quantum-Kibble-Zurek}%
  \BibitemOpen
  \bibfield  {author} {\bibinfo {author} {\bibfnamefont {A.}~\bibnamefont {Keesling}}, \bibinfo {author} {\bibfnamefont {A.}~\bibnamefont {Omran}}, \bibinfo {author} {\bibfnamefont {H.}~\bibnamefont {Levine}}, \bibinfo {author} {\bibfnamefont {H.}~\bibnamefont {Bernien}}, \bibinfo {author} {\bibfnamefont {H.}~\bibnamefont {Pichler}}, \bibinfo {author} {\bibfnamefont {S.}~\bibnamefont {Choi}}, \bibinfo {author} {\bibfnamefont {R.}~\bibnamefont {Samajdar}}, \bibinfo {author} {\bibfnamefont {S.}~\bibnamefont {Schwartz}}, \bibinfo {author} {\bibfnamefont {P.}~\bibnamefont {Silvi}}, \bibinfo {author} {\bibfnamefont {S.}~\bibnamefont {Sachdev}}, \emph {et~al.},\ }\bibfield  {title} {\bibinfo {title} {Quantum kibble--zurek mechanism and critical dynamics on a programmable rydberg simulator},\ }\href {https://doi.org/10.1038/s41586-019-1070-1} {\bibfield  {journal} {\bibinfo  {journal} {Nature}\ }\textbf {\bibinfo {volume} {568}},\ \bibinfo {pages} {207} (\bibinfo {year} {2019})}\BibitemShut {NoStop}%
\bibitem [{\citenamefont {Andersen}\ \emph {et~al.}(2024)\citenamefont {Andersen}, \citenamefont {Astrakhantsev}, \citenamefont {Karamlou}, \citenamefont {Berndtsson}, \citenamefont {Motruk}, \citenamefont {Szasz}, \citenamefont {Gross}, \citenamefont {Westerhout}, \citenamefont {Zhang}, \citenamefont {Forati} \emph {et~al.}}]{andersen2024thermalization}%
  \BibitemOpen
  \bibfield  {author} {\bibinfo {author} {\bibfnamefont {T.~I.}\ \bibnamefont {Andersen}}, \bibinfo {author} {\bibfnamefont {N.}~\bibnamefont {Astrakhantsev}}, \bibinfo {author} {\bibfnamefont {A.}~\bibnamefont {Karamlou}}, \bibinfo {author} {\bibfnamefont {J.}~\bibnamefont {Berndtsson}}, \bibinfo {author} {\bibfnamefont {J.}~\bibnamefont {Motruk}}, \bibinfo {author} {\bibfnamefont {A.}~\bibnamefont {Szasz}}, \bibinfo {author} {\bibfnamefont {J.~A.}\ \bibnamefont {Gross}}, \bibinfo {author} {\bibfnamefont {T.}~\bibnamefont {Westerhout}}, \bibinfo {author} {\bibfnamefont {Y.}~\bibnamefont {Zhang}}, \bibinfo {author} {\bibfnamefont {E.}~\bibnamefont {Forati}}, \emph {et~al.},\ }\bibfield  {title} {\bibinfo {title} {Thermalization and criticality on an analog-digital quantum simulator},\ }\href {https://arxiv.org/abs/2405.17385} {\bibfield  {journal} {\bibinfo  {journal} {arXiv preprint arXiv:2405.17385}\ } (\bibinfo {year} {2024})}\BibitemShut {NoStop}%
\bibitem [{\citenamefont {Farhi}\ \emph {et~al.}(2001)\citenamefont {Farhi}, \citenamefont {Goldstone}, \citenamefont {Gutmann}, \citenamefont {Lapan}, \citenamefont {Lundgren},\ and\ \citenamefont {Preda}}]{Farhi_2001}%
  \BibitemOpen
  \bibfield  {author} {\bibinfo {author} {\bibfnamefont {E.}~\bibnamefont {Farhi}}, \bibinfo {author} {\bibfnamefont {J.}~\bibnamefont {Goldstone}}, \bibinfo {author} {\bibfnamefont {S.}~\bibnamefont {Gutmann}}, \bibinfo {author} {\bibfnamefont {J.}~\bibnamefont {Lapan}}, \bibinfo {author} {\bibfnamefont {A.}~\bibnamefont {Lundgren}},\ and\ \bibinfo {author} {\bibfnamefont {D.}~\bibnamefont {Preda}},\ }\bibfield  {title} {\bibinfo {title} {A quantum adiabatic evolution algorithm applied to random instances of an np-complete problem},\ }\href {https://doi.org/10.1126/science.1057726} {\bibfield  {journal} {\bibinfo  {journal} {Science}\ }\textbf {\bibinfo {volume} {292}},\ \bibinfo {pages} {472–475} (\bibinfo {year} {2001})}\BibitemShut {NoStop}%
\bibitem [{\citenamefont {Crosson}\ and\ \citenamefont {Lidar}(2021)}]{Crosson_2021}%
  \BibitemOpen
  \bibfield  {author} {\bibinfo {author} {\bibfnamefont {E.~J.}\ \bibnamefont {Crosson}}\ and\ \bibinfo {author} {\bibfnamefont {D.~A.}\ \bibnamefont {Lidar}},\ }\bibfield  {title} {\bibinfo {title} {Prospects for quantum enhancement with diabatic quantum annealing},\ }\href {https://doi.org/10.1038/s42254-021-00313-6} {\bibfield  {journal} {\bibinfo  {journal} {Nature Reviews Physics}\ }\textbf {\bibinfo {volume} {3}},\ \bibinfo {pages} {466–489} (\bibinfo {year} {2021})}\BibitemShut {NoStop}%
\bibitem [{\citenamefont {Ebadi}\ \emph {et~al.}(2022)\citenamefont {Ebadi}, \citenamefont {Keesling}, \citenamefont {Cain}, \citenamefont {Wang}, \citenamefont {Levine}, \citenamefont {Bluvstein}, \citenamefont {Semeghini}, \citenamefont {Omran}, \citenamefont {Liu}, \citenamefont {Samajdar}, \citenamefont {Luo}, \citenamefont {Nash}, \citenamefont {Gao}, \citenamefont {Barak}, \citenamefont {Farhi}, \citenamefont {Sachdev}, \citenamefont {Gemelke}, \citenamefont {Zhou}, \citenamefont {Choi}, \citenamefont {Pichler}, \citenamefont {Wang}, \citenamefont {Greiner}, \citenamefont {Vuletić},\ and\ \citenamefont {Lukin}}]{ebadi2022quantum-optimization-mis}%
  \BibitemOpen
  \bibfield  {author} {\bibinfo {author} {\bibfnamefont {S.}~\bibnamefont {Ebadi}}, \bibinfo {author} {\bibfnamefont {A.}~\bibnamefont {Keesling}}, \bibinfo {author} {\bibfnamefont {M.}~\bibnamefont {Cain}}, \bibinfo {author} {\bibfnamefont {T.~T.}\ \bibnamefont {Wang}}, \bibinfo {author} {\bibfnamefont {H.}~\bibnamefont {Levine}}, \bibinfo {author} {\bibfnamefont {D.}~\bibnamefont {Bluvstein}}, \bibinfo {author} {\bibfnamefont {G.}~\bibnamefont {Semeghini}}, \bibinfo {author} {\bibfnamefont {A.}~\bibnamefont {Omran}}, \bibinfo {author} {\bibfnamefont {J.-G.}\ \bibnamefont {Liu}}, \bibinfo {author} {\bibfnamefont {R.}~\bibnamefont {Samajdar}}, \bibinfo {author} {\bibfnamefont {X.-Z.}\ \bibnamefont {Luo}}, \bibinfo {author} {\bibfnamefont {B.}~\bibnamefont {Nash}}, \bibinfo {author} {\bibfnamefont {X.}~\bibnamefont {Gao}}, \bibinfo {author} {\bibfnamefont {B.}~\bibnamefont {Barak}}, \bibinfo {author} {\bibfnamefont {E.}~\bibnamefont {Farhi}}, \bibinfo {author} {\bibfnamefont {S.}~\bibnamefont {Sachdev}},
  \bibinfo {author} {\bibfnamefont {N.}~\bibnamefont {Gemelke}}, \bibinfo {author} {\bibfnamefont {L.}~\bibnamefont {Zhou}}, \bibinfo {author} {\bibfnamefont {S.}~\bibnamefont {Choi}}, \bibinfo {author} {\bibfnamefont {H.}~\bibnamefont {Pichler}}, \bibinfo {author} {\bibfnamefont {S.-T.}\ \bibnamefont {Wang}}, \bibinfo {author} {\bibfnamefont {M.}~\bibnamefont {Greiner}}, \bibinfo {author} {\bibfnamefont {V.}~\bibnamefont {Vuletić}},\ and\ \bibinfo {author} {\bibfnamefont {M.~D.}\ \bibnamefont {Lukin}},\ }\bibfield  {title} {\bibinfo {title} {Quantum optimization of maximum independent set using rydberg atom arrays},\ }\href {https://doi.org/10.1126/science.abo6587} {\bibfield  {journal} {\bibinfo  {journal} {Science}\ }\textbf {\bibinfo {volume} {376}},\ \bibinfo {pages} {1209} (\bibinfo {year} {2022})}\BibitemShut {NoStop}%
\bibitem [{\citenamefont {Leng}\ \emph {et~al.}(2023)\citenamefont {Leng}, \citenamefont {Hickman}, \citenamefont {Li},\ and\ \citenamefont {Wu}}]{leng2023QHD}%
  \BibitemOpen
  \bibfield  {author} {\bibinfo {author} {\bibfnamefont {J.}~\bibnamefont {Leng}}, \bibinfo {author} {\bibfnamefont {E.}~\bibnamefont {Hickman}}, \bibinfo {author} {\bibfnamefont {J.}~\bibnamefont {Li}},\ and\ \bibinfo {author} {\bibfnamefont {X.}~\bibnamefont {Wu}},\ }\bibfield  {title} {\bibinfo {title} {Quantum hamiltonian descent},\ }\href {https://arxiv.org/abs/2303.01471} {\bibfield  {journal} {\bibinfo  {journal} {arXiv preprint arXiv:2303.01471}\ } (\bibinfo {year} {2023})}\BibitemShut {NoStop}%
\bibitem [{\citenamefont {Cirac}\ and\ \citenamefont {Zoller}(2012)}]{cirac2012goals-and-opportunities-qsim}%
  \BibitemOpen
  \bibfield  {author} {\bibinfo {author} {\bibfnamefont {J.~I.}\ \bibnamefont {Cirac}}\ and\ \bibinfo {author} {\bibfnamefont {P.}~\bibnamefont {Zoller}},\ }\bibfield  {title} {\bibinfo {title} {Goals and opportunities in quantum simulation},\ }\href {https://doi.org/10.1038/nphys2275} {\bibfield  {journal} {\bibinfo  {journal} {Nature physics}\ }\textbf {\bibinfo {volume} {8}},\ \bibinfo {pages} {264} (\bibinfo {year} {2012})}\BibitemShut {NoStop}%
\bibitem [{\citenamefont {Gross}\ and\ \citenamefont {Bloch}(2017)}]{gross2017qsim-ultracold-atoms-optical-lattice}%
  \BibitemOpen
  \bibfield  {author} {\bibinfo {author} {\bibfnamefont {C.}~\bibnamefont {Gross}}\ and\ \bibinfo {author} {\bibfnamefont {I.}~\bibnamefont {Bloch}},\ }\bibfield  {title} {\bibinfo {title} {Quantum simulations with ultracold atoms in optical lattices},\ }\href {https://doi.org/10.1126/science.aal3837} {\bibfield  {journal} {\bibinfo  {journal} {Science}\ }\textbf {\bibinfo {volume} {357}},\ \bibinfo {pages} {995} (\bibinfo {year} {2017})}\BibitemShut {NoStop}%
\bibitem [{\citenamefont {Cubitt}\ \emph {et~al.}(2018)\citenamefont {Cubitt}, \citenamefont {Montanaro},\ and\ \citenamefont {Piddock}}]{cubitt2018universal}%
  \BibitemOpen
  \bibfield  {author} {\bibinfo {author} {\bibfnamefont {T.~S.}\ \bibnamefont {Cubitt}}, \bibinfo {author} {\bibfnamefont {A.}~\bibnamefont {Montanaro}},\ and\ \bibinfo {author} {\bibfnamefont {S.}~\bibnamefont {Piddock}},\ }\bibfield  {title} {\bibinfo {title} {Universal quantum hamiltonians},\ }\href {https://doi.org/10.1073/pnas.1804949115} {\bibfield  {journal} {\bibinfo  {journal} {Proceedings of the National Academy of Sciences}\ }\textbf {\bibinfo {volume} {115}},\ \bibinfo {pages} {9497} (\bibinfo {year} {2018})}\BibitemShut {NoStop}%
\bibitem [{\citenamefont {Zhou}\ and\ \citenamefont {Aharonov}(2021)}]{zhou2021strongly}%
  \BibitemOpen
  \bibfield  {author} {\bibinfo {author} {\bibfnamefont {L.}~\bibnamefont {Zhou}}\ and\ \bibinfo {author} {\bibfnamefont {D.}~\bibnamefont {Aharonov}},\ }\bibfield  {title} {\bibinfo {title} {Strongly universal hamiltonian simulators},\ }\href {https://arxiv.org/abs/2102.02991} {\bibfield  {journal} {\bibinfo  {journal} {arXiv preprint arXiv:2102.02991}\ } (\bibinfo {year} {2021})}\BibitemShut {NoStop}%
\bibitem [{\citenamefont {Altman}\ \emph {et~al.}(2021)\citenamefont {Altman}, \citenamefont {Brown}, \citenamefont {Carleo}, \citenamefont {Carr}, \citenamefont {Demler}, \citenamefont {Chin}, \citenamefont {DeMarco}, \citenamefont {Economou}, \citenamefont {Eriksson}, \citenamefont {Fu}, \citenamefont {Greiner}, \citenamefont {Hazzard}, \citenamefont {Hulet}, \citenamefont {Koll\'ar}, \citenamefont {Lev}, \citenamefont {Lukin}, \citenamefont {Ma}, \citenamefont {Mi}, \citenamefont {Misra}, \citenamefont {Monroe}, \citenamefont {Murch}, \citenamefont {Nazario}, \citenamefont {Ni}, \citenamefont {Potter}, \citenamefont {Roushan}, \citenamefont {Saffman}, \citenamefont {Schleier-Smith}, \citenamefont {Siddiqi}, \citenamefont {Simmonds}, \citenamefont {Singh}, \citenamefont {Spielman}, \citenamefont {Temme}, \citenamefont {Weiss}, \citenamefont {Vu\ifmmode \check{c}\else \v{c}\fi{}kovi\ifmmode~\acute{c}\else \'{c}\fi{}}, \citenamefont {Vuleti\ifmmode~\acute{c}\else \'{c}\fi{}}, \citenamefont {Ye},\ and\
  \citenamefont {Zwierlein}}]{altman2021quantum}%
  \BibitemOpen
  \bibfield  {author} {\bibinfo {author} {\bibfnamefont {E.}~\bibnamefont {Altman}}, \bibinfo {author} {\bibfnamefont {K.~R.}\ \bibnamefont {Brown}}, \bibinfo {author} {\bibfnamefont {G.}~\bibnamefont {Carleo}}, \bibinfo {author} {\bibfnamefont {L.~D.}\ \bibnamefont {Carr}}, \bibinfo {author} {\bibfnamefont {E.}~\bibnamefont {Demler}}, \bibinfo {author} {\bibfnamefont {C.}~\bibnamefont {Chin}}, \bibinfo {author} {\bibfnamefont {B.}~\bibnamefont {DeMarco}}, \bibinfo {author} {\bibfnamefont {S.~E.}\ \bibnamefont {Economou}}, \bibinfo {author} {\bibfnamefont {M.~A.}\ \bibnamefont {Eriksson}}, \bibinfo {author} {\bibfnamefont {K.-M.~C.}\ \bibnamefont {Fu}}, \bibinfo {author} {\bibfnamefont {M.}~\bibnamefont {Greiner}}, \bibinfo {author} {\bibfnamefont {K.~R.}\ \bibnamefont {Hazzard}}, \bibinfo {author} {\bibfnamefont {R.~G.}\ \bibnamefont {Hulet}}, \bibinfo {author} {\bibfnamefont {A.~J.}\ \bibnamefont {Koll\'ar}}, \bibinfo {author} {\bibfnamefont {B.~L.}\ \bibnamefont {Lev}}, \bibinfo {author} {\bibfnamefont
  {M.~D.}\ \bibnamefont {Lukin}}, \bibinfo {author} {\bibfnamefont {R.}~\bibnamefont {Ma}}, \bibinfo {author} {\bibfnamefont {X.}~\bibnamefont {Mi}}, \bibinfo {author} {\bibfnamefont {S.}~\bibnamefont {Misra}}, \bibinfo {author} {\bibfnamefont {C.}~\bibnamefont {Monroe}}, \bibinfo {author} {\bibfnamefont {K.}~\bibnamefont {Murch}}, \bibinfo {author} {\bibfnamefont {Z.}~\bibnamefont {Nazario}}, \bibinfo {author} {\bibfnamefont {K.-K.}\ \bibnamefont {Ni}}, \bibinfo {author} {\bibfnamefont {A.~C.}\ \bibnamefont {Potter}}, \bibinfo {author} {\bibfnamefont {P.}~\bibnamefont {Roushan}}, \bibinfo {author} {\bibfnamefont {M.}~\bibnamefont {Saffman}}, \bibinfo {author} {\bibfnamefont {M.}~\bibnamefont {Schleier-Smith}}, \bibinfo {author} {\bibfnamefont {I.}~\bibnamefont {Siddiqi}}, \bibinfo {author} {\bibfnamefont {R.}~\bibnamefont {Simmonds}}, \bibinfo {author} {\bibfnamefont {M.}~\bibnamefont {Singh}}, \bibinfo {author} {\bibfnamefont {I.}~\bibnamefont {Spielman}}, \bibinfo {author} {\bibfnamefont {K.}~\bibnamefont
  {Temme}}, \bibinfo {author} {\bibfnamefont {D.~S.}\ \bibnamefont {Weiss}}, \bibinfo {author} {\bibfnamefont {J.}~\bibnamefont {Vu\ifmmode \check{c}\else \v{c}\fi{}kovi\ifmmode~\acute{c}\else \'{c}\fi{}}}, \bibinfo {author} {\bibfnamefont {V.}~\bibnamefont {Vuleti\ifmmode~\acute{c}\else \'{c}\fi{}}}, \bibinfo {author} {\bibfnamefont {J.}~\bibnamefont {Ye}},\ and\ \bibinfo {author} {\bibfnamefont {M.}~\bibnamefont {Zwierlein}},\ }\bibfield  {title} {\bibinfo {title} {Quantum simulators: Architectures and opportunities},\ }\href {https://doi.org/10.1103/PRXQuantum.2.017003} {\bibfield  {journal} {\bibinfo  {journal} {PRX Quantum}\ }\textbf {\bibinfo {volume} {2}},\ \bibinfo {pages} {017003} (\bibinfo {year} {2021})}\BibitemShut {NoStop}%
\bibitem [{\citenamefont {Daley}\ \emph {et~al.}(2022)\citenamefont {Daley}, \citenamefont {Bloch}, \citenamefont {Kokail}, \citenamefont {Flannigan}, \citenamefont {Pearson}, \citenamefont {Troyer},\ and\ \citenamefont {Zoller}}]{daley2022practical}%
  \BibitemOpen
  \bibfield  {author} {\bibinfo {author} {\bibfnamefont {A.~J.}\ \bibnamefont {Daley}}, \bibinfo {author} {\bibfnamefont {I.}~\bibnamefont {Bloch}}, \bibinfo {author} {\bibfnamefont {C.}~\bibnamefont {Kokail}}, \bibinfo {author} {\bibfnamefont {S.}~\bibnamefont {Flannigan}}, \bibinfo {author} {\bibfnamefont {N.}~\bibnamefont {Pearson}}, \bibinfo {author} {\bibfnamefont {M.}~\bibnamefont {Troyer}},\ and\ \bibinfo {author} {\bibfnamefont {P.}~\bibnamefont {Zoller}},\ }\bibfield  {title} {\bibinfo {title} {Practical quantum advantage in quantum simulation},\ }\href {https://doi.org/10.1038/s41586-022-04940-6} {\bibfield  {journal} {\bibinfo  {journal} {Nature}\ }\textbf {\bibinfo {volume} {607}},\ \bibinfo {pages} {667} (\bibinfo {year} {2022})}\BibitemShut {NoStop}%
\bibitem [{\citenamefont {Trivedi}\ \emph {et~al.}(2024)\citenamefont {Trivedi}, \citenamefont {Franco~Rubio},\ and\ \citenamefont {Cirac}}]{trivedi2024quantum-adv-stab-to-error-analog-qsim}%
  \BibitemOpen
  \bibfield  {author} {\bibinfo {author} {\bibfnamefont {R.}~\bibnamefont {Trivedi}}, \bibinfo {author} {\bibfnamefont {A.}~\bibnamefont {Franco~Rubio}},\ and\ \bibinfo {author} {\bibfnamefont {J.~I.}\ \bibnamefont {Cirac}},\ }\bibfield  {title} {\bibinfo {title} {Quantum advantage and stability to errors in analogue quantum simulators},\ }\href {https://doi.org/10.1038/s41467-024-50750-x} {\bibfield  {journal} {\bibinfo  {journal} {Nature Communications}\ }\textbf {\bibinfo {volume} {15}},\ \bibinfo {pages} {6507} (\bibinfo {year} {2024})}\BibitemShut {NoStop}%
\bibitem [{\citenamefont {Gottesman}(2010)}]{gottesman2010introduction}%
  \BibitemOpen
  \bibfield  {author} {\bibinfo {author} {\bibfnamefont {D.}~\bibnamefont {Gottesman}},\ }\bibfield  {title} {\bibinfo {title} {An introduction to quantum error correction and fault-tolerant quantum computation},\ }in\ \href@noop {} {\emph {\bibinfo {booktitle} {Quantum information science and its contributions to mathematics, Proceedings of Symposia in Applied Mathematics}}},\ Vol.~\bibinfo {volume} {68}\ (\bibinfo {year} {2010})\ pp.\ \bibinfo {pages} {13--58}\BibitemShut {NoStop}%
\bibitem [{\citenamefont {Kitaev}(2003)}]{kitaev2003fault}%
  \BibitemOpen
  \bibfield  {author} {\bibinfo {author} {\bibfnamefont {A.}~\bibnamefont {Kitaev}},\ }\bibfield  {title} {\bibinfo {title} {Fault-tolerant quantum computation by anyons},\ }\href {https://doi.org/https://doi.org/10.1016/S0003-4916(02)00018-0} {\bibfield  {journal} {\bibinfo  {journal} {Annals of Physics}\ }\textbf {\bibinfo {volume} {303}},\ \bibinfo {pages} {2} (\bibinfo {year} {2003})}\BibitemShut {NoStop}%
\bibitem [{\citenamefont {Dennis}\ \emph {et~al.}(2002)\citenamefont {Dennis}, \citenamefont {Kitaev}, \citenamefont {Landahl},\ and\ \citenamefont {Preskill}}]{dennis2002topological}%
  \BibitemOpen
  \bibfield  {author} {\bibinfo {author} {\bibfnamefont {E.}~\bibnamefont {Dennis}}, \bibinfo {author} {\bibfnamefont {A.}~\bibnamefont {Kitaev}}, \bibinfo {author} {\bibfnamefont {A.}~\bibnamefont {Landahl}},\ and\ \bibinfo {author} {\bibfnamefont {J.}~\bibnamefont {Preskill}},\ }\bibfield  {title} {\bibinfo {title} {{Topological quantum memory}},\ }\href {https://doi.org/10.1063/1.1499754} {\bibfield  {journal} {\bibinfo  {journal} {Journal of Mathematical Physics}\ }\textbf {\bibinfo {volume} {43}},\ \bibinfo {pages} {4452} (\bibinfo {year} {2002})}\BibitemShut {NoStop}%
\bibitem [{\citenamefont {Fowler}\ \emph {et~al.}(2009)\citenamefont {Fowler}, \citenamefont {Stephens},\ and\ \citenamefont {Groszkowski}}]{fowler2009high}%
  \BibitemOpen
  \bibfield  {author} {\bibinfo {author} {\bibfnamefont {A.~G.}\ \bibnamefont {Fowler}}, \bibinfo {author} {\bibfnamefont {A.~M.}\ \bibnamefont {Stephens}},\ and\ \bibinfo {author} {\bibfnamefont {P.}~\bibnamefont {Groszkowski}},\ }\bibfield  {title} {\bibinfo {title} {High-threshold universal quantum computation on the surface code},\ }\href {https://doi.org/10.1103/PhysRevA.80.052312} {\bibfield  {journal} {\bibinfo  {journal} {Phys. Rev. A}\ }\textbf {\bibinfo {volume} {80}},\ \bibinfo {pages} {052312} (\bibinfo {year} {2009})}\BibitemShut {NoStop}%
\bibitem [{\citenamefont {Takita}\ \emph {et~al.}(2016)\citenamefont {Takita}, \citenamefont {C\'orcoles}, \citenamefont {Magesan}, \citenamefont {Abdo}, \citenamefont {Brink}, \citenamefont {Cross}, \citenamefont {Chow},\ and\ \citenamefont {Gambetta}}]{takita2016demonstration}%
  \BibitemOpen
  \bibfield  {author} {\bibinfo {author} {\bibfnamefont {M.}~\bibnamefont {Takita}}, \bibinfo {author} {\bibfnamefont {A.~D.}\ \bibnamefont {C\'orcoles}}, \bibinfo {author} {\bibfnamefont {E.}~\bibnamefont {Magesan}}, \bibinfo {author} {\bibfnamefont {B.}~\bibnamefont {Abdo}}, \bibinfo {author} {\bibfnamefont {M.}~\bibnamefont {Brink}}, \bibinfo {author} {\bibfnamefont {A.}~\bibnamefont {Cross}}, \bibinfo {author} {\bibfnamefont {J.~M.}\ \bibnamefont {Chow}},\ and\ \bibinfo {author} {\bibfnamefont {J.~M.}\ \bibnamefont {Gambetta}},\ }\bibfield  {title} {\bibinfo {title} {Demonstration of weight-four parity measurements in the surface code architecture},\ }\href {https://doi.org/10.1103/PhysRevLett.117.210505} {\bibfield  {journal} {\bibinfo  {journal} {Phys. Rev. Lett.}\ }\textbf {\bibinfo {volume} {117}},\ \bibinfo {pages} {210505} (\bibinfo {year} {2016})}\BibitemShut {NoStop}%
\bibitem [{\citenamefont {Marques}\ \emph {et~al.}(2022)\citenamefont {Marques}, \citenamefont {Varbanov}, \citenamefont {Moreira}, \citenamefont {Ali}, \citenamefont {Muthusubramanian}, \citenamefont {Zachariadis}, \citenamefont {Battistel}, \citenamefont {Beekman}, \citenamefont {Haider}, \citenamefont {Vlothuizen} \emph {et~al.}}]{marques2022logical}%
  \BibitemOpen
  \bibfield  {author} {\bibinfo {author} {\bibfnamefont {J.~F.}\ \bibnamefont {Marques}}, \bibinfo {author} {\bibfnamefont {B.}~\bibnamefont {Varbanov}}, \bibinfo {author} {\bibfnamefont {M.}~\bibnamefont {Moreira}}, \bibinfo {author} {\bibfnamefont {H.}~\bibnamefont {Ali}}, \bibinfo {author} {\bibfnamefont {N.}~\bibnamefont {Muthusubramanian}}, \bibinfo {author} {\bibfnamefont {C.}~\bibnamefont {Zachariadis}}, \bibinfo {author} {\bibfnamefont {F.}~\bibnamefont {Battistel}}, \bibinfo {author} {\bibfnamefont {M.}~\bibnamefont {Beekman}}, \bibinfo {author} {\bibfnamefont {N.}~\bibnamefont {Haider}}, \bibinfo {author} {\bibfnamefont {W.}~\bibnamefont {Vlothuizen}}, \emph {et~al.},\ }\bibfield  {title} {\bibinfo {title} {Logical-qubit operations in an error-detecting surface code},\ }\href {https://doi.org/10.1038/s41567-021-01423-9} {\bibfield  {journal} {\bibinfo  {journal} {Nature Physics}\ }\textbf {\bibinfo {volume} {18}},\ \bibinfo {pages} {80} (\bibinfo {year} {2022})}\BibitemShut {NoStop}%
\bibitem [{\citenamefont {Krinner}\ \emph {et~al.}(2022)\citenamefont {Krinner}, \citenamefont {Lacroix}, \citenamefont {Remm}, \citenamefont {Di~Paolo}, \citenamefont {Genois}, \citenamefont {Leroux}, \citenamefont {Hellings}, \citenamefont {Lazar}, \citenamefont {Swiadek}, \citenamefont {Herrmann} \emph {et~al.}}]{krinner2022realizing}%
  \BibitemOpen
  \bibfield  {author} {\bibinfo {author} {\bibfnamefont {S.}~\bibnamefont {Krinner}}, \bibinfo {author} {\bibfnamefont {N.}~\bibnamefont {Lacroix}}, \bibinfo {author} {\bibfnamefont {A.}~\bibnamefont {Remm}}, \bibinfo {author} {\bibfnamefont {A.}~\bibnamefont {Di~Paolo}}, \bibinfo {author} {\bibfnamefont {E.}~\bibnamefont {Genois}}, \bibinfo {author} {\bibfnamefont {C.}~\bibnamefont {Leroux}}, \bibinfo {author} {\bibfnamefont {C.}~\bibnamefont {Hellings}}, \bibinfo {author} {\bibfnamefont {S.}~\bibnamefont {Lazar}}, \bibinfo {author} {\bibfnamefont {F.}~\bibnamefont {Swiadek}}, \bibinfo {author} {\bibfnamefont {J.}~\bibnamefont {Herrmann}}, \emph {et~al.},\ }\bibfield  {title} {\bibinfo {title} {Realizing repeated quantum error correction in a distance-three surface code},\ }\href {https://doi.org/10.1038/s41586-022-04566-8} {\bibfield  {journal} {\bibinfo  {journal} {Nature}\ }\textbf {\bibinfo {volume} {605}},\ \bibinfo {pages} {669} (\bibinfo {year} {2022})}\BibitemShut {NoStop}%
\bibitem [{\citenamefont {Zhao}\ \emph {et~al.}(2022)\citenamefont {Zhao}, \citenamefont {Ye}, \citenamefont {Huang}, \citenamefont {Zhang}, \citenamefont {Wu}, \citenamefont {Guan}, \citenamefont {Zhu}, \citenamefont {Wei}, \citenamefont {He}, \citenamefont {Cao}, \citenamefont {Chen}, \citenamefont {Chung}, \citenamefont {Deng}, \citenamefont {Fan}, \citenamefont {Gong}, \citenamefont {Guo}, \citenamefont {Guo}, \citenamefont {Han}, \citenamefont {Li}, \citenamefont {Li}, \citenamefont {Li}, \citenamefont {Liang}, \citenamefont {Lin}, \citenamefont {Qian}, \citenamefont {Rong}, \citenamefont {Su}, \citenamefont {Sun}, \citenamefont {Wang}, \citenamefont {Wu}, \citenamefont {Xu}, \citenamefont {Ying}, \citenamefont {Yu}, \citenamefont {Zha}, \citenamefont {Zhang}, \citenamefont {Huo}, \citenamefont {Lu}, \citenamefont {Peng}, \citenamefont {Zhu},\ and\ \citenamefont {Pan}}]{zhao2022realization}%
  \BibitemOpen
  \bibfield  {author} {\bibinfo {author} {\bibfnamefont {Y.}~\bibnamefont {Zhao}}, \bibinfo {author} {\bibfnamefont {Y.}~\bibnamefont {Ye}}, \bibinfo {author} {\bibfnamefont {H.-L.}\ \bibnamefont {Huang}}, \bibinfo {author} {\bibfnamefont {Y.}~\bibnamefont {Zhang}}, \bibinfo {author} {\bibfnamefont {D.}~\bibnamefont {Wu}}, \bibinfo {author} {\bibfnamefont {H.}~\bibnamefont {Guan}}, \bibinfo {author} {\bibfnamefont {Q.}~\bibnamefont {Zhu}}, \bibinfo {author} {\bibfnamefont {Z.}~\bibnamefont {Wei}}, \bibinfo {author} {\bibfnamefont {T.}~\bibnamefont {He}}, \bibinfo {author} {\bibfnamefont {S.}~\bibnamefont {Cao}}, \bibinfo {author} {\bibfnamefont {F.}~\bibnamefont {Chen}}, \bibinfo {author} {\bibfnamefont {T.-H.}\ \bibnamefont {Chung}}, \bibinfo {author} {\bibfnamefont {H.}~\bibnamefont {Deng}}, \bibinfo {author} {\bibfnamefont {D.}~\bibnamefont {Fan}}, \bibinfo {author} {\bibfnamefont {M.}~\bibnamefont {Gong}}, \bibinfo {author} {\bibfnamefont {C.}~\bibnamefont {Guo}}, \bibinfo {author} {\bibfnamefont
  {S.}~\bibnamefont {Guo}}, \bibinfo {author} {\bibfnamefont {L.}~\bibnamefont {Han}}, \bibinfo {author} {\bibfnamefont {N.}~\bibnamefont {Li}}, \bibinfo {author} {\bibfnamefont {S.}~\bibnamefont {Li}}, \bibinfo {author} {\bibfnamefont {Y.}~\bibnamefont {Li}}, \bibinfo {author} {\bibfnamefont {F.}~\bibnamefont {Liang}}, \bibinfo {author} {\bibfnamefont {J.}~\bibnamefont {Lin}}, \bibinfo {author} {\bibfnamefont {H.}~\bibnamefont {Qian}}, \bibinfo {author} {\bibfnamefont {H.}~\bibnamefont {Rong}}, \bibinfo {author} {\bibfnamefont {H.}~\bibnamefont {Su}}, \bibinfo {author} {\bibfnamefont {L.}~\bibnamefont {Sun}}, \bibinfo {author} {\bibfnamefont {S.}~\bibnamefont {Wang}}, \bibinfo {author} {\bibfnamefont {Y.}~\bibnamefont {Wu}}, \bibinfo {author} {\bibfnamefont {Y.}~\bibnamefont {Xu}}, \bibinfo {author} {\bibfnamefont {C.}~\bibnamefont {Ying}}, \bibinfo {author} {\bibfnamefont {J.}~\bibnamefont {Yu}}, \bibinfo {author} {\bibfnamefont {C.}~\bibnamefont {Zha}}, \bibinfo {author} {\bibfnamefont {K.}~\bibnamefont
  {Zhang}}, \bibinfo {author} {\bibfnamefont {Y.-H.}\ \bibnamefont {Huo}}, \bibinfo {author} {\bibfnamefont {C.-Y.}\ \bibnamefont {Lu}}, \bibinfo {author} {\bibfnamefont {C.-Z.}\ \bibnamefont {Peng}}, \bibinfo {author} {\bibfnamefont {X.}~\bibnamefont {Zhu}},\ and\ \bibinfo {author} {\bibfnamefont {J.-W.}\ \bibnamefont {Pan}},\ }\bibfield  {title} {\bibinfo {title} {Realization of an error-correcting surface code with superconducting qubits},\ }\href {https://doi.org/10.1103/PhysRevLett.129.030501} {\bibfield  {journal} {\bibinfo  {journal} {Phys. Rev. Lett.}\ }\textbf {\bibinfo {volume} {129}},\ \bibinfo {pages} {030501} (\bibinfo {year} {2022})}\BibitemShut {NoStop}%
\bibitem [{\citenamefont {Bluvstein}\ \emph {et~al.}(2023)\citenamefont {Bluvstein}, \citenamefont {Evered}, \citenamefont {Geim}, \citenamefont {Li}, \citenamefont {Zhou}, \citenamefont {Manovitz}, \citenamefont {Ebadi}, \citenamefont {Cain}, \citenamefont {Kalinowski}, \citenamefont {Hangleiter}, \citenamefont {Bonilla~Ataides}, \citenamefont {Maskara}, \citenamefont {Cong}, \citenamefont {Gao}, \citenamefont {Sales~Rodriguez}, \citenamefont {Karolyshyn}, \citenamefont {Semeghini}, \citenamefont {Gullans}, \citenamefont {Greiner}, \citenamefont {Vuletić},\ and\ \citenamefont {Lukin}}]{harvard-qec-Bluvstein_2023}%
  \BibitemOpen
  \bibfield  {author} {\bibinfo {author} {\bibfnamefont {D.}~\bibnamefont {Bluvstein}}, \bibinfo {author} {\bibfnamefont {S.~J.}\ \bibnamefont {Evered}}, \bibinfo {author} {\bibfnamefont {A.~A.}\ \bibnamefont {Geim}}, \bibinfo {author} {\bibfnamefont {S.~H.}\ \bibnamefont {Li}}, \bibinfo {author} {\bibfnamefont {H.}~\bibnamefont {Zhou}}, \bibinfo {author} {\bibfnamefont {T.}~\bibnamefont {Manovitz}}, \bibinfo {author} {\bibfnamefont {S.}~\bibnamefont {Ebadi}}, \bibinfo {author} {\bibfnamefont {M.}~\bibnamefont {Cain}}, \bibinfo {author} {\bibfnamefont {M.}~\bibnamefont {Kalinowski}}, \bibinfo {author} {\bibfnamefont {D.}~\bibnamefont {Hangleiter}}, \bibinfo {author} {\bibfnamefont {J.~P.}\ \bibnamefont {Bonilla~Ataides}}, \bibinfo {author} {\bibfnamefont {N.}~\bibnamefont {Maskara}}, \bibinfo {author} {\bibfnamefont {I.}~\bibnamefont {Cong}}, \bibinfo {author} {\bibfnamefont {X.}~\bibnamefont {Gao}}, \bibinfo {author} {\bibfnamefont {P.}~\bibnamefont {Sales~Rodriguez}}, \bibinfo {author} {\bibfnamefont
  {T.}~\bibnamefont {Karolyshyn}}, \bibinfo {author} {\bibfnamefont {G.}~\bibnamefont {Semeghini}}, \bibinfo {author} {\bibfnamefont {M.~J.}\ \bibnamefont {Gullans}}, \bibinfo {author} {\bibfnamefont {M.}~\bibnamefont {Greiner}}, \bibinfo {author} {\bibfnamefont {V.}~\bibnamefont {Vuletić}},\ and\ \bibinfo {author} {\bibfnamefont {M.~D.}\ \bibnamefont {Lukin}},\ }\bibfield  {title} {\bibinfo {title} {Logical quantum processor based on reconfigurable atom arrays},\ }\href {https://doi.org/10.1038/s41586-023-06927-3} {\bibfield  {journal} {\bibinfo  {journal} {Nature}\ }\textbf {\bibinfo {volume} {626}},\ \bibinfo {pages} {58–65} (\bibinfo {year} {2023})}\BibitemShut {NoStop}%
\bibitem [{\citenamefont {Da~Silva}\ \emph {et~al.}(2024)\citenamefont {Da~Silva}, \citenamefont {Ryan-Anderson}, \citenamefont {Bello-Rivas}, \citenamefont {Chernoguzov}, \citenamefont {Dreiling}, \citenamefont {Foltz}, \citenamefont {Gaebler}, \citenamefont {Gatterman}, \citenamefont {Hayes}, \citenamefont {Hewitt} \emph {et~al.}}]{quantinuum-qec1-da2024demonstration}%
  \BibitemOpen
  \bibfield  {author} {\bibinfo {author} {\bibfnamefont {M.}~\bibnamefont {Da~Silva}}, \bibinfo {author} {\bibfnamefont {C.}~\bibnamefont {Ryan-Anderson}}, \bibinfo {author} {\bibfnamefont {J.}~\bibnamefont {Bello-Rivas}}, \bibinfo {author} {\bibfnamefont {A.}~\bibnamefont {Chernoguzov}}, \bibinfo {author} {\bibfnamefont {J.}~\bibnamefont {Dreiling}}, \bibinfo {author} {\bibfnamefont {C.}~\bibnamefont {Foltz}}, \bibinfo {author} {\bibfnamefont {J.}~\bibnamefont {Gaebler}}, \bibinfo {author} {\bibfnamefont {T.}~\bibnamefont {Gatterman}}, \bibinfo {author} {\bibfnamefont {D.}~\bibnamefont {Hayes}}, \bibinfo {author} {\bibfnamefont {N.}~\bibnamefont {Hewitt}}, \emph {et~al.},\ }\bibfield  {title} {\bibinfo {title} {Demonstration of logical qubits and repeated error correction with better-than-physical error rates},\ }\href {https://arxiv.org/abs/2404.02280} {\bibfield  {journal} {\bibinfo  {journal} {arXiv preprint arXiv:2404.02280}\ } (\bibinfo {year} {2024})}\BibitemShut {NoStop}%
\bibitem [{\citenamefont {Ryan-Anderson}\ \emph {et~al.}(2024)\citenamefont {Ryan-Anderson}, \citenamefont {Brown}, \citenamefont {Baldwin}, \citenamefont {Dreiling}, \citenamefont {Foltz}, \citenamefont {Gaebler}, \citenamefont {Gatterman}, \citenamefont {Hewitt}, \citenamefont {Holliman}, \citenamefont {Horst} \emph {et~al.}}]{quantinuum-qec2-ryan2024high}%
  \BibitemOpen
  \bibfield  {author} {\bibinfo {author} {\bibfnamefont {C.}~\bibnamefont {Ryan-Anderson}}, \bibinfo {author} {\bibfnamefont {N.}~\bibnamefont {Brown}}, \bibinfo {author} {\bibfnamefont {C.}~\bibnamefont {Baldwin}}, \bibinfo {author} {\bibfnamefont {J.}~\bibnamefont {Dreiling}}, \bibinfo {author} {\bibfnamefont {C.}~\bibnamefont {Foltz}}, \bibinfo {author} {\bibfnamefont {J.}~\bibnamefont {Gaebler}}, \bibinfo {author} {\bibfnamefont {T.}~\bibnamefont {Gatterman}}, \bibinfo {author} {\bibfnamefont {N.}~\bibnamefont {Hewitt}}, \bibinfo {author} {\bibfnamefont {C.}~\bibnamefont {Holliman}}, \bibinfo {author} {\bibfnamefont {C.}~\bibnamefont {Horst}}, \emph {et~al.},\ }\bibfield  {title} {\bibinfo {title} {High-fidelity and fault-tolerant teleportation of a logical qubit using transversal gates and lattice surgery on a trapped-ion quantum computer},\ }\href {https://arxiv.org/abs/2404.16728} {\bibfield  {journal} {\bibinfo  {journal} {arXiv preprint arXiv:2404.16728}\ } (\bibinfo {year} {2024})}\BibitemShut
  {NoStop}%
\bibitem [{\citenamefont {Acharya}\ \emph {et~al.}(2023)\citenamefont {Acharya}, \citenamefont {Aleiner}, \citenamefont {Allen}, \citenamefont {Andersen}, \citenamefont {Ansmann}, \citenamefont {Arute}, \citenamefont {Arya}, \citenamefont {Asfaw}, \citenamefont {Atalaya}, \citenamefont {Babbush} \emph {et~al.}}]{google-qec-2023}%
  \BibitemOpen
  \bibfield  {author} {\bibinfo {author} {\bibfnamefont {R.}~\bibnamefont {Acharya}}, \bibinfo {author} {\bibfnamefont {I.}~\bibnamefont {Aleiner}}, \bibinfo {author} {\bibfnamefont {R.}~\bibnamefont {Allen}}, \bibinfo {author} {\bibfnamefont {T.~I.}\ \bibnamefont {Andersen}}, \bibinfo {author} {\bibfnamefont {M.}~\bibnamefont {Ansmann}}, \bibinfo {author} {\bibfnamefont {F.}~\bibnamefont {Arute}}, \bibinfo {author} {\bibfnamefont {K.}~\bibnamefont {Arya}}, \bibinfo {author} {\bibfnamefont {A.}~\bibnamefont {Asfaw}}, \bibinfo {author} {\bibfnamefont {J.}~\bibnamefont {Atalaya}}, \bibinfo {author} {\bibfnamefont {R.}~\bibnamefont {Babbush}}, \emph {et~al.},\ }\bibfield  {title} {\bibinfo {title} {Suppressing quantum errors by scaling a surface code logical qubit},\ }\href {https://doi.org/10.1038/s41586-022-05434-1} {\bibfield  {journal} {\bibinfo  {journal} {Nature}\ }\textbf {\bibinfo {volume} {614}},\ \bibinfo {pages} {676–681} (\bibinfo {year} {2023})}\BibitemShut {NoStop}%
\bibitem [{\citenamefont {Self}\ \emph {et~al.}(2024{\natexlab{a}})\citenamefont {Self}, \citenamefont {Benedetti},\ and\ \citenamefont {Amaro}}]{eth-qec-Self_2024}%
  \BibitemOpen
  \bibfield  {author} {\bibinfo {author} {\bibfnamefont {C.~N.}\ \bibnamefont {Self}}, \bibinfo {author} {\bibfnamefont {M.}~\bibnamefont {Benedetti}},\ and\ \bibinfo {author} {\bibfnamefont {D.}~\bibnamefont {Amaro}},\ }\bibfield  {title} {\bibinfo {title} {Protecting expressive circuits with a quantum error detection code},\ }\href {https://doi.org/10.1038/s41567-023-02282-2} {\bibfield  {journal} {\bibinfo  {journal} {Nature Physics}\ }\textbf {\bibinfo {volume} {20}},\ \bibinfo {pages} {219–224} (\bibinfo {year} {2024}{\natexlab{a}})}\BibitemShut {NoStop}%
\bibitem [{\citenamefont {Bacon}(2006)}]{bacon-shor-code}%
  \BibitemOpen
  \bibfield  {author} {\bibinfo {author} {\bibfnamefont {D.}~\bibnamefont {Bacon}},\ }\bibfield  {title} {\bibinfo {title} {Operator quantum error-correcting subsystems for self-correcting quantum memories},\ }\href {https://doi.org/10.1103/PhysRevA.73.012340} {\bibfield  {journal} {\bibinfo  {journal} {Phys. Rev. A}\ }\textbf {\bibinfo {volume} {73}},\ \bibinfo {pages} {012340} (\bibinfo {year} {2006})}\BibitemShut {NoStop}%
\bibitem [{\citenamefont {Jordan}\ \emph {et~al.}(2006)\citenamefont {Jordan}, \citenamefont {Farhi},\ and\ \citenamefont {Shor}}]{JFS06-PRA}%
  \BibitemOpen
  \bibfield  {author} {\bibinfo {author} {\bibfnamefont {S.~P.}\ \bibnamefont {Jordan}}, \bibinfo {author} {\bibfnamefont {E.}~\bibnamefont {Farhi}},\ and\ \bibinfo {author} {\bibfnamefont {P.~W.}\ \bibnamefont {Shor}},\ }\bibfield  {title} {\bibinfo {title} {Error-correcting codes for adiabatic quantum computation},\ }\href {https://doi.org/10.1103/PhysRevA.74.052322} {\bibfield  {journal} {\bibinfo  {journal} {Phys. Rev. A}\ }\textbf {\bibinfo {volume} {74}},\ \bibinfo {pages} {052322} (\bibinfo {year} {2006})}\BibitemShut {NoStop}%
\bibitem [{\citenamefont {Young}\ \emph {et~al.}(2013)\citenamefont {Young}, \citenamefont {Sarovar},\ and\ \citenamefont {Blume-Kohout}}]{kevin-young2013error}%
  \BibitemOpen
  \bibfield  {author} {\bibinfo {author} {\bibfnamefont {K.~C.}\ \bibnamefont {Young}}, \bibinfo {author} {\bibfnamefont {M.}~\bibnamefont {Sarovar}},\ and\ \bibinfo {author} {\bibfnamefont {R.}~\bibnamefont {Blume-Kohout}},\ }\bibfield  {title} {\bibinfo {title} {Error suppression and error correction in adiabatic quantum computation: Techniques and challenges},\ }\href {https://doi.org/10.1103/PhysRevX.3.041013} {\bibfield  {journal} {\bibinfo  {journal} {Phys. Rev. X}\ }\textbf {\bibinfo {volume} {3}},\ \bibinfo {pages} {041013} (\bibinfo {year} {2013})}\BibitemShut {NoStop}%
\bibitem [{\citenamefont {Pudenz}\ \emph {et~al.}(2014)\citenamefont {Pudenz}, \citenamefont {Albash},\ and\ \citenamefont {Lidar}}]{Pudenz_2014}%
  \BibitemOpen
  \bibfield  {author} {\bibinfo {author} {\bibfnamefont {K.~L.}\ \bibnamefont {Pudenz}}, \bibinfo {author} {\bibfnamefont {T.}~\bibnamefont {Albash}},\ and\ \bibinfo {author} {\bibfnamefont {D.~A.}\ \bibnamefont {Lidar}},\ }\bibfield  {title} {\bibinfo {title} {Error-corrected quantum annealing with hundreds of qubits},\ }\href {https://doi.org/10.1038/ncomms4243} {\bibfield  {journal} {\bibinfo  {journal} {Nature Communications}\ }\textbf {\bibinfo {volume} {5}},\ \bibinfo {pages} {3243} (\bibinfo {year} {2014})}\BibitemShut {NoStop}%
\bibitem [{\citenamefont {Marvian}\ and\ \citenamefont {Lidar}(2014)}]{Marvian-Lidar-2014-PRL-2local-is-too-local}%
  \BibitemOpen
  \bibfield  {author} {\bibinfo {author} {\bibfnamefont {I.}~\bibnamefont {Marvian}}\ and\ \bibinfo {author} {\bibfnamefont {D.~A.}\ \bibnamefont {Lidar}},\ }\bibfield  {title} {\bibinfo {title} {Quantum error suppression with commuting hamiltonians: Two local is too local},\ }\href {https://doi.org/10.1103/PhysRevLett.113.260504} {\bibfield  {journal} {\bibinfo  {journal} {Phys. Rev. Lett.}\ }\textbf {\bibinfo {volume} {113}},\ \bibinfo {pages} {260504} (\bibinfo {year} {2014})}\BibitemShut {NoStop}%
\bibitem [{\citenamefont {Bookatz}\ \emph {et~al.}(2015)\citenamefont {Bookatz}, \citenamefont {Farhi},\ and\ \citenamefont {Zhou}}]{bookatz2015error}%
  \BibitemOpen
  \bibfield  {author} {\bibinfo {author} {\bibfnamefont {A.~D.}\ \bibnamefont {Bookatz}}, \bibinfo {author} {\bibfnamefont {E.}~\bibnamefont {Farhi}},\ and\ \bibinfo {author} {\bibfnamefont {L.}~\bibnamefont {Zhou}},\ }\bibfield  {title} {\bibinfo {title} {Error suppression in hamiltonian-based quantum computation using energy penalties},\ }\href {https://doi.org/10.1103/PhysRevA.92.022317} {\bibfield  {journal} {\bibinfo  {journal} {Phys. Rev. A}\ }\textbf {\bibinfo {volume} {92}},\ \bibinfo {pages} {022317} (\bibinfo {year} {2015})}\BibitemShut {NoStop}%
\bibitem [{\citenamefont {Matsuura}\ \emph {et~al.}(2016)\citenamefont {Matsuura}, \citenamefont {Nishimori}, \citenamefont {Albash},\ and\ \citenamefont {Lidar}}]{Matsuura_2016}%
  \BibitemOpen
  \bibfield  {author} {\bibinfo {author} {\bibfnamefont {S.}~\bibnamefont {Matsuura}}, \bibinfo {author} {\bibfnamefont {H.}~\bibnamefont {Nishimori}}, \bibinfo {author} {\bibfnamefont {T.}~\bibnamefont {Albash}},\ and\ \bibinfo {author} {\bibfnamefont {D.~A.}\ \bibnamefont {Lidar}},\ }\bibfield  {title} {\bibinfo {title} {Mean field analysis of quantum annealing correction},\ }\href {https://doi.org/10.1103/PhysRevLett.116.220501} {\bibfield  {journal} {\bibinfo  {journal} {Phys. Rev. Lett.}\ }\textbf {\bibinfo {volume} {116}},\ \bibinfo {pages} {220501} (\bibinfo {year} {2016})}\BibitemShut {NoStop}%
\bibitem [{\citenamefont {Vinci}\ \emph {et~al.}(2016)\citenamefont {Vinci}, \citenamefont {Albash},\ and\ \citenamefont {Lidar}}]{Vinci_2016}%
  \BibitemOpen
  \bibfield  {author} {\bibinfo {author} {\bibfnamefont {W.}~\bibnamefont {Vinci}}, \bibinfo {author} {\bibfnamefont {T.}~\bibnamefont {Albash}},\ and\ \bibinfo {author} {\bibfnamefont {D.~A.}\ \bibnamefont {Lidar}},\ }\bibfield  {title} {\bibinfo {title} {Nested quantum annealing correction},\ }\href {https://doi.org/10.1038/npjqi.2016.17} {\bibfield  {journal} {\bibinfo  {journal} {npj Quantum Information}\ }\textbf {\bibinfo {volume} {2}},\ \bibinfo {pages} {16017} (\bibinfo {year} {2016})}\BibitemShut {NoStop}%
\bibitem [{\citenamefont {Marvian}\ and\ \citenamefont {Lidar}(2017{\natexlab{a}})}]{marvian2017error}%
  \BibitemOpen
  \bibfield  {author} {\bibinfo {author} {\bibfnamefont {M.}~\bibnamefont {Marvian}}\ and\ \bibinfo {author} {\bibfnamefont {D.~A.}\ \bibnamefont {Lidar}},\ }\bibfield  {title} {\bibinfo {title} {Error suppression for hamiltonian quantum computing in markovian environments},\ }\href {https://doi.org/10.1103/PhysRevA.95.032302} {\bibfield  {journal} {\bibinfo  {journal} {Phys. Rev. A}\ }\textbf {\bibinfo {volume} {95}},\ \bibinfo {pages} {032302} (\bibinfo {year} {2017}{\natexlab{a}})}\BibitemShut {NoStop}%
\bibitem [{\citenamefont {Marvian}\ and\ \citenamefont {Lidar}(2017{\natexlab{b}})}]{marvian2017-PRL-error-supp-subsystem-codes}%
  \BibitemOpen
  \bibfield  {author} {\bibinfo {author} {\bibfnamefont {M.}~\bibnamefont {Marvian}}\ and\ \bibinfo {author} {\bibfnamefont {D.~A.}\ \bibnamefont {Lidar}},\ }\bibfield  {title} {\bibinfo {title} {Error suppression for hamiltonian-based quantum computation using subsystem codes},\ }\href {https://doi.org/10.1103/PhysRevLett.118.030504} {\bibfield  {journal} {\bibinfo  {journal} {Phys. Rev. Lett.}\ }\textbf {\bibinfo {volume} {118}},\ \bibinfo {pages} {030504} (\bibinfo {year} {2017}{\natexlab{b}})}\BibitemShut {NoStop}%
\bibitem [{\citenamefont {Marvian}\ and\ \citenamefont {Lloyd}(2019)}]{marvian2019-arxiv-robust}%
  \BibitemOpen
  \bibfield  {author} {\bibinfo {author} {\bibfnamefont {M.}~\bibnamefont {Marvian}}\ and\ \bibinfo {author} {\bibfnamefont {S.}~\bibnamefont {Lloyd}},\ }\bibfield  {title} {\bibinfo {title} {Robust universal hamiltonian quantum computing using two-body interactions},\ }\href {https://arxiv.org/abs/1911.01354} {\bibfield  {journal} {\bibinfo  {journal} {arXiv preprint arXiv:1911.01354}\ } (\bibinfo {year} {2019})}\BibitemShut {NoStop}%
\bibitem [{\citenamefont {Pearson}\ \emph {et~al.}(2019)\citenamefont {Pearson}, \citenamefont {Mishra}, \citenamefont {Hen},\ and\ \citenamefont {Lidar}}]{Pearson_2019}%
  \BibitemOpen
  \bibfield  {author} {\bibinfo {author} {\bibfnamefont {A.}~\bibnamefont {Pearson}}, \bibinfo {author} {\bibfnamefont {A.}~\bibnamefont {Mishra}}, \bibinfo {author} {\bibfnamefont {I.}~\bibnamefont {Hen}},\ and\ \bibinfo {author} {\bibfnamefont {D.~A.}\ \bibnamefont {Lidar}},\ }\bibfield  {title} {\bibinfo {title} {Analog errors in quantum annealing: doom and hope},\ }\href {https://doi.org/10.1038/s41534-019-0210-7} {\bibfield  {journal} {\bibinfo  {journal} {npj Quantum Information}\ }\textbf {\bibinfo {volume} {5}},\ \bibinfo {pages} {107} (\bibinfo {year} {2019})}\BibitemShut {NoStop}%
\bibitem [{\citenamefont {Singkanipa}\ \emph {et~al.}()\citenamefont {Singkanipa}, \citenamefont {Xia},\ and\ \citenamefont {Lidar}}]{Xia_2024}%
  \BibitemOpen
  \bibfield  {author} {\bibinfo {author} {\bibfnamefont {P.}~\bibnamefont {Singkanipa}}, \bibinfo {author} {\bibfnamefont {Z.}~\bibnamefont {Xia}},\ and\ \bibinfo {author} {\bibfnamefont {D.~A.}\ \bibnamefont {Lidar}},\ }\href@noop {} {\bibinfo {title} {Families of $d=2$ 2d subsystem stabilizer codes for universal hamiltonian quantum computation with two-body interactions}},\ \Eprint {https://arxiv.org/abs/2412.06744} {arXiv:2412.06744} \BibitemShut {NoStop}%
\bibitem [{\citenamefont {Karamlou}\ \emph {et~al.}(2024)\citenamefont {Karamlou}, \citenamefont {Rosen}, \citenamefont {Muschinske}, \citenamefont {Barrett}, \citenamefont {Di~Paolo}, \citenamefont {Ding}, \citenamefont {Harrington}, \citenamefont {Hays}, \citenamefont {Das}, \citenamefont {Kim} \emph {et~al.}}]{oliver-2024probing-2D-hubbard}%
  \BibitemOpen
  \bibfield  {author} {\bibinfo {author} {\bibfnamefont {A.~H.}\ \bibnamefont {Karamlou}}, \bibinfo {author} {\bibfnamefont {I.~T.}\ \bibnamefont {Rosen}}, \bibinfo {author} {\bibfnamefont {S.~E.}\ \bibnamefont {Muschinske}}, \bibinfo {author} {\bibfnamefont {C.~N.}\ \bibnamefont {Barrett}}, \bibinfo {author} {\bibfnamefont {A.}~\bibnamefont {Di~Paolo}}, \bibinfo {author} {\bibfnamefont {L.}~\bibnamefont {Ding}}, \bibinfo {author} {\bibfnamefont {P.~M.}\ \bibnamefont {Harrington}}, \bibinfo {author} {\bibfnamefont {M.}~\bibnamefont {Hays}}, \bibinfo {author} {\bibfnamefont {R.}~\bibnamefont {Das}}, \bibinfo {author} {\bibfnamefont {D.~K.}\ \bibnamefont {Kim}}, \emph {et~al.},\ }\bibfield  {title} {\bibinfo {title} {Probing entanglement in a 2d hard-core bose--hubbard lattice},\ }\href {https://doi.org/10.1038/s41586-024-07325-z} {\bibfield  {journal} {\bibinfo  {journal} {Nature}\ ,\ \bibinfo {pages} {1}} (\bibinfo {year} {2024})}\BibitemShut {NoStop}%
\bibitem [{\citenamefont {Wienand}\ \emph {et~al.}(2024)\citenamefont {Wienand}, \citenamefont {Karch}, \citenamefont {Impertro}, \citenamefont {Schweizer}, \citenamefont {McCulloch}, \citenamefont {Vasseur}, \citenamefont {Gopalakrishnan}, \citenamefont {Aidelsburger},\ and\ \citenamefont {Bloch}}]{immanuel-2024emergence}%
  \BibitemOpen
  \bibfield  {author} {\bibinfo {author} {\bibfnamefont {J.~F.}\ \bibnamefont {Wienand}}, \bibinfo {author} {\bibfnamefont {S.}~\bibnamefont {Karch}}, \bibinfo {author} {\bibfnamefont {A.}~\bibnamefont {Impertro}}, \bibinfo {author} {\bibfnamefont {C.}~\bibnamefont {Schweizer}}, \bibinfo {author} {\bibfnamefont {E.}~\bibnamefont {McCulloch}}, \bibinfo {author} {\bibfnamefont {R.}~\bibnamefont {Vasseur}}, \bibinfo {author} {\bibfnamefont {S.}~\bibnamefont {Gopalakrishnan}}, \bibinfo {author} {\bibfnamefont {M.}~\bibnamefont {Aidelsburger}},\ and\ \bibinfo {author} {\bibfnamefont {I.}~\bibnamefont {Bloch}},\ }\bibfield  {title} {\bibinfo {title} {Emergence of fluctuating hydrodynamics in chaotic quantum systems},\ }\href {https://doi.org/10.1038/s41567-024-02611-z} {\bibfield  {journal} {\bibinfo  {journal} {Nature Physics}\ ,\ \bibinfo {pages} {1}} (\bibinfo {year} {2024})}\BibitemShut {NoStop}%
\bibitem [{\citenamefont {Feng}\ \emph {et~al.}(2023)\citenamefont {Feng}, \citenamefont {Katz}, \citenamefont {Haack}, \citenamefont {Maghrebi}, \citenamefont {Gorshkov}, \citenamefont {Gong}, \citenamefont {Cetina},\ and\ \citenamefont {Monroe}}]{monroe-2023continuous}%
  \BibitemOpen
  \bibfield  {author} {\bibinfo {author} {\bibfnamefont {L.}~\bibnamefont {Feng}}, \bibinfo {author} {\bibfnamefont {O.}~\bibnamefont {Katz}}, \bibinfo {author} {\bibfnamefont {C.}~\bibnamefont {Haack}}, \bibinfo {author} {\bibfnamefont {M.}~\bibnamefont {Maghrebi}}, \bibinfo {author} {\bibfnamefont {A.~V.}\ \bibnamefont {Gorshkov}}, \bibinfo {author} {\bibfnamefont {Z.}~\bibnamefont {Gong}}, \bibinfo {author} {\bibfnamefont {M.}~\bibnamefont {Cetina}},\ and\ \bibinfo {author} {\bibfnamefont {C.}~\bibnamefont {Monroe}},\ }\bibfield  {title} {\bibinfo {title} {Continuous symmetry breaking in a trapped-ion spin chain},\ }\href {https://doi.org/10.1038/s41586-023-06656-7} {\bibfield  {journal} {\bibinfo  {journal} {Nature}\ }\textbf {\bibinfo {volume} {623}},\ \bibinfo {pages} {713} (\bibinfo {year} {2023})}\BibitemShut {NoStop}%
\bibitem [{Note1()}]{Note1}%
  \BibitemOpen
  \bibinfo {note} {Here we use $H_{\protect \mathrm {tar}}$ to denote the target Hamiltonian after performing certain logical encoding. Such Hamiltonian operator can thus be related to the original Hamiltonian (i.e., without encoding) via an isometry defined by the logical codewords.}\BibitemShut {Stop}%
\bibitem [{\citenamefont {Shaw}\ \emph {et~al.}(2024)\citenamefont {Shaw}, \citenamefont {Chen}, \citenamefont {Choi}, \citenamefont {Mark}, \citenamefont {Scholl}, \citenamefont {Finkelstein}, \citenamefont {Elben}, \citenamefont {Choi},\ and\ \citenamefont {Endres}}]{manuel-endres-2024benchmarking}%
  \BibitemOpen
  \bibfield  {author} {\bibinfo {author} {\bibfnamefont {A.~L.}\ \bibnamefont {Shaw}}, \bibinfo {author} {\bibfnamefont {Z.}~\bibnamefont {Chen}}, \bibinfo {author} {\bibfnamefont {J.}~\bibnamefont {Choi}}, \bibinfo {author} {\bibfnamefont {D.~K.}\ \bibnamefont {Mark}}, \bibinfo {author} {\bibfnamefont {P.}~\bibnamefont {Scholl}}, \bibinfo {author} {\bibfnamefont {R.}~\bibnamefont {Finkelstein}}, \bibinfo {author} {\bibfnamefont {A.}~\bibnamefont {Elben}}, \bibinfo {author} {\bibfnamefont {S.}~\bibnamefont {Choi}},\ and\ \bibinfo {author} {\bibfnamefont {M.}~\bibnamefont {Endres}},\ }\bibfield  {title} {\bibinfo {title} {Benchmarking highly entangled states on a 60-atom analogue quantum simulator},\ }\href {https://doi.org/10.1038/s41586-024-07173-x} {\bibfield  {journal} {\bibinfo  {journal} {Nature}\ }\textbf {\bibinfo {volume} {628}},\ \bibinfo {pages} {71} (\bibinfo {year} {2024})}\BibitemShut {NoStop}%
\bibitem [{\citenamefont {Zanardi}\ and\ \citenamefont {Campos~Venuti}(2014)}]{zanardi2014coherent}%
  \BibitemOpen
  \bibfield  {author} {\bibinfo {author} {\bibfnamefont {P.}~\bibnamefont {Zanardi}}\ and\ \bibinfo {author} {\bibfnamefont {L.}~\bibnamefont {Campos~Venuti}},\ }\bibfield  {title} {\bibinfo {title} {Coherent quantum dynamics in steady-state manifolds of strongly dissipative systems},\ }\href {https://doi.org/10.1103/PhysRevLett.113.240406} {\bibfield  {journal} {\bibinfo  {journal} {Phys. Rev. Lett.}\ }\textbf {\bibinfo {volume} {113}},\ \bibinfo {pages} {240406} (\bibinfo {year} {2014})}\BibitemShut {NoStop}%
\bibitem [{ecz(2024)}]{eczoo_stab_4_2_2}%
  \BibitemOpen
  \bibfield  {title} {\bibinfo {title} {\([[4,2,2]]\) css code},\ }in\ \href {https://errorcorrectionzoo.org/c/stab_4_2_2} {\emph {\bibinfo {booktitle} {The Error Correction Zoo}}},\ \bibinfo {editor} {edited by\ \bibinfo {editor} {\bibfnamefont {V.~V.}\ \bibnamefont {Albert}}\ and\ \bibinfo {editor} {\bibfnamefont {P.}~\bibnamefont {Faist}}}\ (\bibinfo {year} {2024})\BibitemShut {NoStop}%
\bibitem [{\citenamefont {Self}\ \emph {et~al.}(2024{\natexlab{b}})\citenamefont {Self}, \citenamefont {Benedetti},\ and\ \citenamefont {Amaro}}]{quantinuum-icebergcode-Self_2024}%
  \BibitemOpen
  \bibfield  {author} {\bibinfo {author} {\bibfnamefont {C.~N.}\ \bibnamefont {Self}}, \bibinfo {author} {\bibfnamefont {M.}~\bibnamefont {Benedetti}},\ and\ \bibinfo {author} {\bibfnamefont {D.}~\bibnamefont {Amaro}},\ }\bibfield  {title} {\bibinfo {title} {Protecting expressive circuits with a quantum error detection code},\ }\href {https://doi.org/10.1038/s41567-023-02282-2} {\bibfield  {journal} {\bibinfo  {journal} {Nature Physics}\ }\textbf {\bibinfo {volume} {20}},\ \bibinfo {pages} {219–224} (\bibinfo {year} {2024}{\natexlab{b}})}\BibitemShut {NoStop}%
\bibitem [{Note2()}]{Note2}%
  \BibitemOpen
  \bibinfo {note} {Here we refer to the constraint such that all the nonzero eigenvalues of $H_{\protect \mathrm {pen}}$ are outside the interval $(-1,1)$.}\BibitemShut {Stop}%
\bibitem [{\citenamefont {Bernaschi}\ \emph {et~al.}(2024)\citenamefont {Bernaschi}, \citenamefont {Gonz{\'a}lez-Adalid~Pemart{\'\i}n}, \citenamefont {Mart{\'\i}n-Mayor},\ and\ \citenamefont {Parisi}}]{2d-TFIM-bernaschi2024quantum}%
  \BibitemOpen
  \bibfield  {author} {\bibinfo {author} {\bibfnamefont {M.}~\bibnamefont {Bernaschi}}, \bibinfo {author} {\bibfnamefont {I.}~\bibnamefont {Gonz{\'a}lez-Adalid~Pemart{\'\i}n}}, \bibinfo {author} {\bibfnamefont {V.}~\bibnamefont {Mart{\'\i}n-Mayor}},\ and\ \bibinfo {author} {\bibfnamefont {G.}~\bibnamefont {Parisi}},\ }\bibfield  {title} {\bibinfo {title} {The quantum transition of the two-dimensional ising spin glass},\ }\href {https://doi.org/10.1038/s41586-024-07647-y} {\bibfield  {journal} {\bibinfo  {journal} {Nature}\ }\textbf {\bibinfo {volume} {631}},\ \bibinfo {pages} {749} (\bibinfo {year} {2024})}\BibitemShut {NoStop}%
\bibitem [{\citenamefont {Kahanamoku-Meyer}\ and\ \citenamefont {Wei}(2024)}]{gregory_d_kahanamoku_meyer_2024_10906046}%
  \BibitemOpen
  \bibfield  {author} {\bibinfo {author} {\bibfnamefont {G.~D.}\ \bibnamefont {Kahanamoku-Meyer}}\ and\ \bibinfo {author} {\bibfnamefont {J.}~\bibnamefont {Wei}},\ }\href {https://doi.org/10.5281/zenodo.10906046} {\bibinfo {title} {Gregdmeyer/dynamite: v0.4.0}} (\bibinfo {year} {2024})\BibitemShut {NoStop}%
\bibitem [{\citenamefont {Biamonte}\ and\ \citenamefont {Love}(2008)}]{biamonte-realizable}%
  \BibitemOpen
  \bibfield  {author} {\bibinfo {author} {\bibfnamefont {J.~D.}\ \bibnamefont {Biamonte}}\ and\ \bibinfo {author} {\bibfnamefont {P.~J.}\ \bibnamefont {Love}},\ }\bibfield  {title} {\bibinfo {title} {Realizable hamiltonians for universal adiabatic quantum computers},\ }\href {https://doi.org/10.1103/PhysRevA.78.012352} {\bibfield  {journal} {\bibinfo  {journal} {Phys. Rev. A}\ }\textbf {\bibinfo {volume} {78}},\ \bibinfo {pages} {012352} (\bibinfo {year} {2008})}\BibitemShut {NoStop}%
\bibitem [{\citenamefont {Oliveira}\ and\ \citenamefont {Terhal}(2008)}]{2d-lattice-complexity-oliveira2005complexity}%
  \BibitemOpen
  \bibfield  {author} {\bibinfo {author} {\bibfnamefont {R.}~\bibnamefont {Oliveira}}\ and\ \bibinfo {author} {\bibfnamefont {B.~M.}\ \bibnamefont {Terhal}},\ }\bibfield  {title} {\bibinfo {title} {The complexity of quantum spin systems on a two-dimensional square lattice},\ }\href@noop {} {\bibfield  {journal} {\bibinfo  {journal} {Quantum Info. Comput.}\ }\textbf {\bibinfo {volume} {8}},\ \bibinfo {pages} {900–924} (\bibinfo {year} {2008})}\BibitemShut {NoStop}%
\bibitem [{\citenamefont {Lidar}(2019)}]{Lidar2019arbitrarytime}%
  \BibitemOpen
  \bibfield  {author} {\bibinfo {author} {\bibfnamefont {D.~A.}\ \bibnamefont {Lidar}},\ }\bibfield  {title} {\bibinfo {title} {Arbitrary-time error suppression for markovian adiabatic quantum computing using stabilizer subspace codes},\ }\href {https://doi.org/10.1103/PhysRevA.100.022326} {\bibfield  {journal} {\bibinfo  {journal} {Phys. Rev. A}\ }\textbf {\bibinfo {volume} {100}},\ \bibinfo {pages} {022326} (\bibinfo {year} {2019})}\BibitemShut {NoStop}%
\bibitem [{\citenamefont {K{\"o}yl{\"u}o{\u{g}}lu}\ \emph {et~al.}(2024)\citenamefont {K{\"o}yl{\"u}o{\u{g}}lu}, \citenamefont {Maskara}, \citenamefont {Feldmeier},\ and\ \citenamefont {Lukin}}]{koyluouglu2024floquet}%
  \BibitemOpen
  \bibfield  {author} {\bibinfo {author} {\bibfnamefont {N.~U.}\ \bibnamefont {K{\"o}yl{\"u}o{\u{g}}lu}}, \bibinfo {author} {\bibfnamefont {N.}~\bibnamefont {Maskara}}, \bibinfo {author} {\bibfnamefont {J.}~\bibnamefont {Feldmeier}},\ and\ \bibinfo {author} {\bibfnamefont {M.~D.}\ \bibnamefont {Lukin}},\ }\bibfield  {title} {\bibinfo {title} {Floquet engineering of interactions and entanglement in periodically driven rydberg chains},\ }\href {https://arxiv.org/abs/2408.02741} {\bibfield  {journal} {\bibinfo  {journal} {arXiv preprint arXiv:2408.02741}\ } (\bibinfo {year} {2024})}\BibitemShut {NoStop}%
\bibitem [{\citenamefont {Harley}\ \emph {et~al.}(2024)\citenamefont {Harley}, \citenamefont {Datta}, \citenamefont {Klausen}, \citenamefont {Bluhm}, \citenamefont {Fran{\c{c}}a}, \citenamefont {Werner},\ and\ \citenamefont {Christandl}}]{harley2024going-beyond-gadget}%
  \BibitemOpen
  \bibfield  {author} {\bibinfo {author} {\bibfnamefont {D.}~\bibnamefont {Harley}}, \bibinfo {author} {\bibfnamefont {I.}~\bibnamefont {Datta}}, \bibinfo {author} {\bibfnamefont {F.~R.}\ \bibnamefont {Klausen}}, \bibinfo {author} {\bibfnamefont {A.}~\bibnamefont {Bluhm}}, \bibinfo {author} {\bibfnamefont {D.~S.}\ \bibnamefont {Fran{\c{c}}a}}, \bibinfo {author} {\bibfnamefont {A.~H.}\ \bibnamefont {Werner}},\ and\ \bibinfo {author} {\bibfnamefont {M.}~\bibnamefont {Christandl}},\ }\bibfield  {title} {\bibinfo {title} {Going beyond gadgets: the importance of scalability for analogue quantum simulators},\ }\href {https://doi.org/10.1038/s41467-024-50744-9} {\bibfield  {journal} {\bibinfo  {journal} {Nature Communications}\ }\textbf {\bibinfo {volume} {15}},\ \bibinfo {pages} {6527} (\bibinfo {year} {2024})}\BibitemShut {NoStop}%
\bibitem [{\citenamefont {Burgarth}\ \emph {et~al.}(2021)\citenamefont {Burgarth}, \citenamefont {Facchi}, \citenamefont {Nakazato}, \citenamefont {Pascazio},\ and\ \citenamefont {Yuasa}}]{burgarth2021eter-adiab-q-evol}%
  \BibitemOpen
  \bibfield  {author} {\bibinfo {author} {\bibfnamefont {D.}~\bibnamefont {Burgarth}}, \bibinfo {author} {\bibfnamefont {P.}~\bibnamefont {Facchi}}, \bibinfo {author} {\bibfnamefont {H.}~\bibnamefont {Nakazato}}, \bibinfo {author} {\bibfnamefont {S.}~\bibnamefont {Pascazio}},\ and\ \bibinfo {author} {\bibfnamefont {K.}~\bibnamefont {Yuasa}},\ }\bibfield  {title} {\bibinfo {title} {Eternal adiabaticity in quantum evolution},\ }\href {https://doi.org/10.1103/PhysRevA.103.032214} {\bibfield  {journal} {\bibinfo  {journal} {Phys. Rev. A}\ }\textbf {\bibinfo {volume} {103}},\ \bibinfo {pages} {032214} (\bibinfo {year} {2021})}\BibitemShut {NoStop}%
\bibitem [{\citenamefont {Bravyi}\ and\ \citenamefont {Vyalyi}(2005)}]{bravyi2003commutative}%
  \BibitemOpen
  \bibfield  {author} {\bibinfo {author} {\bibfnamefont {S.}~\bibnamefont {Bravyi}}\ and\ \bibinfo {author} {\bibfnamefont {M.}~\bibnamefont {Vyalyi}},\ }\bibfield  {title} {\bibinfo {title} {Commutative version of the local hamiltonian problem and common eigenspace problem},\ }\href@noop {} {\bibfield  {journal} {\bibinfo  {journal} {Quantum Info. Comput.}\ }\textbf {\bibinfo {volume} {5}},\ \bibinfo {pages} {187–215} (\bibinfo {year} {2005})}\BibitemShut {NoStop}%
\bibitem [{\citenamefont {Dorier}\ \emph {et~al.}(2005)\citenamefont {Dorier}, \citenamefont {Becca},\ and\ \citenamefont {Mila}}]{dorier2005quantum-compass-model}%
  \BibitemOpen
  \bibfield  {author} {\bibinfo {author} {\bibfnamefont {J.}~\bibnamefont {Dorier}}, \bibinfo {author} {\bibfnamefont {F.}~\bibnamefont {Becca}},\ and\ \bibinfo {author} {\bibfnamefont {F.}~\bibnamefont {Mila}},\ }\bibfield  {title} {\bibinfo {title} {Quantum compass model on the square lattice},\ }\href {https://doi.org/10.1103/PhysRevB.72.024448} {\bibfield  {journal} {\bibinfo  {journal} {Phys. Rev. B}\ }\textbf {\bibinfo {volume} {72}},\ \bibinfo {pages} {024448} (\bibinfo {year} {2005})}\BibitemShut {NoStop}%
\end{thebibliography}%


\clearpage

\onecolumngrid

\section{Code equivalence between Hamiltonian $[[4,2,2]]$ code and CSS stabilizer code}\label{subsec: code equivalence}
One can straightforwardly show that the Hamiltonian $[[4,2,2]]$ code discussed in the main text is equivalent to the CSS $[[4,2,2]]$ stabilizer code (see, e.g., Refs.~\cite{eczoo_stab_4_2_2,quantinuum-icebergcode-Self_2024}) up to a single-qubit unitary.
More explicitly, the four logical states of the former (c.f. Eq.~\eqref{neq:hcode.422.basis}) are simultaneous eigenstates of the stabilizers of the latter:
\begin{align}
& P_{\text{enc}} \equiv | \overline{0+} \rangle \langle \overline{0+} | +| \overline{0-} \rangle \langle \overline{0-} | +| \overline{1+} \rangle \langle \overline{1+} | +| \overline{1-} \rangle \langle \overline{1-} | 
, \\
\label{seq:hcode.to.css422}
&X_{1} X_{2} X_{3} X_{4} P_{\text{enc}} = - P_{\text{enc}}, \quad 
Z_{1} Z_{2} Z_{3} Z_{4} P_{\text{enc}} = - P_{\text{enc}}. 
\end{align}
As two codes have the same code parameters, Eq.~\eqref{seq:hcode.to.css422} shows that they are equivalent up to local unitary transformation.

\section{Simulation error bounds}\label{appxsec:sim-err-bounds}
We aim to prove in this section that under the error suppression condition, our framework does suppress $1$-local coherent errors for analog quantum simulation, and provides a rigorous bound for the simulation error.
We will treat the simulation error bound as a special case of a more general bound that shows how a perturbation to a Hamiltonian leads to a unitary evolution that can be approximated by an effective Hamiltonian derived from the second-order perturbation theory.

\subsection{A general bound}
\begin{theorem}\label{thm:diff-H-Hle2}
    Consider a Hamiltonian $H_0$ that has an eigenspace $\mathcal{S}_0$ of zero energy with a spectral gap at least $\Delta > 0$, meaning that all the nonzero eigenvalues of $H_0$ are outside the interval $(-\Delta,\Delta)$.
    Denote by $P_0$ and $Q_0$ the projectors onto $\mathcal{S}_0$ and its orthogonal complement $\mathcal{S}_0^{\perp}$, respectively.
    Denote by $R_0$ the inverse operator of $H_0$ on the subspace $\mathcal{S}_0^{\perp}$, i.e., $R_0$ is the pseudoinverse of $H_0$ which satisfies $R_0 = Q_0 R_0 Q_0$ and $R_0 H_0 = H_0 R_0 = Q_0$.
    Let $W$ be a Hermitian operator with $\kappa := \frac{\|W\|}{\Delta} \le \frac{1}{4}$, where $\|\cdot\|$ stands for the operator norm.
    We also introduce the first- and second-order perturbative contributions to the effective Hamiltonian
    \begin{subequations}
        \begin{align}
        \Heff^{(1)} &:= P_0 W P_0, \\
        \Heff^{(2)} &:= -P_0 W R_0 W P_0,
        \end{align}
    \end{subequations}
    and denote
    \begin{equation}
        \Heff^{(\le 2)} := \Heff^{(1)} + \Heff^{(2)}.
    \end{equation}
    Then the difference between the unitary generated by the perturbed Hamiltonian $H := H_0 + W$ and that generated by $\Heff^{(\le 2)}$, when applied to initial states in $\mathcal{S}_0$, is bounded by:
    \begin{equation}\label{eq:diff-H-Hle2-P0}
        \left\| e^{-iHt} P_0 - e^{-i\Heff^{(\le 2)}t} P_0 \right\| \le 5 \kappa + 6 \kappa^2 \|Q_0 W P_0\| t, \quad \forall t > 0.
    \end{equation}
\end{theorem}

To prove the theorem, we make use of a formula in \cite{burgarth2021eter-adiab-q-evol} to separate different timescales.
For completeness, we provide the full details below.

\begin{lemma}\label{lem:int-by-part-trick}(Adapted from \cite[Eq.~(2.10)]{burgarth2021eter-adiab-q-evol}.)
Let $S$ be an arbitrary operator such that $S = Q_0 S P_0$.
Then
\begin{equation}
    \int_0^t e^{iH\tau} S e^{-i\Heff^{(\le 2)}\tau} \mathrm{d} \tau = -i \left[ e^{iH\tau} R_0 S e^{-i\Heff^{(\le 2)}\tau} \right]_{\tau=0}^{t} - \int_0^t e^{iH\tau} P_0 W R_0 S e^{-i\Heff^{(\le 2)}\tau} \mathrm{d} \tau + \int_0^t e^{iH\tau} \widetilde{S} e^{-i\Heff^{(\le 2)}\tau} \mathrm{d} \tau,
\end{equation}
where $\widetilde{S} := -Q_0 W R_0 S + R_0 S W P_0 - R_0 S W R_0 W P_0$ satisfies $\widetilde{S} = Q_0 \widetilde{S} P_0$ and $\|\widetilde{S}\| \le \left( 2\frac{\|W\|}{\Delta} + \frac{\|W\|^2}{\Delta^2} \right) \|S\|$.
\end{lemma}
\begin{proof}
    \begin{equation}
    \begin{aligned}
    & \int_0^t e^{iH\tau} S e^{-i\Heff^{(\le 2)}\tau} \mathrm{d} \tau \\
    &= \int_0^t e^{iH\tau} \underbrace{H_0 R_0}_{=Q_0} S e^{-i\Heff^{(\le 2)}\tau} \mathrm{d} \tau \\
    &= -i \int_0^t \frac{\mathrm{d}}{\mathrm{d} \tau} \left( e^{iH\tau} e^{-iW\tau} \right) e^{iW\tau} R_0 S e^{-i\Heff^{(\le 2)}\tau} \mathrm{d} \tau \\
    &= -i \left[ e^{iH\tau} R_0 S e^{-i\Heff^{(\le 2)}\tau} \right]_{\tau=0}^{t} + i \int_0^t e^{iH\tau} e^{-iW\tau} \frac{\mathrm{d}}{\mathrm{d} \tau} \left( e^{iW\tau} R_0 S e^{-i\Heff^{(\le 2)}\tau} \right) \mathrm{d} \tau \\
    &= -i \left[ e^{iH\tau} R_0 S e^{-i\Heff^{(\le 2)}\tau} \right]_{\tau=0}^{t} - \int_0^t e^{iH\tau} \left( W R_0 S - R_0 S \left( P_0 W P_0 - P_0 W R_0 W P_0 \right) \right) e^{-i\Heff^{(\le 2)}\tau} \mathrm{d} \tau \\
    &= -i \left[ e^{iH\tau} R_0 S e^{-i\Heff^{(\le 2)}\tau} \right]_{\tau=0}^{t} - \int_0^t e^{iH\tau} \left( P_0 W R_0 S + Q_0 W R_0 S - R_0 S W P_0 + R_0 S W R_0 W P_0 \right) e^{-i\Heff^{(\le 2)}\tau} \mathrm{d} \tau \\
    &= -i \left[ e^{iH\tau} R_0 S e^{-i\Heff^{(\le 2)}\tau} \right]_{\tau=0}^{t} - \int_0^t e^{iH\tau} P_0 W R_0 S e^{-i\Heff^{(\le 2)}\tau} \mathrm{d} \tau + \int_0^t e^{iH\tau} \widetilde{S} e^{-i\Heff^{(\le 2)}\tau} \mathrm{d} \tau,
    \end{aligned}
    \end{equation}
    where the third equality holds due to integration by part.
    The inequality $\|\widetilde{S}\| \le \left( 2\frac{\|W\|}{\Delta} + \frac{\|W\|^2}{\Delta^2} \right) \|S\|$ follows from the fact that $\|R_0\| \le \frac{1}{\Delta}$.
\end{proof}

\begin{proof}[Proof of Theorem \ref{thm:diff-H-Hle2}]
    Using the fundamental theorem of calculus,
    \begin{equation}
    \begin{aligned}
    \left( \id - e^{iHt} e^{-i\Heff^{(\le 2)}t} \right) P_{0} &= (-1) \int_0^t \frac{\mathrm{d}}{\mathrm{d} \tau} \left( e^{iH\tau} e^{-i\Heff^{(\le 2)}\tau} \right) P_0 \mathrm{d} \tau \\
    &= -i \int_0^t e^{iH\tau} \left( H - \Heff^{(\le 2)} \right) P_0 e^{-i\Heff^{(\le 2)}\tau} \mathrm{d} \tau \\
    &= -i \underbrace{\int_0^t e^{iH\tau} Q_0 W P_0 e^{-i\Heff^{(\le 2)}\tau} \mathrm{d} \tau}_{\text{(I)}} - i \int_0^t e^{iH\tau} P_0 W R_0 W P_0 e^{-i\Heff^{(\le 2)}\tau} \mathrm{d} \tau.
    \end{aligned}
    \end{equation}
    Applying Lemma \ref{lem:int-by-part-trick} to (I), we obtain
    \begin{equation}
    \begin{aligned}
    \left( \id - e^{iHt} e^{-i\Heff^{(\le 2)}t} \right) P_{0} &= (-1) \left[ e^{iH\tau} R_0 W P_0 e^{-i\Heff^{(\le 2)}\tau} \right]_{\tau=0}^{t} + \cancel{i \int_0^t e^{iH\tau} P_0 W R_0 W P_0 e^{-i\Heff^{(\le 2)}\tau} \mathrm{d} \tau} \\
    & \quad -i \int_0^t e^{iH\tau} S^{(1)} e^{-i\Heff^{(\le 2)}\tau} \mathrm{d} \tau - \cancel{i \int_0^t e^{iH\tau} P_0 W R_0 W P_0 e^{-i\Heff^{(\le 2)}\tau} \mathrm{d} \tau} \\
    &= (-1) \left[ e^{iH\tau} R_0 W P_0 e^{-i\Heff^{(\le 2)}\tau} \right]_{\tau=0}^{t} -i \underbrace{\int_0^t e^{iH\tau} S^{(1)} e^{-i\Heff^{(\le 2)}\tau} \mathrm{d} \tau}_{\text{(II)}},
    \end{aligned}
    \end{equation}
    where
    \begin{equation}
        S^{(1)} := - Q_0 W R_0 W P_0 + R_0 W P_0 W P_0 - R_0 W P_0 W R_0 W P_0
    \end{equation}
    satisfies $S^{(1)} = Q_0 S^{(1)} P_0$ and $\|S^{(1)}\| \le \left( 2\kappa + \kappa^2 \right) \|Q_0 W P_0\|$.
    Applying Lemma \ref{lem:int-by-part-trick} again to (II), we obtain
    \begin{equation}
    \begin{aligned}
    \left( \id - e^{iHt} e^{-i\Heff^{(\le 2)}t} \right) P_{0} &= (-1) \left[ e^{iH\tau} R_0 \left( Q_0 W P_0 + S^{(1)} \right) e^{-i\Heff^{(\le 2)}\tau} \right]_{\tau=0}^{t} \\
    &+i \int_0^t e^{iH\tau} P_0 W R_0 S^{(1)} e^{-i\Heff^{(\le 2)}\tau} \mathrm{d} \tau -i \int_0^t e^{iH\tau} S^{(2)} e^{-i\Heff^{(\le 2)}\tau} \mathrm{d} \tau,
    \end{aligned}
    \end{equation}
    where
    \begin{equation}
        S^{(2)} := - Q_0 W R_0 S^{(1)} + R_0 S^{(1)} W P_0 - R_0 S^{(1)} W R_0 W P_0
    \end{equation}
    satisfies $S^{(2)} = Q_0 S^{(2)} P_0$ and $\|S^{(2)}\| \le \left( 2 \kappa + \kappa^2 \right)^2 \|Q_0 W P_0\|$.
    This process can be repeated: after applying Lemma \ref{lem:int-by-part-trick} for $k$ times, we obtain
    \begin{equation}\label{eq:int-by-part-applied-k-times}
    \begin{aligned}
    \left( \id - e^{iHt} e^{-i\Heff^{(\le 2)}t} \right) P_{0} =& (-1) \left[ e^{iH\tau} R_0 \left( Q_0 W P_0 + \sum_{\ell=1}^{k-1} S^{(\ell)} \right) e^{-i\Heff^{(\le 2)}\tau} \right]_{\tau=0}^{t} \\
    &+i \int_0^t e^{iH\tau} P_0 W R_0 \left( \sum_{\ell=1}^{k-1} S^{(\ell)} \right) e^{-i\Heff^{(\le 2)}\tau} \mathrm{d} \tau -i \int_0^t e^{iH\tau} S^{(k)} e^{-i\Heff^{(\le 2)}\tau} \mathrm{d} \tau,
    \end{aligned}
    \end{equation}
    where $S^{(\ell)}$ is an operator that satisfies $S^{(\ell)} = Q_0 S^{(\ell)} P_0$ and $\|S^{(\ell)}\| \le \left( 2 \kappa + \kappa^2 \right)^{\ell} \|Q_0 W P_0\|$.
    By assumption $\kappa \le \frac{1}{4}$, then $2 \kappa + \kappa^2 \le \frac{9}{4} \kappa \le \frac{9}{16}$.
    It follows that $\lim_{k\to\infty} S^{(k)} = 0$ and $\left\| \sum_{\ell=1}^{\infty} S^{(\ell)} \right\| \le \sum_{\ell=1}^{\infty} \| S^{(\ell)} \| \le \frac{2\kappa + \kappa^2}{1-(2\kappa + \kappa^2)} \|Q_0 W P_0\| \le \frac{36}{7} \kappa \|Q_0 W P_0\| \le \frac{9}{7} \|Q_0 W P_0\|$.
    Therefore, taking the limit $k \to \infty$ in Eq.~\eqref{eq:int-by-part-applied-k-times}, we have
    \begin{equation}\label{eq:int-by-part-applied-inf-times}
    \begin{aligned}
    \left( \id - e^{iHt} e^{-i\Heff^{(\le 2)}t} \right) P_{0} &= (-1) \left[ e^{iH\tau} R_0 \left( Q_0 W P_0 + \sum_{\ell=1}^{\infty} S^{(\ell)} \right) e^{-i\Heff^{(\le 2)}\tau} \right]_{\tau=0}^{t} \\
    &\quad +i \int_0^t e^{iH\tau} P_0 W R_0 \left( \sum_{\ell=1}^{\infty} S^{(\ell)} \right) e^{-i\Heff^{(\le 2)}\tau} \mathrm{d} \tau,
    \end{aligned}
    \end{equation}
    from which we conclude that
    \begin{equation}\label{eq:unitary-diff-ub}
    \begin{aligned}
    \left\| e^{-iHt} P_{0} - e^{-i\Heff^{(\le 2)}t} P_{0} \right\| &= \left\| \left( \id - e^{iHt} e^{-i\Heff^{(\le 2)}t} \right) P_{0} \right\| \\
    &\le 2 \|R_0\| \left\| Q_0 W P_0 + \sum_{\ell=1}^{\infty} S^{(\ell)} \right\| + \|W\| \|R_0\| \left\| \sum_{\ell=1}^{\infty} S^{(\ell)} \right\| t \\
    &\le \frac{32}{7} \kappa + \frac{36}{7} \kappa^2 \|Q_0 W P_0\| t \\
    &\le 5 \kappa + 6 \kappa^2 \|Q_0 W P_0\| t.
    \end{aligned}
    \end{equation}
\end{proof}

\subsection{Analog quantum simulation error bound}
\begin{theorem}\label{thm:sim-err}
    Suppose we aim to simulate a target Hamiltonian $\Htar$ using a simulator Hamiltonian of the form $\Hsim = \lambda \Hpen + \sqrt{\lambda} \Henc^{(2)} + \Henc^{(1)}$, subject to Hermitian coherent errors $V = \sum_i \varepsilon_i V_i$.
    Suppose the penalty Hamiltonian $\Hpen$ has an eigenspace $\Szero$ of zero energy with a spectral gap of at least 1, meaning that all the nonzero eigenvalues of $\Hpen$ are outside the interval $(-1,1)$.
    Denote by $P_0$ and $Q_0$ the projectors onto $\Szero$ and its orthogonal complement $\Szero^{\perp}$, respectively.
    Suppose $\Htar$ is supported on a subspace $\Senc$ of $\Szero$, i.e., $\Penc \Htar \Penc = \Htar$ where $\Penc$ denotes the projector onto $\Senc$.
    We make the following assumptions:
    \begin{enumerate}
        \item $\forall i$, $P_0 V_i P_0 = c_i P_0$ for some scalar $c_i \in \mathbb{R}$.
        \item $(P_0 - \Penc) \Hencfirst \Penc = 0$.
        \item $P_0 \Hencsecond P_0 = 0$ and $(P_0 - \Penc) \Hencsecond Q_0 \Hpen^{-1} Q_0 \Hencsecond \Penc = 0$.
        \item $\Penc \Hencfirst \Penc - \Penc \Hencsecond Q_0 \Hpen^{-1} Q_0 \Hencsecond \Penc = \Htar$.
    \end{enumerate}
    Let $M := \max \left\{\left\|\Hencsecond\right\|^2, \left\|\Hencfirst + V\right\|\right\}$ and suppose $\lambda \ge 25 M$.
    The difference (up to a global phase) between the noisy unitary dynamics generated by $\Hsim + V$ and the desired unitary dynamics generated by $\Htar$, when applied to initial states in $\Senc$, is bounded by:
    \begin{equation}\label{eq:diff-Hsim-Htar-P0}
        \min_{\theta \in \mathbb{R}} \left\| e^{-i(\Hsim + V)t} \Penc - e^{i\theta} e^{-i \Htar t} \Penc \right\| \le 6 \sqrt{\frac{M}{\lambda}} + 13 \sqrt{\frac{M}{\lambda}} M t, \quad \forall t > 0.
    \end{equation}
\end{theorem}
\begin{proof}
    Denote $H_0 := \lambda \Hpen$, $W := \sqrt{\lambda} \Hencsecond + \Hencfirst + V$, and
    \begin{equation}
        \Heff^{(\le 2)} := P_0 W P_0 - P_0 W Q_0 H_0^{-1} Q_0 W P_0,
    \end{equation}
    By our assumptions, it is easy to verify that
    \begin{subequations}
        \begin{align}
        \Penc \Heff^{(\le 2)} \Penc &= \Htar + \alpha \Penc - \frac{1}{\lambda} \Penc A \Penc - \frac{1}{\sqrt{\lambda}} \Penc B \Penc, \\
        (P_0 - \Penc) \Heff^{(\le 2)} \Penc &= - \frac{1}{\lambda} (P_0 - \Penc) A \Penc - \frac{1}{\sqrt{\lambda}} (P_0 - \Penc) B \Penc, \\
        \Heff^{(\le 2)} \Penc &= \Htar + \alpha \Penc - \frac{1}{\lambda} P_0 A \Penc - \frac{1}{\sqrt{\lambda}} P_0 B \Penc,
        \end{align}
    \end{subequations}
    where $\alpha := \sum_{i} \varepsilon_i c_i$, and we have introduced two operators $A$ and $B$ that are independent of $\lambda$:
    \begin{subequations}
        \begin{align}
        A &:= \left( \Henc^{(1)}+V \right) Q_0 \Hpen^{-1} Q_0 \left( \Henc^{(1)}+V \right), \\
        B &:= \left( \Henc^{(1)}+V \right) Q_0 \Hpen^{-1} Q_0 \Henc^{(2)} + \Henc^{(2)} Q_0 \Hpen^{-1} Q_0 \left( \Henc^{(1)}+V \right).
        \end{align}
    \end{subequations}
    Note that $\frac{\|W\|}{\lambda} \le \sqrt{\frac{M}{\lambda}} + \frac{M}{\lambda} \le \frac{6}{5} \sqrt{\frac{M}{\lambda}} \le \frac{1}{4}$.
    Thus, applying Theorem \ref{thm:diff-H-Hle2}, we arrive at
    \begin{equation}
        \left\| e^{-i(\Hsim + V)t} P_0 - e^{-i \Heff^{(\le 2)} t} P_0 \right\| \le 5 \frac{\|W\|}{\lambda} + 6 \frac{\|W\|^3}{\lambda^2} t \le 6 \sqrt{\frac{M}{\lambda}} + 6 \left( \frac{6}{5} \right)^3 \sqrt{\frac{M}{\lambda}} M t.
    \end{equation}
    On the other hand,
    \begin{equation}
    \begin{aligned}
    \left\| e^{-i \Heff^{(\le 2)} t} \Penc - e^{-i \left( \Htar + \alpha \Penc \right) t} \Penc \right\| &= \left\| \left( \id - e^{i \Heff^{(\le 2)} t} e^{-i \left( \Htar + \alpha \Penc \right) t} \right) \Penc \right\| \\
    &= \left\| \int_0^t \frac{\mathrm{d}}{\mathrm{d} \tau} \left( e^{i \Heff^{(\le 2)} \tau} e^{-i \left( \Htar + \alpha \Penc \right) \tau} \right) \Penc \mathrm{d}\tau \right\| \\
    &= \left\| \int_0^t e^{i \Heff^{(\le 2)} \tau} \left( \Heff^{(\le 2)} - \Htar - \alpha \Penc \right) e^{-i \left( \Htar + \alpha \Penc \right) \tau} \Penc \mathrm{d}\tau \right\| \\
    &= \left\| \int_0^t e^{i \Heff^{(\le 2)} \tau} \left( \Heff^{(\le 2)} \Penc - \Htar - \alpha \Penc \right) e^{-i \left( \Htar + \alpha \Penc \right) \tau} \mathrm{d}\tau \right\| \\
    &\le \left\| \Heff^{(\le 2)} \Penc - \Htar - \alpha \Penc \right\| t \\
    &= \left\| \frac{1}{\lambda} P_0 A \Penc + \frac{1}{\sqrt{\lambda}} P_0 B \Penc \right\| t \\
    &\le \left( \frac{\|A\|}{\lambda} + \frac{\|B\|}{\sqrt{\lambda}} \right) t \\
    &\le \left( \frac{M}{\lambda} + 2 \sqrt{\frac{M}{\lambda}} \right) M t \\
    &\le \frac{11}{5} \sqrt{\frac{M}{\lambda}} M t.
    \end{aligned} 
    \end{equation}
    Combining the above two inequalities, we conclude that
    \begin{equation}
    \begin{aligned}
    &\quad \left\| e^{-i(\Hsim + V)t} \Penc - e^{-i\alpha t} e^{-i \Htar t} \Penc \right\| \\
    &\le \left\| e^{-i(\Hsim + V)t} \Penc - e^{-i \Heff^{(\le 2)} t} \Penc \right\| + \left\| e^{-i \Heff^{(\le 2)} t} \Penc - e^{-i \left( \Htar + \alpha \Penc \right) t} \Penc \right\| \\
    &\le \left\| e^{-i(\Hsim + V)t} P_0 - e^{-i \Heff^{(\le 2)} t} P_0 \right\| + \left\| e^{-i \Heff^{(\le 2)} t} \Penc - e^{-i \left( \Htar + \alpha \Penc \right) t} \Penc \right\| \\
    &\le 6 \sqrt{\frac{M}{\lambda}} + \left( 6 \left( \frac{6}{5} \right)^3 + \frac{11}{5} \right) \sqrt{\frac{M}{\lambda}} M t \\
    &\le 6 \sqrt{\frac{M}{\lambda}} + 13 \sqrt{\frac{M}{\lambda}} M t.
    \end{aligned}
    \end{equation}
\end{proof}

\begin{corollary}
    Adopt the same notations and assumptions as in Theorem \ref{thm:sim-err}.
    Given an arbitrary initial state $\ket{\psi_0} \in \Senc$, let $\ket{\psi_{\mathrm{tar}}(t)} := e^{-i \Htar t} \ket{\psi_0}$ be the state obtained by evolving the target Hamiltonian for a duration of time $t>0$, and let $\ket{\psi(t)} := e^{-i (\Hsim+V) t} \ket{\psi_0}$ be the state obtained by evolving the noisy simulator Hamiltonian for the same amount of time.
    Then the infidelity of the analog quantum simulation is bounded by
    \begin{equation}
        1 - \left| \langle \psi_{\mathrm{tar}}(t) | \psi(t) \rangle \right|^2 \le \frac{M}{\lambda} \left( 6 + 13Mt \right)^2.
    \end{equation}
\end{corollary}
\begin{proof}
    Theorem \ref{thm:sim-err} implies that $\left\| \ket{\psi(t)} - \ket{\psi_{\mathrm{tar}}(t)}\right\| \le \sqrt{\frac{M}{\lambda}} \left( 6 + 13 Mt \right)$.
    The corollary follows from the inequality
    \begin{equation}
        \left| \langle \psi_{\mathrm{tar}}(t) | \psi(t) \rangle \right|^2 \ge \left| \operatorname{Re} \langle \psi_{\mathrm{tar}}(t) | \psi(t) \rangle \right|^2 = \left( 1 - \frac{1}{2} \left\| \ket{\psi(t)} - \ket{\psi_{\mathrm{tar}}(t)}\right\|^2 \right)^2 \ge 1 - \left\| \ket{\psi(t)} - \ket{\psi_{\mathrm{tar}}(t)}\right\|^2.
    \end{equation}
\end{proof}

\section{A No-go theorem for 2-local commuting penalty Hamiltonians for stabilizing code space with a distance greater than 2}\label{appendix:new no-go}
As mentioned in the main text, it was proved in \cite{Marvian-Lidar-2014-PRL-2local-is-too-local} that the \textit{ground space} of any 2-local commuting Hamiltonian of qudits (of any local dimension $d$) cannot be a quantum code of distance greater than 1.
We emphasize that the ``ground space'' assumption is important, as we have already demonstrated through an example that quantum codes of distance 2 can exist as an \textit{excited} eigenspace of the Hamiltonian.
In this section, we extend the above no-go theorem to show that distance 2 is the best we can obtain for any eigenspace of any commuting $2$-local Hamiltonian.

\begin{theorem}
    No eigenspace of a 2-local commuting Hamiltonian of qudits (of any local dimension $d$) can be a quantum code of dimension greater than 1 and distance greater than 2 simultaneously.
\end{theorem}
\begin{proof}
    Similar to the original no-go theorem in \cite{Marvian-Lidar-2014-PRL-2local-is-too-local}, the theorem here is a corollary of a well-known fact about the structure of 2-local commuting Hamiltonians \cite{bravyi2003commutative}.
    Before stating the structural fact, let's set up some notations first: We consider a quantum system consisting of $n$ particles, with Hilbert space $\mathcal{H} = \bigotimes_{i=1}^{n} \mathcal{H}_{i}$, and a 2-local commuting Hamiltonian
    \begin{equation}
        H = \sum_{i<j} H_{ij},
    \end{equation}
    where $H_{ij}$ acts nontrivially only on $\mathcal{H}_{i} \otimes \mathcal{H}_{j}$, and $[H_{ij}, H_{kl}] = 0$ for any $i,j,k,l$.

    The structure lemma \cite[Lemma 8]{bravyi2003commutative} says that there is a way to decompose the Hilbert space of each particle as a direct sum
    \begin{equation}
        \mathcal{H}_i = \bigoplus_{\alpha_i} \mathcal{H}_i^{(\alpha_i)}, \quad \forall i \in \{1,\dots,n\},
    \end{equation}
    which induces a decomposition of the whole Hilbert space into multiple ``sectors''
    \begin{equation}
        \mathcal{H} = \bigoplus_{\alpha} \mathcal{H}^{(\alpha)}, \quad \alpha \equiv (\alpha_1, \dots, \alpha_n), \quad \mathcal{H}^{(\alpha)} \equiv \mathcal{H}_1^{(\alpha_1)} \otimes \cdots \otimes \mathcal{H}_n^{(\alpha_n)},
    \end{equation}
    and within each sector $\alpha$, there is a way to ``break'' each particle $i$ into several ``subparticles'' that we label $i.1, i.2, \dots, i.n$ (note that some subparticles could be trivial, i.e., $\mathcal{H}_{i.j}^{(\alpha_i)} = \mathbb{C}$)
    \begin{equation}
        \mathcal{H}_i^{(\alpha_i)} = \mathcal{H}_{i.i}^{(\alpha_i)} \otimes \bigotimes_{j \ne i} \mathcal{H}_{i.j}^{(\alpha_i)}, \quad \forall i \in \{1,\dots,n\}, \forall \alpha_i,
    \end{equation}
    such that the Hamiltonian $H$ admits a simple form:
    \begin{itemize}
        \item Every term $H_{ij}$ is block-diagonal with respect to the direct sum structure above, i.e.,
        \begin{equation}
            H_{ij} = \bigoplus_{\alpha} H_{ij}^{(\alpha_i, \alpha_j)}.
        \end{equation}
        This induces a block-diagonal structure of the total Hamiltonian:
        \begin{equation}
            H = \bigoplus_{\alpha} H^{(\alpha)}, \quad H^{(\alpha)} \equiv \sum_{i<j} H_{ij}^{(\alpha_i, \alpha_j)}.
        \end{equation}
        \item Within each sector $\alpha$, the interaction between particles $i$ and $j$ only affects the subparticles $i.j$ and $j.i$, i.e., $H_{ij}^{(\alpha_i, \alpha_j)}$ acts nontrivially only on $\mathcal{H}_{i.j}^{(\alpha_i)} \otimes \mathcal{H}_{j.i}^{(\alpha_j)}$.
    \end{itemize}
    
    Now, let's consider an arbitrary eigenspace $\mathcal{S}$ of $H$ with energy $E$ and dimension greater than 1.
    Since $H$ is block diagonal, we know immediately that $\mathcal{S} = \bigoplus_{\alpha} \mathcal{S}^{(\alpha)}$, where $\mathcal{S}^{(\alpha)} := \left\{ \ket{\psi} \in \mathcal{H}^{(\alpha)} \mid H^{(\alpha)} \ket{\psi} = E \ket{\psi} \right\}$.
    Note that some $\mathcal{S}^{(\alpha)}$ could be the zero vector space.
    We encounter two cases:
    
    Case 1: There exist two distinct sectors $\alpha$ and $\beta$ with nonzero $\mathcal{S}^{(\alpha)}$ and $\mathcal{S}^{(\beta)}$.
    Choose one state $\ket{\psi}$ from $\mathcal{S}^{(\alpha)}$ and one state $\ket{\phi}$ from $\mathcal{S}^{(\beta)}$.
    Since $\alpha \ne \beta$, there is a particle $i$ such that $\alpha_i \ne \beta_i$.
    It follows that if we define the following 1-local operator acting on the $i$-th particle
    \begin{equation}
        V := P_i^{(\alpha_i)} - P_i^{(\beta_i)},
    \end{equation}
    where $P_i^{(\alpha_i)}$ and $P_i^{(\beta_i)}$ are the projectors of $\mathcal{H}_i$ onto $\mathcal{H}_i^{(\alpha_i)}$ and $\mathcal{H}_i^{(\beta_i)}$, respectively, then we have $\bra{\psi} V \ket{\psi} = 1$ whereas $\bra{\phi} V \ket{\phi} = -1$.
    So $V$ is not a detectable error for $\mathcal{S}$; hence, the code distance of $\mathcal{S}$ is only one.

    Case 2: There is only one sector $\alpha$ with a nonzero $\mathcal{S}^{(\alpha)}$; in other words, $\mathcal{S}$ is entirely contained in sector $\alpha$.
    In this case, $\mathcal{S}$ is exactly the eigenspace of $H^{(\alpha)}$ with energy $E$.
    However, recall that $H^{(\alpha)}$ has a decoupled form describing independent pairs of interacting subparticles, i.e., we can write
    \begin{equation}
        H^{(\alpha)} = \sum_{i < j} h_{ij},
    \end{equation}
    where $h_{ij}$ acts nontrivially only on the two subparticles $i.j$ and $j.i$.
    Therefore, $\mathcal{S}$ has an orthonormal basis $\{\ket{\eta}\}$ where every basis vector has the product form
    \begin{equation}
        \ket{\eta} = \left( \bigotimes_{i=1}^{n} \ket{\eta_{ii}} \right) \otimes \left( \bigotimes_{i<j} \ket{\eta_{ij}} \right),
    \end{equation}
    where $\ket{\eta_{ii}}$ is a state of the subparticle $i.i$, and $\ket{\eta_{ij}}$ is a joint state of subparticles $i.j$ and $j.i$.
    Now, choose two different basis vectors $\ket{\eta}$ and $\ket{\eta'}$.
    Since $\ket{\eta} \perp \ket{\eta'}$, either $\ket{\eta_{ii}} \perp \ket{\eta'_{ii}}$ for some $i$ or $\ket{\eta_{ij}} \perp \ket{\eta'_{ij}}$ for some pair $i<j$.
    In the former case, we define the 1-local operator $V := \ket{\eta_{ii}}\bra{\eta_{ii}} - \ket{\eta'_{ii}}\bra{\eta'_{ii}}$; in the latter case, we define the 2-local operator $V := \ket{\eta_{ij}}\bra{\eta_{ij}} - \ket{\eta'_{ij}}\bra{\eta'_{ij}}$.
    Either way, we arrive at $\bra{\eta} V \ket{\eta} = 1 \ne -1 = \bra{\eta'} V \ket{\eta'}$, and conclude that $V$ is not a detectable error for $\mathcal{S}$, and hence the code distance of $\mathcal{S}$ is at most 2.
\end{proof}

\section{Perturbative gadgets}\label{appxsec:gadgets}
In this section, we verify the correctness of the perturbative gadgets used to generate cross-block logical interactions, by explicitly calculating their 2nd-order effective Hamiltonians inside the zero-energy subspace of the penalty Hamiltonian.
We also identify several conditions under which two perturbative gadgets can coexist without interfering with each other.
These results will be heavily used later in Sections \ref{appxsec:scheme-1D-models} and \ref{appxsec:scheme-2D-models} to aid the design of encoding schemes for specific target Hamiltonians.

\subsection{Some notations}
We need the following notation for the four Bell states of a pair of physical qubits:
\begin{equation}
    \begin{aligned}
    \ket{\phi_{1,1}} &\equiv \frac{1}{\sqrt{2}} \left( \ket{00} + \ket{11} \right), \\
    \ket{\phi_{1,-1}} &\equiv \frac{1}{\sqrt{2}} \left( \ket{00} - \ket{11} \right), \\
    \ket{\phi_{-1,1}} &\equiv \frac{1}{\sqrt{2}} \left( \ket{01} + \ket{10} \right), \\
    \ket{\phi_{-1,-1}} &\equiv \frac{1}{\sqrt{2}} \left( \ket{01} - \ket{10} \right).
    \end{aligned}
\end{equation}
The subscripts on the left-hand side are chosen so that $\ket{\phi_{s,t}}$ is the simultaneous eigenstate of $Z \otimes Z$ and $X \otimes X$ with eigenvalues $s$ and $t$, respectively.
It is directly verifiable that $\forall s, t \in \{-1,1\}$,
\begin{equation}\label{eq:single-Pauli-on-Bell-pair}
    \begin{aligned}
    & (Z \otimes \id) \ket{\phi_{s,t}} = \ket{\phi_{s,-t}},  & (\id \otimes Z) \ket{\phi_{s,t}} = s \ket{\phi_{s,-t}}, \\
    & (X \otimes \id) \ket{\phi_{s,t}} = t \ket{\phi_{-s,t}},  & (\id \otimes X) \ket{\phi_{s,t}} = \ket{\phi_{-s,t}}.
    \end{aligned}
\end{equation}
With the above notation, we can rewrite the four logical states in a single [[4,2,2]] code block introduced in the main article as
\begin{equation}
    \begin{aligned}
    \ket{\overline{0+}} &= \ket{\phi_{1,1}} \ket{\phi_{-1,-1}}, \\
    \ket{\overline{0-}} &= \ket{\phi_{1,-1}} \ket{\phi_{-1,1}}, \\
    \ket{\overline{1+}} &= \ket{\phi_{-1,1}} \ket{\phi_{1,-1}}, \\
    \ket{\overline{1-}} &= \ket{\phi_{-1,-1}} \ket{\phi_{1,1}}.
    \end{aligned}
\end{equation}
In words, $\ket{\phi_{s,t}} \ket{\phi_{-s,-t}}$ is the logical state where the first logical qubit is in the $\overline{Z}$-basis with eigenvalue $s$ and the second logical qubit in the $\overline{X}$-basis with eigenvalue $t$.

Now consider a quantum system consisting of $n$ [[4,2,2]] code blocks, where the $b$th block has parameters $(g_x^{(b)}, g_z^{(b)})$ with $|g_x^{(b)}| \ne |g_z^{(b)}|$, $b=1,2,\dots,n$, meaning that the total penalty Hamiltonian is given by
\begin{equation}
    \Hpen = \sum_{b=1}^{n} \left( g_x^{(b)} X_{1}^{(b)} X_{2}^{(b)} + g_z^{(b)} Z_{1}^{(b)} Z_{2}^{(b)} + g_x^{(b)} X_{3}^{(b)} X_{4}^{(b)} + g_z^{(b)} Z_{3}^{(b)} Z_{4}^{(b)} \right),
\end{equation}
where $X_{i}^{(b)}$ and $Z_{i}^{(b)}$ denote the Pauli X and Pauli Z operators acting on the $i$th physical qubit in the $b$th block, respectively.
An eigenbasis of $\Hpen$ consists of the following states
\begin{equation}\label{eq:Hpen-eigenbasis}
    \ket{\phi_{s_1,t_1}} \ket{\phi_{s'_1,t'_1}} \cdots \ket{\phi_{s_n,t_n}} \ket{\phi_{s'_n,t'_n}} \text{ with energy } \sum_{b=1}^{n} \left( (s_b+s'_b) g_z^{(b)} + (t_b+t'_b) g_x^{(b)} \right), \quad \forall \vec{s}, \vec{s'}, \vec{t}, \vec{t'} \in \{-1,1\}^{n}.
\end{equation}
As always, we use $P_0$ to denote the projector onto the zero-energy subspace $\Szero$ of $\Hpen$, $Q_0 = I - P_0$ the projector onto its orthogonal complement $\Szero^{\perp}$, and $\Penc$ the projector onto the encoding subspace $\Senc$.
Note that $\Senc$ is spanned by states of the form
\begin{equation}\label{eq:Senc-basis}
    \ket{\phi_{s_1,t_1}} \ket{\phi_{-s_1,-t_1}} \cdots \ket{\phi_{s_n,t_n}} \ket{\phi_{-s_n,-t_n}}, \quad \forall \vec{s}, \vec{t} \in \{-1,1\}^{n}.
\end{equation}
We denote by $\Hpen^{-1}$ the pseudoinverse of $\Hpen$, which satisfies $\Hpen \Hpen^{-1} = \Hpen^{-1} \Hpen = Q_0$ and $\Hpen^{-1} = Q_0 \Hpen^{-1} Q_0$.

\subsection{Derivation of a single perturbative gadget}
In this work, a perturbative gadget between two blocks refers to a Hermitian operator $G$ acting on the physical qubits of the two blocks, that satisfies the following criteria:
\begin{description}
    \item[Criterion 1] $P_0 G P_0 = 0$.
    \item[Criterion 2] $(P_0 - \Penc) G \Hpen^{-1} G \Penc = 0$.
\end{description}
Here, the first criterion ensures that a perturbation of the form $\sqrt{\lambda} G$ to the Hamiltonian $\lambda \Hpen$ generates the effective Hamiltonian $- P_0 G \Hpen^{-1} G P_0$ in the zero-energy subspace $\Szero$ (in the limit of large $\lambda$), while the second criterion implies that this effective Hamiltonian leaves the encoding subspace $\Senc$ invariant, hence resulting in a logical operator equal to $\overline{L}(G) := - \Penc G \Hpen^{-1} G \Penc$.
Below we check the criteria and calculate the logical operator for a concrete example.
The calculation for other gadgets follows similarly, and hence we omit the details for them and only report the results in Table \ref{tab:gadget-lookup}.

Consider the following perturbative gadget acting on the first two blocks
\begin{equation}\label{eq:example gadget}
    G = \alpha Z_2^{(1)} X_3^{(2)} + \beta Z_4^{(1)} X_3^{(2)},
\end{equation}
where $\alpha, \beta \in \mathbb{R}$.
First of all, we require $|g_x^{(1)}| \ne |g_z^{(2)}|$ for the first criterion to be true.
To see this, note that $Z_2^{(1)}$ anticommutes with $X_1^{(1)} X_2^{(1)}$ and commutes with other terms in $\Hpen$, while $X_3^{(2)}$ anticommutes with $Z_3^{(2)} Z_4^{(2)}$ and commutes with other terms in $\Hpen$.
It follows that the action of $Z_2^{(1)} X_3^{(2)}$ on an arbitrary eigenbasis state of $\Hpen$ in Eq.~\eqref{eq:Hpen-eigenbasis} changes its energy by $\pm 2 g_x^{(1)} \pm 2 g_z^{(2)}$, which is nonzero under the condition $|g_x^{(1)}| \ne |g_z^{(2)}|$.
This implies that $P_0 Z_2^{(1)} X_3^{(2)} P_0 = 0$.
Similarly $P_0 Z_4^{(1)} X_3^{(2)} P_0 = 0$ and therefore $P_0 G P_0 = 0$.
Next, we check the second criterion by considering the action of $G \Hpen^{-1} G$ on a basis vector of $\Senc$ (c.f. Eq.~\eqref{eq:Senc-basis} \& Eq.~\eqref{eq:single-Pauli-on-Bell-pair}): $\forall s_1, t_1, s_2, t_2 \in \{-1,1\}$, we have
\begin{align}
    &~~~ G \Hpen^{-1} G \ket{\phi_{s_1,t_1}} \ket{\phi_{-s_1,-t_1}} \ket{\phi_{s_2,t_2}} \ket{\phi_{-s_2,-t_2}} \ket{\text{rest}} \nonumber \\
    &= G \Hpen^{-1} \left( - s_1 t_2 \alpha \ket{\phi_{s_1,-t_1}} \ket{\phi_{-s_1,-t_1}} \ket{\phi_{s_2,t_2}} \ket{\phi_{s_2,-t_2}} + s_1 t_2 \beta \ket{\phi_{s_1,t_1}} \ket{\phi_{-s_1,t_1}} \ket{\phi_{s_2,t_2}} \ket{\phi_{s_2,-t_2}} \right) \ket{\text{rest}}  \nonumber \\
    &= G \left( \frac{s_1 t_2 \alpha}{2t_1 g_x^{(1)} - 2 s_2 g_z^{(2)}} \ket{\phi_{s_1,-t_1}} \ket{\phi_{-s_1,-t_1}} \ket{\phi_{s_2,t_2}} \ket{\phi_{s_2,-t_2}} \right.
    \nonumber \\
    & \quad \left. + \frac{s_1 t_2 \beta}{2 t_1 g_x^{(1)} + 2 s_2 g_z^{(2)}} \ket{\phi_{s_1,t_1}} \ket{\phi_{-s_1,t_1}} \ket{\phi_{s_2,t_2}} \ket{\phi_{s_2,-t_2}} \right) \ket{\text{rest}}  \nonumber \\
    &= \left[ \frac{\alpha}{2t_1 g_x^{(1)} - 2 s_2 g_z^{(2)}} \left( - \alpha \ket{\phi_{s_1,t_1}} \ket{\phi_{-s_1,-t_1}} \ket{\phi_{s_2,t_2}} \ket{\phi_{-s_2,-t_2}} + \beta \ket{\phi_{s_1,-t_1}} \ket{\phi_{-s_1,t_1}} \ket{\phi_{s_2,t_2}} \ket{\phi_{-s_2,-t_2}} \right) \right.  \nonumber \\
    &~~~~ + \left. \frac{\beta}{2 t_1 g_x^{(1)} + 2 s_2 g_z^{(2)}} \left( - \alpha \ket{\phi_{s_1,-t_1}} \ket{\phi_{-s_1,t_1}} \ket{\phi_{s_2,t_2}} \ket{\phi_{-s_2,-t_2}} + \beta \ket{\phi_{s_1,t_1}} \ket{\phi_{-s_1,-t_1}} \ket{\phi_{s_2,t_2}} \ket{\phi_{-s_2,-t_2}} \right) \right] \ket{\text{rest}}  \nonumber \\
    &= \left[ \frac{1}{2} \left( - \frac{\alpha^2}{t_1 g_x^{(1)} - s_2 g_z^{(2)}} + \frac{\beta^2}{t_1 g_x^{(1)} + s_2 g_z^{(2)}} \right) \ket{\phi_{s_1,t_1}} \ket{\phi_{-s_1,-t_1}} \ket{\phi_{s_2,t_2}} \ket{\phi_{-s_2,-t_2}} \right.  \nonumber \\
    &~~~~ + \frac{\alpha \beta}{2} \left. \left( \frac{1}{t_1 g_x^{(1)} - s_2 g_z^{(2)}} - \frac{1}{t_1 g_x^{(1)} + s_2 g_z^{(2)}} \right) \ket{\phi_{s_1,-t_1}} \ket{\phi_{-s_1,t_1}} \ket{\phi_{s_2,t_2}} \ket{\phi_{-s_2,-t_2}} \right] \ket{\text{rest}}  
    ,
\end{align}
where $\ket{\text{rest}}$ represents the state of the physical qubits from the 3rd to the last block.
Simplifying the above equation, we thus obtain
\begin{align}
 &~~~ G \Hpen^{-1} G \ket{\phi_{s_1,t_1}} \ket{\phi_{-s_1,-t_1}} \ket{\phi_{s_2,t_2}} \ket{\phi_{-s_2,-t_2}} \ket{\text{rest}} \nonumber \\
    &= \left[ \frac{\left( \alpha^2 - \beta^2 \right) g_x^{(1)}}{2 \left( \left( g_z^{(2)} \right)^2 - \left( g_x^{(1)} \right)^2 \right)} t_1 \ket{\phi_{s_1,t_1}} \ket{\phi_{-s_1,-t_1}} \ket{\phi_{s_2,t_2}} \ket{\phi_{-s_2,-t_2}} \right.  \nonumber \\
    &~~~~ + \frac{\left( \beta^2 + \alpha^2 \right) g_z^{(2)}}{2 \left( \left( g_z^{(2)} \right)^2 - \left( g_x^{(1)} \right)^2 \right)} s_2 \ket{\phi_{s_1,t_1}} \ket{\phi_{-s_1,-t_1}} \ket{\phi_{s_2,t_2}} \ket{\phi_{-s_2,-t_2}}  \nonumber \\
    &~~~~ \left. - \frac{\alpha \beta g_z^{(2)}}{\left( g_z^{(2)} \right)^2 - \left( g_x^{(1)} \right)^2} s_2 \ket{\phi_{s_1,-t_1}} \ket{\phi_{-s_1,t_1}} \ket{\phi_{s_2,t_2}} \ket{\phi_{-s_2,-t_2}} \right] \ket{\text{rest}}  \nonumber \\
    &= \frac{1}{\left( g_z^{(2)} \right)^2 - \left( g_x^{(1)} \right)^2} \left[ \frac{\alpha^2 - \beta^2}{2} g_x^{(1)} \overline{X}_2^{(1)} + \frac{\alpha^2 + \beta^2}{2} g_z^{(2)} \overline{Z}_1^{(2)} - \alpha \beta g_z^{(2)} \overline{Z}_2^{(1)} \overline{Z}_1^{(2)} \right] \ket{\phi_{s_1,t_1}} \ket{\phi_{-s_1,-t_1}} \ket{\phi_{s_2,t_2}} \ket{\phi_{-s_2,-t_2}} \ket{\text{rest}}
\end{align}
where $\overline{X}_{i}^{(b)}$ and $\overline{Z}_{i}^{(b)}$ denote the logical Pauli X and Pauli Z operators acting on the $i$th logical qubit in the $b$th block, respectively.
The above calculation shows that after applying $G \Hpen^{-1} G$, the state stays inside $\Senc$, with the resulting logical interaction being
\begin{equation}
    \overline{L}(G) = \frac{1}{\left( g_z^{(2)} \right)^2 - \left( g_x^{(1)} \right)^2} \left( -\frac{\alpha^2 - \beta^2}{2} g_x^{(1)} \overline{X}_2^{(1)} - \frac{\alpha^2 + \beta^2}{2} g_z^{(2)} \overline{Z}_1^{(2)} + \alpha \beta g_z^{(2)} \overline{Z}_2^{(1)} \overline{Z}_1^{(2)} \right).
\end{equation}

\begin{table}[h!]
    \centering
    \begin{tabular}{|c|c|}
    \hline
    $G$ & $\overline{L}(G) := - \Penc G \Hpen^{-1} G \Penc$ \\
    \hline 
    $\alpha Z_2^{(a)} X_3^{(b)} + \beta Z_4^{(a)} X_3^{(b)}$ & $\frac{1}{\left( g_z^{(b)} \right)^2 - \left( g_x^{(a)} \right)^2} \left( \alpha \beta g_z^{(b)} \overline{Z}_2^{(a)} \overline{Z}_1^{(b)} -\frac{\alpha^2 - \beta^2}{2} g_x^{(a)} \overline{X}_2^{(a)} - \frac{\alpha^2 + \beta^2}{2} g_z^{(b)} \overline{Z}_1^{(b)} \right)$ \\
    \hline
    $\alpha Z_2^{(a)} X_1^{(b)} + \beta Z_2^{(a)} X_3^{(b)}$  & $\frac{1}{\left( g_z^{(b)} \right)^2 - \left( g_x^{(a)} \right)^2} \left( \alpha \beta g_x^{(a)} \overline{X}_2^{(a)} \overline{X}_1^{(b)} -\frac{\alpha^2 + \beta^2}{2} g_x^{(a)} \overline{X}_2^{(a)} + \frac{\alpha^2 - \beta^2}{2} g_z^{(b)} \overline{Z}_1^{(b)} \right)$ \\
    \hline
    \makecell{$\alpha Z_2^{(a)} X_3^{(b)} + \beta Z_4^{(a)} X_3^{(b)}$ \\ $+ \gamma Z_2^{(a)} X_1^{(b)}$} & $\frac{1}{\left( g_z^{(b)} \right)^2 - \left( g_x^{(a)} \right)^2} \left( \alpha \beta g_z^{(b)} \overline{Z}_2^{(a)} \overline{Z}_1^{(b)} + \alpha \gamma g_x^{(a)} \overline{X}_2^{(a)} \overline{X}_1^{(b)} -\frac{\alpha^2 - \beta^2 + \gamma^2}{2} g_x^{(a)} \overline{X}_2^{(a)} - \frac{\alpha^2 + \beta^2 - \gamma^2}{2} g_z^{(b)} \overline{Z}_1^{(b)} \right)$ \\
    \hline
    \end{tabular}
    \caption{Lookup table of the perturbative gadgets between block $a$ and block $b$ ($a \ne b$). The first column lists the physical form of the gadgets used in the paper, and the second column shows the generated logical interactions.
    It is implicitly assumed that $|g_x^{(a)}| \ne |g_z^{(b)}|$ so that $P_0 G P_0 = 0$.}
    \label{tab:gadget-lookup}
\end{table}

\subsection{Coexistence of two perturbative gadgets}
In stark contrast to the more straightforward encoding of the inner-block logical interactions, perturbative gadgets are not ``additive'', meaning that if $G_1$ and $G_2$ are two perturbative gadgets that generate cross-block logical interactions $\overline{L}(G_1)$ and $\overline{L}(G_2)$, respectively, then in general it is not guaranteed that $G_1 + G_2$ generates $\overline{L}(G_1) + \overline{L}(G_2)$.
Worse still, it could even happen that $G_1 + G_2$ no longer satisfies the second criterion in the previous subsection, leading to leakage of quantum states from the encoding subspace to its orthogonal complement in the zero-energy subspace.
To understand the cause of this unwanted behavior, let us expand the 2nd-order effective Hamiltonian of $G_1 + G_2$ inside the zero-energy subspace:
\begin{equation}
    - P_0 (G_1 + G_2) \Hpen^{-1} (G_1 + G_2) P_0 = \underbrace{- P_0 G_1 \Hpen^{-1} G_1 P_0}_{\overline{L}(G_1)} \underbrace{- P_0 G_2 \Hpen^{-1} G_2 P_0}_{\overline{L}(G_2)} - \left( \underbrace{P_0 G_1 \Hpen^{-1} G_2 P_0 + P_0 G_2 \Hpen^{-1} G_1 P_0}_{\text{two-gadget interference}} \right).
\end{equation}
We see that for the logical effects of the two gadgets to combine additively, it is crucial to make sure that the two-gadget interference term vanishes.
The following lemmas will be useful for that purpose.

\begin{lemma}\label{lem:gadget-coexist-4blocks}
    Consider four blocks $a$, $b$, $c$, $d$ with arbitrary parameters $g_x^{(a)}, g_z^{(a)}, g_x^{(b)}, g_z^{(b)}, g_x^{(c)}, g_z^{(c)}, g_x^{(d)}, g_z^{(d)}$.
    Let $A := Z_i^{(a)} X_j^{(b)}$ and $B := Z_k^{(c)} X_l^{(d)}$ for arbitrary $i,j,k,l \in \{1,2,3,4\}$.
    Then
    \begin{equation}
        P_0 A \Hpen^{-1} B P_0 + P_0 B \Hpen^{-1} A P_0 = 0.
    \end{equation}
\end{lemma}
\begin{proof}
    Without loss of generality assume that $a,b,c,d$ are the first four blocks among the $n$ blocks.
    For simplicity, we prove the lemma for the case $i,j,k,l \in \{1,2\}$.
    The proof for the other choices of $i,j,k,l$ follows in the same manner.
    Let $\ket{\psi}$ be a zero-energy eigenstate of $\Hpen$ in the form of Eq.~\eqref{eq:Hpen-eigenbasis}, i.e., $\ket{\psi} = \ket{\phi_{s_1,t_1}} \ket{\phi_{s'_1,t'_1}} \cdots \ket{\phi_{s_n,t_n}} \ket{\phi_{s'_n,t'_n}}$ with $(s_1+s'_1) g_z^{(1)} + (t_1+t'_1) g_x^{(1)} + \cdots + (s_n+s'_n) g_z^{(n)} + (t_n+t'_n) g_x^{(n)} = 0$.
    We have
    \begin{equation}
        \begin{aligned}
        &~~ B \Hpen^{-1} A \ket{\psi} \\
        &= B \Hpen^{-1} \left( Z_i^{(1)} \ket{\phi_{s_1,t_1}} \right) \ket{\phi_{s'_1,t'_1}} \left( X_j^{(2)} \ket{\phi_{s_2,t_2}} \right) \ket{\phi_{s'_2,t'_2}} \ket{\phi_{s_3,t_3}} \ket{\phi_{s'_3,t'_3}} \ket{\phi_{s_4,t_4}} \ket{\phi_{s'_4,t'_4}} \ket{\text{rest}} \\
        &= \frac{1}{-2 t_1 g_x^{(1)} - 2 s_2 g_z^{(2)}} B \left( Z_i^{(1)} \ket{\phi_{s_1,t_1}} \right) \ket{\phi_{s'_1,t'_1}} \left( X_j^{(2)} \ket{\phi_{s_2,t_2}} \right) \ket{\phi_{s'_2,t'_2}} \ket{\phi_{s_3,t_3}} \ket{\phi_{s'_3,t'_3}} \ket{\phi_{s_4,t_4}} \ket{\phi_{s'_4,t'_4}} \ket{\text{rest}} \\
        &= \frac{1}{-2 t_1 g_x^{(1)} - 2 s_2 g_z^{(2)}} \left( Z_i^{(1)} \ket{\phi_{s_1,t_1}} \right) \ket{\phi_{s'_1,t'_1}} \left( X_j^{(2)} \ket{\phi_{s_2,t_2}} \right) \ket{\phi_{s'_2,t'_2}} \\
        & \quad \quad \quad \quad \quad \quad \quad \quad \quad \quad \otimes \left( Z_k^{(3)} \ket{\phi_{s_3,t_3}} \right) \ket{\phi_{s'_3,t'_3}} \left( X_l^{(4)} \ket{\phi_{s_4,t_4}} \right)  \ket{\phi_{s'_4,t'_4}} \ket{\text{rest}}, \\
        \end{aligned}
    \end{equation}
    where $\ket{\text{rest}}$ represents the state of the physical qubits from the 5th to the last block.
    The second equality above follows from the observation that $Z_i^{(1)} \ket{\phi_{s_1,t_1}} \propto \ket{\phi_{s_1,-t_1}}$, which has an energy difference of $-2 t_1 g_x^{(1)}$ relative to $\ket{\phi_{s_1,t_1}}$, and that $X_j^{(2)} \ket{\phi_{s_2,t_2}} \propto \ket{\phi_{-s_2,t_2}}$, which has an energy difference of $-2 s_2 g_z^{(2)}$ relative to $\ket{\phi_{s_2,t_2}}$.
    Similarly,
    \begin{equation}
        \begin{aligned}
        &~~ A \Hpen^{-1} B \ket{\psi} \\
        &= A \Hpen^{-1} \ket{\phi_{s_1,t_1}} \ket{\phi_{s'_1,t'_1}} \ket{\phi_{s_2,t_2}} \ket{\phi_{s'_2,t'_2}} \left( Z_k^{(3)} \ket{\phi_{s_3,t_3}} \right) \ket{\phi_{s'_3,t'_3}} \left( X_l^{(4)} \ket{\phi_{s_4,t_4}} \right)  \ket{\phi_{s'_4,t'_4}} \ket{\text{rest}} \\
        &= \frac{1}{-2 t_3 g_x^{(3)} - 2 s_4 g_z^{(4)}} A \ket{\phi_{s_1,t_1}} \ket{\phi_{s'_1,t'_1}} \ket{\phi_{s_2,t_2}} \ket{\phi_{s'_2,t'_2}} \left( Z_k^{(3)} \ket{\phi_{s_3,t_3}} \right) \ket{\phi_{s'_3,t'_3}} \left( X_l^{(4)} \ket{\phi_{s_4,t_4}} \right)  \ket{\phi_{s'_4,t'_4}} \ket{\text{rest}} \\
        &= \frac{1}{-2 t_3 g_x^{(3)} - 2 s_4 g_z^{(4)}} \left( Z_i^{(1)} \ket{\phi_{s_1,t_1}} \right) \ket{\phi_{s'_1,t'_1}} \left( X_j^{(2)} \ket{\phi_{s_2,t_2}} \right) \ket{\phi_{s'_2,t'_2}} 
        \\
        & \quad \quad \quad \quad \quad \quad \quad \quad \quad \quad \otimes \left( Z_k^{(3)} \ket{\phi_{s_3,t_3}} \right) \ket{\phi_{s'_3,t'_3}} \left( X_l^{(4)} \ket{\phi_{s_4,t_4}} \right)  \ket{\phi_{s'_4,t'_4}} \ket{\text{rest}}. \\
        \end{aligned}
    \end{equation}
    In both of the above two equations, the final state is an eigenstate of $\Hpen$ with energy $E = -2 t_1 g_x^{(1)} - 2 s_2 g_z^{(2)} -2 t_3 g_x^{(3)} - 2 s_4 g_z^{(4)}$.
    If $E \ne 0$, then by definition $P_0 B \Hpen^{-1} A \ket{\psi} = P_0 A \Hpen^{-1} B \ket{\psi} = 0$.
    If $E = 0$, then $2 t_3 g_x^{(3)} + 2 s_4 g_z^{(4)} = - \left( 2 t_1 g_x^{(1)} + 2 s_2 g_z^{(2)} \right)$, hence $A \Hpen^{-1} B \ket{\psi} = - B \Hpen^{-1} A \ket{\psi}$.
    In both cases, we have $P_0 B \Hpen^{-1} A \ket{\psi} + P_0 A \Hpen^{-1} B \ket{\psi} = 0$.
    Since the above argument holds for arbitrary choice of $\ket{\psi}$, we arrive at $P_0 B \Hpen^{-1} A P_0 + P_0 A \Hpen^{-1} B P_0 = 0$.
\end{proof}

Lemma \ref{lem:gadget-coexist-4blocks} implies that any two perturbative gadgets from Table \ref{tab:gadget-lookup} do not interfere with each other whenever they act on four different blocks.
With a closer look at the argument made in the above proof, one easily realizes that the lemma still holds no matter whether the two operators $A$ and $B$ span four, three, or two different blocks; the important assumption is that out of the four physical qubits that support $A$ and $B$, no two of them come from a single Bell pair.
More precisely, we have the following Lemma \ref{lem:gadget-coexist-3blocks} and Lemma \ref{lem:gadget-coexist-2blocks}.
The proof is omitted since it is almost the same as for Lemma \ref{lem:gadget-coexist-4blocks}.

\begin{lemma}\label{lem:gadget-coexist-3blocks}
    Consider three blocks $a$, $b$, $c$ with arbitrary parameters $g_x^{(a)}, g_z^{(a)}, g_x^{(b)}, g_z^{(b)}, g_x^{(c)}, g_z^{(c)}$.
    Let $A := Z_i^{(a)} X_j^{(b)}$ or $A := X_i^{(a)} Z_j^{(b)}$, and $B := Z_k^{(a)} X_l^{(c)}$ or $B := X_k^{(a)} Z_l^{(c)}$, where $(i,k) \in \{(1,3)$, $(1,4)$, $(2,3)$, $(2,4)$, $(3,1)$, $(3,2)$, $(4,1)$, $(4,2)\}$ and $j,l \in \{1,2,3,4\}$.
    Then
    \begin{equation}
        P_0 A \Hpen^{-1} B P_0 + P_0 B \Hpen^{-1} A P_0 = 0.
    \end{equation}
\end{lemma}

\begin{lemma}\label{lem:gadget-coexist-2blocks}
    Consider two blocks $a$ and $b$ with arbitrary parameters $g_x^{(a)}, g_z^{(a)}, g_x^{(b)}, g_z^{(b)}$.
    Let $A := Z_i^{(a)} X_j^{(b)}$ or $A := X_i^{(a)} Z_j^{(b)}$, and $B := Z_k^{(a)} X_l^{(b)}$ or $B := X_k^{(a)} Z_l^{(b)}$, where $(i,k), (j,l) \in \{(1,3)$, $(1,4)$, $(2,3)$, $(2,4)$, $(3,1)$, $(3,2)$, 
    $(4,1)$, $(4,2)\}$.
    Then
    \begin{equation}
        P_0 A \Hpen^{-1} B P_0 + P_0 B \Hpen^{-1} A P_0 = 0.
    \end{equation}
\end{lemma}

Even if the support of the operators $A$ and $B$ involves a common Bell pair, as long as they act on different qubits of the pair, one via the Pauli Z operator and the other via the Pauli X operator, the net interference effect is still zero.
More precisely, we have the following Lemma \ref{lem:gadget-coexist-3blocks-ZX} and Lemma \ref{lem:gadget-coexist-2blocks-ZX}.
We only provide the proof for the latter as the proof for the former is almost identical.
\begin{lemma}\label{lem:gadget-coexist-3blocks-ZX}
    Consider three blocks $a$, $b$, $c$ with arbitrary parameters $g_x^{(a)}, g_z^{(a)}, g_x^{(b)}, g_z^{(b)}, g_x^{(c)}, g_z^{(c)}$.
    Let
    \begin{equation}
        A := Z_i^{(a)} X_j^{(b)}, \quad B := X_k^{(a)} Z_l^{(c)},
    \end{equation}
    or
    \begin{equation}
        A := X_i^{(a)} Z_j^{(b)}, \quad B := Z_k^{(a)} X_l^{(c)},
    \end{equation}
    where $(i,k) \in \{(1,2), (2,1), (3,4), (4,3)\}$ and $j,l \in \{1,2,3,4\}$.
    Then
    \begin{equation}
        P_0 A \Hpen^{-1} B P_0 + P_0 B \Hpen^{-1} A P_0 = 0.
    \end{equation}
\end{lemma}

\begin{lemma}\label{lem:gadget-coexist-2blocks-ZX}
    Consider two blocks $a$ and $b$ with arbitrary parameters $g_x^{(a)}, g_z^{(a)}, g_x^{(b)}, g_z^{(b)}$.
    Let
    \begin{equation}
        A := Z_i^{(a)} X_j^{(b)}, \quad B := X_k^{(a)} Z_l^{(b)},
    \end{equation}
    or
    \begin{equation}
        A := X_i^{(a)} Z_j^{(b)}, \quad B := Z_k^{(a)} X_l^{(b)},
    \end{equation}
    where $(i,k), (j,l) \in \{(1,2), (2,1), (3,4), (4,3)\}$.
    Then
    \begin{equation}
        P_0 A \Hpen^{-1} B P_0 + P_0 B \Hpen^{-1} A P_0 = 0.
    \end{equation}
\end{lemma}
\begin{proof}
    Without loss of generality assume that $a$ and $b$ are the first two blocks among the $n$ blocks.
    For simplicity, we prove the lemma for the case $A = Z_1^{(1)} X_1^{(2)}, B = X_2^{(1)} Z_2^{(2)}$.
    The proof for the other choices of $i,j,k,l$ follows in the same manner.
    Let $\ket{\psi}$ be a zero-energy eigenstate of $\Hpen$ in the form of Eq.~\eqref{eq:Hpen-eigenbasis}, i.e., $\ket{\psi} = \ket{\phi_{s_1,t_1}} \ket{\phi_{s'_1,t'_1}} \cdots \ket{\phi_{s_n,t_n}} \ket{\phi_{s'_n,t'_n}}$ with $(s_1+s'_1) g_z^{(1)} + (t_1+t'_1) g_x^{(1)} + \cdots + (s_n+s'_n) g_z^{(n)} + (t_n+t'_n) g_x^{(n)} = 0$.
    We have
    \begin{equation}
        \begin{aligned}
        B \Hpen^{-1} A \ket{\psi} &= B \Hpen^{-1} \left( (Z \otimes \id) \ket{\phi_{s_1,t_1}} \right) \ket{\phi_{s'_1,t'_1}} \left( (X \otimes \id) \ket{\phi_{s_2,t_2}} \right) \ket{\phi_{s'_2,t'_2}} \ket{\text{rest}} \\
        &= \frac{1}{-2 t_1 g_x^{(1)} - 2 s_2 g_z^{(2)}} B \left( (Z \otimes \id) \ket{\phi_{s_1,t_1}} \right) \ket{\phi_{s'_1,t'_1}} \left( (X \otimes \id) \ket{\phi_{s_2,t_2}} \right) \ket{\phi_{s'_2,t'_2}} \ket{\text{rest}} \\
        &= \frac{1}{-2 t_1 g_x^{(1)} - 2 s_2 g_z^{(2)}} \left( (Z \otimes X) \ket{\phi_{s_1,t_1}} \right) \ket{\phi_{s'_1,t'_1}} \left( (X \otimes Z) \ket{\phi_{s_2,t_2}} \right) \ket{\phi_{s'_2,t'_2}} \ket{\text{rest}}, \\
        \end{aligned}
    \end{equation}
    where $\ket{\text{rest}}$ represents the state of the physical qubits from the 3rd to the last block.
    Similarly,
    \begin{equation}
        \begin{aligned}
        A \Hpen^{-1} B \ket{\psi} &= A \Hpen^{-1} \left( (\id \otimes X) \ket{\phi_{s_1,t_1}} \right) \ket{\phi_{s'_1,t'_1}} \left( (\id \otimes Z) \ket{\phi_{s_2,t_2}} \right) \ket{\phi_{s'_2,t'_2}} \ket{\text{rest}} \\
        &= \frac{1}{-2 s_1 g_z^{(1)} - 2 t_2 g_x^{(2)}} A \left( (\id \otimes X) \ket{\phi_{s_1,t_1}} \right) \ket{\phi_{s'_1,t'_1}} \left( (\id \otimes Z) \ket{\phi_{s_2,t_2}} \right) \ket{\phi_{s'_2,t'_2}} \ket{\text{rest}} \\
        &= \frac{1}{-2 s_1 g_z^{(1)} - 2 t_2 g_x^{(2)}} \left( (Z \otimes X) \ket{\phi_{s_1,t_1}} \right) \ket{\phi_{s'_1,t'_1}} \left( (X \otimes Z) \ket{\phi_{s_2,t_2}} \right) \ket{\phi_{s'_2,t'_2}} \ket{\text{rest}}. \\
        \end{aligned}
    \end{equation}
    Note that $(Z \otimes X) \ket{\phi_{s_1,t_1}} = \ket{\phi_{-s_1,-t_1}}$ has an energy difference of $-2s_1 g_z^{(1)}-2t_1 g_x^{(1)}$ relative to $\ket{\phi_{s_1,t_1}}$, and $(X \otimes Z) \ket{\phi_{s_2,t_2}} = -s_2 t_2 \ket{\phi_{-s_2,-t_2}}$ has an energy difference of $-2s_2 g_z^{(2)}-2t_2 g_x^{(2)}$ relative to $\ket{\phi_{s_2,t_2}}$.
    Thus in both of the above two equations, the final state is an eigenstate of $\Hpen$ with energy $E = -2 s_1 g_z^{(1)} - 2 t_1 g_x^{(1)} -2 s_2 g_z^{(2)} - 2 t_2 g_x^{(2)}$.
    If $E \ne 0$, then by definition $P_0 B \Hpen^{-1} A \ket{\psi} = P_0 A \Hpen^{-1} B \ket{\psi} = 0$.
    If $E = 0$, then $2 s_1 g_z^{(1)} + 2 t_2 g_x^{(2)} = - \left( 2 t_1 g_x^{(1)} + 2 s_2 g_z^{(2)} \right)$, hence $A \Hpen^{-1} B \ket{\psi} = - B \Hpen^{-1} A \ket{\psi}$.
    In both cases, we have $P_0 B \Hpen^{-1} A \ket{\psi} + P_0 A \Hpen^{-1} B \ket{\psi} = 0$.
    Since the above argument holds for arbitrary choice of $\ket{\psi}$, we arrive at $P_0 B \Hpen^{-1} A P_0 + P_0 A \Hpen^{-1} B P_0 = 0$.
\end{proof}

So far, our lemmas have not covered the scenario where the support of the operators $A$ and $B$ share a physical qubit in common, or involve two physical qubits from a single Bell pair with the action being Pauli Z simultaneously or Pauli X simultaneously.
In this scenario, $P_0 A \Hpen^{-1} B P_0 + P_0 B \Hpen^{-1} A P_0$ is in general nonzero unless we impose some extra constraints on the parameters of the blocks.
In particular, later in Section \ref{appxsec:scheme-2D-models} we will need the following lemmas.

\begin{lemma}\label{lem:gadget-coexist-3blocks-whatsoever}
    Consider three blocks $a$, $b$, $c$ with parameters $g_x^{(a)}, g_z^{(a)}, g_x^{(b)}, g_z^{(b)}, g_x^{(c)}, g_z^{(c)}$ satisfying
    \begin{equation}
        \sigma_1 g_{w_1}^{(a)} + \sigma_2 g_{w_2}^{(a)} + \sigma_3 g_{w_3}^{(b)} + \sigma_4 g_{w_4}^{(c)} \ne 0
    \end{equation}
    for all choices of $\sigma_1, \sigma_2, \sigma_3, \sigma_4 \in \{-1,1\}$ and $w_1, w_2, w_3, w_4 \in \{x,z\}$.
    Let $A := Z_i^{(a)} Z_j^{(b)}$, $X_i^{(a)} X_j^{(b)}$, $Z_i^{(a)} X_j^{(b)}$ or $X_i^{(a)} Z_j^{(b)}$.
    Let $B := Z_k^{(a)} Z_l^{(c)}$, $X_k^{(a)} X_l^{(c)}$, $Z_k^{(a)} X_l^{(c)}$ or $X_k^{(a)} Z_l^{(c)}$.
    Then $P_0 A \Hpen^{-1} B P_0 = 0$ and $P_0 B \Hpen^{-1} A P_0 = 0$.
\end{lemma}
\begin{proof}
    For concreteness, we prove the lemma for the case $A = Z_1^{(a)} X_1^{(b)}, B = Z_2^{(a)} X_2^{(c)}$.
    Since $Z_1^{(a)}$ anticommutes with $X_1^{(a)} X_2^{(a)}$ and commutes with other terms in $\Hpen$, while $X_1^{(b)}$ anticommutes with $Z_1^{(b)} Z_2^{(b)}$ and commutes with other terms in $\Hpen$, it follows that the action of $A$ on an arbitrary eigenbasis state of $\Hpen$ in Eq.~\eqref{eq:Hpen-eigenbasis} changes its energy by $\pm 2 g_x^{(a)} \pm 2 g_z^{(b)}$.
    Similarly, the action of $B$ on an arbitrary eigenbasis state of $\Hpen$ in Eq.~\eqref{eq:Hpen-eigenbasis} changes its energy by $\pm 2 g_x^{(a)} \pm 2 g_z^{(c)}$.
    The net effect of $B \Hpen^{-1} A$ or $A \Hpen^{-1} B$ is a change of energy by $\pm 2 g_x^{(a)} \pm 2 g_x^{(a)} \pm 2 g_z^{(b)} \pm 2 g_z^{(c)}$, which is nonzero by our assumption.
    Hence $P_0 B \Hpen^{-1} A P_0 = P_0 A \Hpen^{-1} B P_0 = 0$.
\end{proof}

\begin{lemma}\label{lem:gadget-coexist-3blocks-overlap-qubit}
    Consider three blocks $a$, $b$, $c$ with parameters $g_x^{(a)}, g_z^{(a)}, g_x^{(b)}, g_z^{(b)}, g_x^{(c)}, g_z^{(c)}$.
    Suppose
    \begin{equation}
        A := Z_i^{(a)} X_j^{(b)}, \quad B := Z_i^{(a)} X_k^{(c)}, \quad \left| g_z^{(b)} \right| \ne \left| g_z^{(c)} \right|,
    \end{equation}
    or
    \begin{equation}
        A := X_i^{(a)} Z_j^{(b)}, \quad B := X_i^{(a)} Z_k^{(c)}, \quad \left| g_x^{(b)} \right| \ne \left| g_x^{(c)} \right|,
    \end{equation}
    where $i,j,k \in \{1,2,3,4\}$.
    Then $P_0 A \Hpen^{-1} B P_0 = 0$ and $P_0 B \Hpen^{-1} A P_0 = 0$.
\end{lemma}
\begin{proof}
    Without loss of generality assume that $a,b,c$ are the first three blocks among the $n$ blocks.
    For concreteness, we prove the lemma for the case $A = Z_1^{(1)} X_1^{(2)}, B = Z_1^{(1)} X_1^{(3)}$ under the constraint $|g_z^{(2)}| \ne |g_z^{(3)}|$.
    Let $\ket{\psi}$ be a zero-energy eigenstate of $\Hpen$ in the form of Eq.~\eqref{eq:Hpen-eigenbasis}, i.e., $\ket{\psi} = \ket{\phi_{s_1,t_1}} \ket{\phi_{s'_1,t'_1}} \cdots \ket{\phi_{s_n,t_n}} \ket{\phi_{s'_n,t'_n}}$ with $(s_1+s'_1) g_z^{(1)} + (t_1+t'_1) g_x^{(1)} + \cdots + (s_n+s'_n) g_z^{(n)} + (t_n+t'_n) g_x^{(n)} = 0$.
    We have
    \begin{equation}
        \begin{aligned}
        B \Hpen^{-1} A \ket{\psi} &= B \Hpen^{-1} \left( (Z \otimes \id) \ket{\phi_{s_1,t_1}} \right) \ket{\phi_{s'_1,t'_1}} \left( (X \otimes \id) \ket{\phi_{s_2,t_2}} \right) \ket{\phi_{s'_2,t'_2}} \ket{\phi_{s_3,t_3}} \ket{\phi_{s'_3,t'_3}} \ket{\text{rest}} \\
        &\propto B \left( (Z \otimes \id) \ket{\phi_{s_1,t_1}} \right) \ket{\phi_{s'_1,t'_1}} \left( (X \otimes \id) \ket{\phi_{s_2,t_2}} \right) \ket{\phi_{s'_2,t'_2}} \ket{\phi_{s_3,t_3}} \ket{\phi_{s'_3,t'_3}} \ket{\text{rest}} \\
        &= \ket{\phi_{s_1,t_1}} \ket{\phi_{s'_1,t'_1}} \left( (X \otimes \id) \ket{\phi_{s_2,t_2}} \right) \ket{\phi_{s'_2,t'_2}} \left( (X \otimes \id) \ket{\phi_{s_3,t_3}} \right) \ket{\phi_{s'_3,t'_3}} \ket{\text{rest}}, \\
        &\propto \ket{\phi_{s_1,t_1}} \ket{\phi_{s'_1,t'_1}} \ket{\phi_{-s_2,t_2}} \ket{\phi_{s'_2,t'_2}} \ket{\phi_{-s_3,t_3}} \ket{\phi_{s'_3,t'_3}} \ket{\text{rest}},
        \end{aligned}
    \end{equation}
    where $\ket{\text{rest}}$ represents the state of the physical qubits from the 4th to the last block.
    So $B \Hpen^{-1} A \ket{\psi}$ is an eigenstate of $\Hpen$ with energy $-2 s_2 g_z^{(2)} - 2 s_3 g_z^{(3)}$, which is nonzero by our constraint.
    Therefore $P_0 B \Hpen^{-1} A \ket{\psi} = 0$.
    Since the above argument holds for arbitrary choice of $\ket{\psi}$, we arrive at $P_0 B \Hpen^{-1} A P_0 = 0$.
    Similarly, $P_0 A \Hpen^{-1} B P_0 = 0$.
\end{proof}

\section{Encoding schemes for 1D spin models}\label{appxsec:scheme-1D-models}
In this section, we provide the details of the encoding schemes we used in the main article for simulating the 1D transverse-field Ising (TFI) model and the 1D XY chain.

\subsection{1D TFI model}
\label{appxsec:scheme-1D-TFI}
For the 1D TFI model with $2n$ sites, the target Hamiltonian is given by
\begin{equation}\label{eq:Htar-1d-TFI}
    \Htar^{\text{1d-TFI}} = \sum_{k=1}^{2n-1} J_k \overline{Z}_k \overline{Z}_{k+1} + h_Z \sum_{k=1}^{2n} \overline{Z}_k + h_X \sum_{k=1}^{2n} \overline{X}_k,
\end{equation}
where $J_k, h_Z, h_X$ are real numbers.
Here we have assumed the open boundary condition and a uniform local field, but we note that both the periodic boundary condition and a non-uniform local field can be incorporated straightforwardly.
Moreover, our scheme can deal with both ferromagnetic ($J_k < 0$) and anti-ferromagnetic ($J_k > 0$) interactions.

First of all, we should choose the $(g_x, g_z)$ parameters of the $n$ code blocks and properly assign the $2n$ logical qubits to the $2n$ sites of the Ising chain.
For concreteness, we choose $g_x=1$ and $g_z=3$ for all the blocks, and we identify site $2b-1$ and site $2b$ as the first and the second logical qubit encoded in the $b$th block, respectively, for $b=1,2,\dots,n$.
For better presentation, we rewrite the target Hamiltonian per the notation we used in Section \ref{appxsec:gadgets}:
\begin{equation}\label{eq:Htar-1d-TFI-block-encoded}
    \Htar^{\text{1d-TFI}} = \underbrace{\sum_{b=1}^{n-1} J_{2b} \overline{Z}_2^{(b)} \overline{Z}_{1}^{(b+1)}}_{\text{cross-block terms}} + \underbrace{\sum_{b=1}^{n} J_{2b-1} \overline{Z}_1^{(b)} \overline{Z}_{2}^{(b)} + h_Z \sum_{b=1}^{n} \left( \overline{Z}_1^{(b)} + \overline{Z}_2^{(b)} \right) + h_X \sum_{b=1}^{n} \left( \overline{X}_1^{(b)} + \overline{X}_2^{(b)} \right)}_{\text{inner-block terms}},
\end{equation}
where $\overline{X}_{i}^{(b)}$ and $\overline{Z}_{i}^{(b)}$ denote the logical Pauli X and Pauli Z operators acting on the $i$th logical qubit in the $b$th block, respectively, for $i=1,2$ and $b=1,2,\dots,n$.

Next, we translate the inner-block logical terms directly into $\Hencfirst$ using the mapping described in the main article, which we copy below for reference:
\begin{subequations}
\begin{align}
& \overline{Z}_1^{(b)} = Z_{1}^{(b)} Z_{2}^{(b)}, \quad
\overline{Z}_2^{(b)} = Z_{1}^{(b)} Z_{3}^{(b)}, \quad \overline{X}_1^{(b)} = -X_{1}^{(b)} X_{3}^{(b)}, \quad \overline{X}_2^{(b)} = X_{1}^{(b)} X_{2}^{(b)}, \\
& \overline{Z}_1^{(b)} \overline{Z}_2^{(b)} = Z_{2}^{(b)} Z_{3}^{(b)}, \quad 
\overline{X}_1^{(b)} \overline{X}_2^{(b)} = - X_{2}^{(b)} X_{3}^{(b)}. 
\end{align}
\end{subequations}

Finally, for each cross-block logical term $J_{2b} \overline{Z}_2^{(b)} \overline{Z}_{1}^{(b+1)}$, we add the following perturbative gadget to $\Hencsecond$:
\begin{subequations}
\begin{align}
& G^{(b,b+1)} := \alpha Z_2^{(b)} X_3^{(b+1)} + \beta Z_4^{(b)} X_3^{(b+1)}, \\
& \text{where} \quad \alpha := \sqrt{\frac{8}{3} |J_{2b}|}, \quad \beta := \sgn(J_{2b}) \sqrt{\frac{8}{3} |J_{2b}|}. 
\end{align}
\end{subequations}
That is,
\begin{equation}
    \Hencsecond = \sum_{b=1}^{n-1} G^{(b,b+1)}.
\end{equation}
To see that this gives rise to the desired interaction, note that $g_z \ne g_x$ implies that $P_0 \Hencsecond P_0 = 0$, and
\begin{align}
    - P_0 \Hencsecond \Hpen^{-1} \Hencsecond P_0 = & - \sum_{b=1}^{n-1} P_0 G^{(b,b+1)} \Hpen^{-1} G^{(b,b+1)} P_0 
    \nonumber \\
    & - \sum_{a<b} \left( P_0 G^{(a,a+1)} \Hpen^{-1} G^{(b,b+1)} P_0 + P_0 G^{(b,b+1)} \Hpen^{-1} G^{(a,a+1)} P_0 \right).
\end{align}
All the two-gadget interference terms above vanish for the following reasons:
\begin{itemize}
    \item When $b - a \ge 2$, the gadgets $G^{(a,a+1)}$ and $G^{(b,b+1)}$ act on four different blocks.
    So by Lemma \ref{lem:gadget-coexist-4blocks} they do not interfere with each other.
    \item When $b=a+1$, we further expand
    \begin{equation}
        \begin{aligned}
        &~~~ P_0 G^{(a,a+1)} \Hpen^{-1} G^{(a+1,a+2)} P_0 + P_0 G^{(a+1,a+2)} \Hpen^{-1} G^{(a,a+1)} P_0 \\
        &= \alpha^2 \left( P_0 Z_2^{(a)} X_3^{(a+1)} \Hpen^{-1} Z_2^{(a+1)} X_3^{(a+2)} P_0 + P_0 Z_2^{(a+1)} X_3^{(a+2)} \Hpen^{-1} Z_2^{(a)} X_3^{(a+1)} P_0 \right) \\
        &+ \beta^2 \left( P_0 Z_4^{(a)} X_3^{(a+1)} \Hpen^{-1} Z_4^{(a+1)} X_3^{(a+2)} P_0 + P_0 Z_4^{(a+1)} X_3^{(a+2)} \Hpen^{-1} Z_4^{(a)} X_3^{(a+1)} P_0 \right) \\
        &+ \alpha\beta \left( P_0 Z_2^{(a)} X_3^{(a+1)} \Hpen^{-1} Z_4^{(a+1)} X_3^{(a+2)} P_0 + P_0 Z_4^{(a+1)} X_3^{(a+2)} \Hpen^{-1} Z_2^{(a)} X_3^{(a+1)} P_0 \right) \\
        &+ \alpha\beta \left( P_0 Z_4^{(a)} X_3^{(a+1)} \Hpen^{-1} Z_2^{(a+1)} X_3^{(a+2)} P_0 + P_0 Z_2^{(a+1)} X_3^{(a+2)} \Hpen^{-1} Z_4^{(a)} X_3^{(a+1)} P_0 \right). \\
        \end{aligned}
    \end{equation}
    By Lemma \ref{lem:gadget-coexist-3blocks}, the first and the fourth line on the right-hand side is zero; while by Lemma \ref{lem:gadget-coexist-3blocks-ZX}, the second and the third line on the right-hand side is zero.
\end{itemize}
Thus, the logical interactions generated by the perturbative gadgets combine additively, and according to the first row of Table \ref{tab:gadget-lookup}, we obtain that
\begin{equation}
- \Penc \Hencsecond \Hpen^{-1} \Hencsecond \Penc = \sum_{b=1}^{n-1} J_{2b} \overline{Z}_2^{(b)} \overline{Z}_1^{(b+1)} - \sum_{b=1}^{n-1} |J_{2b}| \overline{Z}_1^{(b+1)}.
\end{equation}
Note that besides the desired cross-block Ising interactions, $\Hencsecond$ also induces extra single-qubit logical operations.
Fortunately, these can be compensated for by adding correction terms to $\Hencfirst$.
We arrive at
\begin{align}
    \Hencfirst = & \sum_{b=1}^{n} J_{2b-1} Z_{2}^{(b)} Z_{3}^{(b)} + h_Z \sum_{b=1}^{n} \left( Z_{1}^{(b)} Z_{2}^{(b)} + Z_{1}^{(b)} Z_{3}^{(b)} \right) 
    \nonumber \\
    &+ h_X \sum_{b=1}^{n} \left( X_{1}^{(b)} X_{2}^{(b)} - X_{1}^{(b)} X_{3}^{(b)} \right) + \sum_{b=1}^{n-1} |J_{2b}| Z_{1}^{(b+1)} Z_{2}^{(b+1)}.
\end{align}

\subsection{1D XY chain}
The target Hamiltonian of a $2n$-site XY chain is given by
\begin{equation}\label{eq:Htar-1d-XY}
    \Htar^{\text{1d-XY}} = \sum_{k=1}^{2n-1} J_k \left( \overline{Z}_k \overline{Z}_{k+1} + \overline{X}_k \overline{X}_{k+1} \right),
\end{equation}
where $J_k$'s are real numbers.
We have assumed the open boundary condition and isotropic coupling coefficients for simplicity, but our encoding scheme does not have such constraints.
We choose $g_x=1$ and $g_z=3$ for all the blocks and assign the logical qubits in the same way as we did for the 1D TFI model.
The only difference from the 1D TFI model is the choice of perturbative gadgets; here we apply the third row of Table \ref{tab:gadget-lookup}, and the result is
\begin{subequations}
\begin{align}
\Hencfirst &= \sum_{b=1}^{n} J_{2b-1} \left( Z_{2}^{(b)} Z_{3}^{(b)} - X_{2}^{(b)} X_{3}^{(b)} \right) + \sum_{b=1}^{n-1} \frac{3}{2} |J_{2b}| X_{1}^{(b)} X_{2}^{(b)} - \sum_{b=1}^{n-1} \frac{7}{2} |J_{2b}| Z_{1}^{(b+1)} Z_{2}^{(b+1)}, \\
\Hencsecond &= \sum_{b=1}^{n-1} G^{(b,b+1)},
\end{align}
\end{subequations}
with
\begin{subequations}
\begin{align}
& G^{(b,b+1)} := \alpha Z_2^{(b)} X_3^{(b+1)} + \beta Z_4^{(b)} X_3^{(b+1)} + \gamma Z_2^{(b)} X_1^{(b+1)}, \\
& \text{where} \quad \alpha := \sqrt{\frac{8}{3} |J_{2b}|}, \quad \beta := \sgn(J_{2b}) \sqrt{\frac{8}{3} |J_{2b}|}, \quad \gamma := \sgn(J_{2b}) 3 \sqrt{\frac{8}{3} |J_{2b}|}. 
\end{align}
\end{subequations}
The reason for all the two-gadget interference terms being zero is similar to that for the 1D TFI model.

\section{Encoding schemes for 2D spin models and beyond}\label{appxsec:scheme-2D-models}
In Section \ref{appxsec:scheme-1D-models} we chose $g_x$ and $g_z$ to be uniform among all blocks.
For 1D TFI and XY models, we have shown that this choice does not incur two-gadget interference.
This is because we only made use of Lemmas \ref{lem:gadget-coexist-4blocks}, \ref{lem:gadget-coexist-3blocks} and \ref{lem:gadget-coexist-3blocks-ZX}, none of which requires any constraint on the parameters of the blocks.
Unfortunately, when the interaction graph gets more complicated, these lemmas are no longer applicable to every pair of perturbative gadgets.
We solve this problem by selecting the block parameters \textit{non-uniformly} subject to the constraints made in Lemma \ref{lem:gadget-coexist-3blocks-whatsoever} and \ref{lem:gadget-coexist-3blocks-overlap-qubit}.
We first present the solution for a general interaction graph, and then focus on the 2D TFI model and 2D compass model.

\subsection{General interaction graph}
Suppose we have already assigned the logical qubits of a 2-local target Hamiltonian $\Htar$ to $n$ blocks, but have not determined the block parameters $g_x^{(1)}, g_z^{(1)}, \dots, g_x^{(n)}, g_z^{(n)}$.
We assume that for each pair of blocks $a$ and $b$, their cross-block logical interaction can be generated by a perturbative gadget $G^{(a,b)}$ (modulo some residual inner-block terms) that can be written as a weighted sum of terms of the form $X_i^{(a)} X_j^{(b)}$, $Z_i^{(a)} Z_j^{(b)}$, $X_i^{(a)} Z_j^{(b)}$ or $Z_i^{(a)} X_j^{(b)}$ where $i,j \in \{1,2,3,4\}$.
Now, the goal is to design the block parameters such that different gadgets do not interfere with each other.

Let us define the block-level interaction graph $\mathcal{G}$ to be an $n$-vertex graph, where vertex $a$ and vertex $b$ are connected, denoted $a \stackrel{\mathcal{G}}{\sim} b$, if and only if there is a nonzero cross-block logical interaction between block $a$ and block $b$.
We further define graph $\mathcal{G}'$ over the same vertex set, where vertex $a$ and vertex $b$ are connected, denoted $a \stackrel{\mathcal{G}'}{\sim} b$, if and only if $a \stackrel{\mathcal{G}}{\sim} b$ or there exists another vertex $c$ such that both $a \stackrel{\mathcal{G}}{\sim} c$ and $b \stackrel{\mathcal{G}}{\sim} c$.
Now, we choose a coloring of the vertices of $\mathcal{G}'$ such that no two adjacent vertices have the same color.
Suppose a total number of $N_{\mathrm{color}}$ colors are used.
Then, for each block $b$ of color $C_b \in \{0,1,2,\dots,N_{\mathrm{color}}-1\}$, we set
\begin{equation}
    g_x^{(b)} := 2^{2 C_b}, \quad g_z^{(b)} := 2^{2 C_b + 1}.
\end{equation}
We claim that with this choice of block parameters, all the perturbative gadgets collectively work as desired.
Firstly, for each nonzero gadget $G^{(a,b)}$, by definition $a \stackrel{\mathcal{G}}{\sim} b$, and hence $C_a \ne C_b$.
It follows that the four numbers $|g_x^{(a)}|, |g_z^{(a)}|, |g_x^{(b)}|, |g_z^{(b)}|$ are distinct from each other, which implies that $P_0 G^{(a,b)} P_0 = 0$.
Secondly, any two gadgets $G^{(a,b)}$ and $G^{(c,d)}$ that act on four different blocks do not interfere according to Lemma \ref{lem:gadget-coexist-4blocks}, i.e., $P_0 G^{(a,b)} \Hpen^{-1} G^{(c,d)} P_0 + P_0 G^{(c,d)} \Hpen^{-1} G^{(a,b)} P_0 = 0$.
Lastly, consider two nonzero gadgets $G^{(a,b)}$ and $G^{(a,c)}$ that shares a common block $a$.
By definition, $a$, $b$ and $c$ are adjacent to each other in $\mathcal{G}'$ and hence colored differently.
It is then easy to see that $\sigma_1 g_{w_1}^{(a)} + \sigma_2 g_{w_2}^{(a)} + \sigma_3 g_{w_3}^{(b)} + \sigma_4 g_{w_4}^{(c)} \ne 0$ for all choices of $\sigma_1, \sigma_2, \sigma_3, \sigma_4 \in \{-1,1\}$ and $w_1, w_2, w_3, w_4 \in \{x,z\}$.
By Lemma \ref{lem:gadget-coexist-3blocks-whatsoever}, we deduce that $P_0 G^{(a,b)} \Hpen^{-1} G^{(a,c)} P_0 = P_0 G^{(a,c)} \Hpen^{-1} G^{(a,b)} P_0 = 0$.

We remark that if $\mathcal{G}$ has a finite degree (i.e., every vertex is adjacent to only finitely many other vertices), then $\mathcal{G}'$ also has a finite degree and hence can be colored using a finite number of different colors, independent of the number of blocks $n$.
We also note that the strategy above based on graph coloring is just one out of many possibilities that guarantee the correct functioning of the gadgets.
Depending on the topology of the interaction graph and the chosen gadgets, one may find other strategies with a smaller number of different block parameters.
Indeed, the 2D TFI model and the 2D compass model are two such examples that we address below.

\subsection{2D TFI model}
Consider a 2D TFI model on a $2n \times 2n$ square lattice:
\begin{equation}\label{eq:Htar-2dTFI}
    \Htar^{\text{2d-TFI}} = J_1 \sum_{i=1}^{2n-1} \sum_{j=1}^{2n} \overline{Z}_{i,j} \overline{Z}_{i+1,j} + J_2 \sum_{i=1}^{2n} \sum_{j=1}^{2n-1} \overline{Z}_{i,j} \overline{Z}_{i,j+1} + h_Z \sum_{i,j=1}^{2n} \overline{Z}_{i,j} + h_X \sum_{i,j=1}^{2n} \overline{X}_{i,j}.
\end{equation}
We identify every two consecutive lattice sites in the horizontal direction as the two logical qubits encoded in a code block, as illustrated in Fig.~\ref{fig:2d-model-block-layout}(a) for a $6 \times 6$ lattice.
Note that we have reversed the positional ordering of the two logical qubits inside a block for every other row, such that all the cross-block Ising interactions take the form $\overline{Z}_2^{(a)} \overline{Z}_1^{(b)}$ and therefore can be encoded using perturbative gadgets of the form $\alpha Z_2^{(a)} X_3^{(b)} + \beta Z_4^{(a)} X_3^{(b)}$ (c.f. Table \ref{tab:gadget-lookup}).
Also, note that every row of the lattice is in the same configuration as the encoding for the 1D TFIM model.
Motivated by this observation, we can safely choose the block parameters to be identical for all the blocks in the same horizontal chain, as we did in Section \ref{appxsec:scheme-1D-models}.
It remains to make sure that the incorporation of the Ising interactions in the vertical direction does not incur unwanted two-gadget interference.
One can check that the only two-gadget interference terms that are not covered by Lemmas \ref{lem:gadget-coexist-4blocks}, \ref{lem:gadget-coexist-3blocks}, \ref{lem:gadget-coexist-3blocks-ZX}, \ref{lem:gadget-coexist-2blocks}, \ref{lem:gadget-coexist-2blocks-ZX} are the situation described in Lemma \ref{lem:gadget-coexist-3blocks-overlap-qubit}.
Therefore, all we need to do is choose the code parameters for each row to comply with the constraints in Lemma \ref{lem:gadget-coexist-3blocks-overlap-qubit}.
This amounts to requiring that the absolute value of the code parameters be different for every three consecutive rows.
For example, in our numerical simulation, we have chosen $g_x = 1, g_z = 4$ for all the blocks in the $(3k-2)$-th row, $g_x = 2, g_z = 5$ for all the blocks in the $(3k-1)$-th row, and $g_x = 3, g_z = 6$ for all the blocks in the $3k$-th row, where $k=1,2,\dots$.

\begin{figure}[h!]
\includegraphics[width=0.7\linewidth]{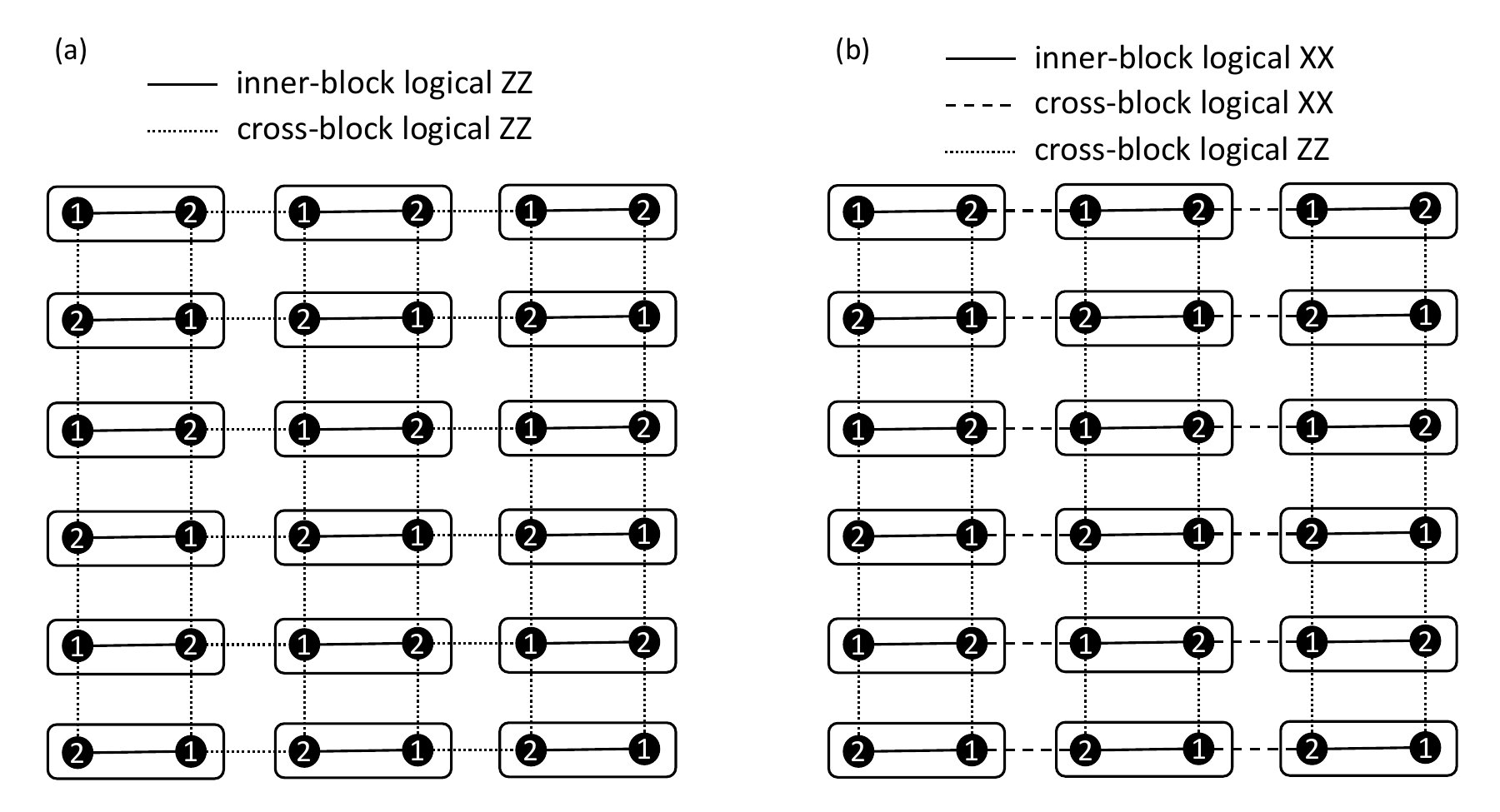}
\centering
\caption{Layout of the code blocks (shown as rectangles) used to simulate (a) the 2D TFI model and (b) the 2D compass model on a $6 \times 6$ square lattice.
Also included in the figure are the logical qubits occupying the lattice sites (shown as solid circles labeled by their logical indices within each block) and the 2-site interaction terms (shown as solid, dashed, or dotted lines).}
\label{fig:2d-model-block-layout}
\end{figure}

\subsection{2D compass model}
The 2D compass model on a $2n \times 2n$ square lattice is given by \cite{dorier2005quantum-compass-model}:
\begin{equation}\label{eq:Htar-2dCompass}
    \Htar^{\text{2d-Compass}} = J_Z \sum_{i=1}^{2n-1} \sum_{j=1}^{2n} \overline{Z}_{i,j} \overline{Z}_{i+1,j} + J_X \sum_{i=1}^{2n} \sum_{j=1}^{2n-1} \overline{X}_{i,j} \overline{X}_{i,j+1}.
\end{equation}
We can use the same layout of the logical qubits (see Fig.~\ref{fig:2d-model-block-layout}(b)) and the same block parameters as before.
The only difference is the choice of the perturbative gadget to realize the cross-block logical XX interactions---now we need to apply the second row of Table \ref{tab:gadget-lookup}.

\section{Numerical simulation}\label{appxsec:numerics}
In order to verify the performance of our error suppression scheme, we perform extensive numerical simulations of the encoded dynamics under the spin models mentioned in Appendices~\ref{appxsec:scheme-1D-models} and~\ref{appxsec:scheme-2D-models}.
In FIG.~\ref{fig:sweep_lamb}, we plot the scaling of the infidelity of the encoded simulation with respect to the penalty coefficient at a fixed time ($t=1$); the generic behavior does not depend on the choice of the evolution time.
More concretely, each run of the scheme consists of time-evolving the Hamiltonian $\Hsim + V$ for time $t$ starting from a Haar-random initial state inside $\Senc$.
We take the noise term as $V = \sum_{\ell} (\epsilon^{X}_{\ell} X_{\ell} + \epsilon^{Y}_{\ell} Y_{\ell} + \epsilon^{Z}_{\ell} Z_{\ell})$, with $\epsilon ^{\alpha}_{\ell}$ for $\alpha = X,Y,Z$ denoting random variables independently sampled from the uniform distribution over $[-0.1,0.1]$.
We plot the average infidelities between the encoded dynamics and the target evolution in FIG.~\ref{fig:sweep_lamb}; as shown in this figure, the infidelity of the noisy, encoded evolution scales inversely with the penalty coefficient $\lambda$, which agrees excellently with the analytical error bound in Theorem~\ref{thm:main-theorem}.

\begin{figure}[h]
\includegraphics[width=0.8\linewidth]{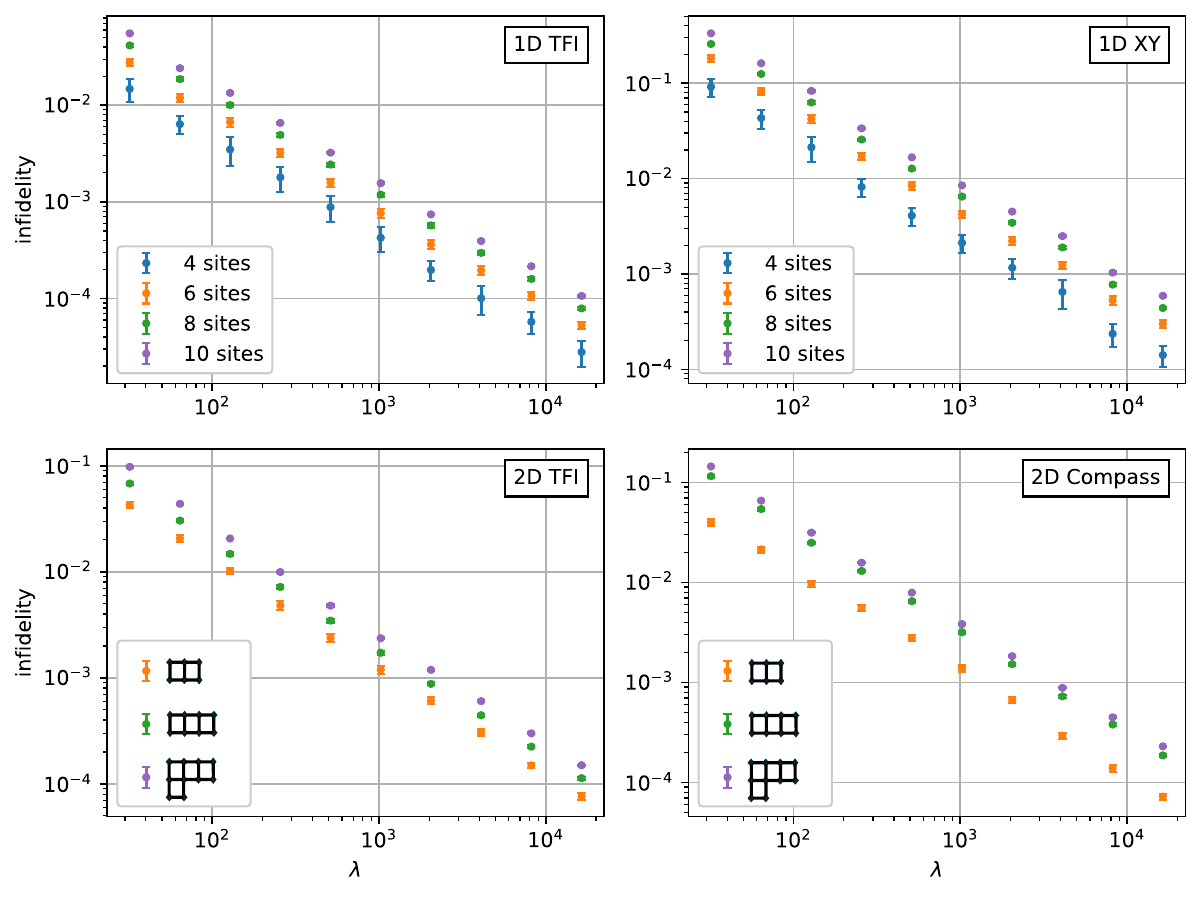}
\centering
\caption{Numerically simulated performance of the encoded many-body quantum spin models on a 1D chain or a 2D lattice, as quantified by the average infidelity over a randomly selected initial state within the encoding subspace $\Senc$.
All Hamiltonian coefficients are chosen to be $1$ (see Eqs.~\eqref{eq:Htar-1d-TFI}, \eqref{eq:Htar-1d-XY}, \eqref{eq:Htar-2dTFI}, \eqref{eq:Htar-2dCompass}), and the penalty coefficient is chosen from the range $\lambda \in [2^5, 2^{14}]$.
Lattice shapes and dimensions for the 2D models are shown in the legend of the lower two plots.
Each data point is obtained by averaging over 20 random samples of initial states for systems with sizes smaller than $10$ (or 5 samples for the ones with 10 sites).
The numerically computed infidelity exhibits a $\propto \lambda ^{-1}$ scaling, in agreement with the prediction from Theorem \ref{thm:main-theorem}.
}
\label{fig:sweep_lamb}
\end{figure}

\end{document}